\documentclass[11pt]{article}%
\usepackage{amsfonts}
\usepackage{amsmath}
\usepackage{fullpage}
\usepackage{amssymb}
\usepackage{graphicx}
\usepackage{hyperref}%
\providecommand{\U}[1]{\protect\rule{.1in}{.1in}}
\newtheorem{theorem}{Theorem}

\newtheorem{corollary}[theorem]{Corollary}

\newtheorem{definition}[theorem]{Definition}

\newtheorem{lemma}[theorem]{Lemma}

\newtheorem{proposition}[theorem]{Proposition}
\newtheorem{remark}[theorem]{Remark}

\newenvironment{proof}[1][Proof]{\noindent\textbf{#1.} }{\ \rule{0.5em}{0.5em}}
\numberwithin{equation}{section}
\begin{document}

\title{\textbf{Converse bounds for private communication over quantum channels}}
\author{Mark M. Wilde\thanks{ Hearne Institute for Theoretical Physics, Department of
Physics and Astronomy, Center for Computation and Technology, Louisiana State
University, Baton Rouge, Louisiana 70803, USA}
\and Marco Tomamichel\thanks{School of Physics, The University of Sydney, Sydney,
Australia}
\and Mario Berta\thanks{Institute for Quantum Information and Matter, California
Institute of Technology, Pasadena, California 91125, USA} }
\maketitle

\begin{abstract}
This paper establishes several converse bounds on the private transmission
capabilities of a quantum channel. The main conceptual development builds
firmly on the notion of a private state, which is a powerful, uniquely quantum
method for simplifying the tripartite picture of privacy involving local
operations and public classical communication to a bipartite picture of
quantum privacy involving local operations and classical communication. This
approach has previously led to some of the strongest upper bounds on secret
key rates, including the squashed entanglement and the relative entropy of
entanglement. Here we use this approach along with a \textquotedblleft privacy
test\textquotedblright\ to establish a general meta-converse bound for private
communication, which has a number of applications. The meta-converse allows
for proving that any quantum channel's relative entropy of entanglement is a
strong converse rate for private communication. For covariant channels, the
meta-converse also leads to second-order expansions of relative entropy of
entanglement bounds for private communication rates. For such channels, the
bounds also apply to the private communication setting in which the sender and
receiver are assisted by unlimited public classical communication, and as
such, they are relevant for establishing various converse bounds for quantum
key distribution protocols conducted over these channels. We find precise
characterizations for several channels of interest and apply the methods to
establish converse bounds on the private transmission capabilities of all
phase-insensitive bosonic channels.

\end{abstract}

\section{Introduction}

Ever since the discovery of quantum key distribution \cite{bb84}, researchers
have been interested in exploiting quantum-mechanical effects in order to
ensure the secrecy of communication. This has led to a large amount of
research in many directions \cite{SBCDLP09}, both experimental and
theoretical, and one of the recent challenges has been to connect both of
these directions.

On the theoretical side, much progress has been made by generalizing several
ideas developed in the context of classical information theory. For example,
the wiretap channel is a simple model for private communication, and one can
study its capacity for secure data transmission \cite{W75} (see
\cite{CK78,TB14,HTW14,Win16} for later progress on refining this capacity). In
this model, two honest parties, usually called Alice (the sender)\ and Bob
(the receiver), are connected by a classical channel. At the same time, there
is a classical channel connecting Alice to an eavesdropper or wiretapper,
usually called Eve. The goal is to devise a communication scheme such that
Alice can communicate to Bob with small error in such a way that Eve gets
nearly zero information about the message communicated (with both the
probability of error and information leakage vanishing in the limit of many
channel uses). One can further generalize the model to allow for public
classical communication and study capacities in this context \cite{M93,AC93}.
However, two major drawbacks of the wiretap model is that the honest parties
need to assume that they have fully characterized both 1)\ their channel and
2)\ the channel to the eavesdropper, which may not be possible in practice.
Nevertheless, techniques developed in the context of the wiretap channel have
been foundational to our understanding of information-theoretically secure communication.

Quantum mechanics offers a route around one of the aforementioned problems
with the classical model, via the notion of purification. Indeed, for any
quantum channel connecting Alice to Bob, there is a purification (or isometric
extension) of this channel that is unique up to unitary rotations~\cite{S55}.
All the degrees of freedom that are not accessible to the receiver Bob are
accessible to the environment of the channel, and in the spirit of being
cautious, as is usually the case in cryptography, we assume that the
eavesdropper has full access to the environmental system. For example,
communication from Alice to Bob in free space can be modeled by an interaction
at a beamsplitter \cite{S09}, and in the wiretap model, we assume that all of
the light that is lost along the way can be collected by the eavesdropper Eve
\cite{GSE08}. Thus, in the quantum wiretap model, Alice and Bob can perform
parameter estimation in order to characterize their channel, and once they
have a complete characterization, they also have a model for the channel to
the eavesdropper, circumventing one of the aforementioned problems with the
classical model. If we allow for Alice and Bob to make use of public classical
communication in addition to the quantum channel (see, e.g.,
\cite{TGW14IEEE,TGW14Nat}), then this model is closely related to that which
is used in some quantum key distribution protocols. In practice, one drawback
of this model is that the channel from Alice to Bob might be changing with
time or difficult to characterize, but nevertheless one can study the private
capacities of this quantum wiretap channel model in an attempt to gain some
understanding of what rates might be achievable in principle.

With this Shannon-theoretic viewpoint, the quantum wiretap model has been
studied in much detail. The private capacity of a quantum wiretap channel was
defined and characterized in \cite{ieee2005dev,1050633}. For the class of
degradable quantum channels, there is a tractable formula for the private
capacity \cite{S08}. The same occurs for conjugate degradable \cite{BDHM10},
less noisy, and more capable channels \cite{Wat12}. Beyond such channels,
little is known and recent evidence suggests that characterizing private
capacity effectively could be a very difficult challenge. For example, the
formula for private information from \cite{ieee2005dev,1050633}\ is now known
to be superadditive in general \cite{smith:170502,PhysRevLett.115.040501}, and
the private capacity itself is as well \cite{LWZG09}.

More recently there has been progress on characterizing the private capacity
when public classical communication is available (for a given channel
$\mathcal{N}$, let $P^{\leftrightarrow}(\mathcal{N})$ denote this quantity).
Building on the notion of squashed entanglement \cite{CW04}\ and the fact that
this quantity is an upper bound on distillable key \cite{C06}, the authors of
\cite{TGW14IEEE}\ defined the squashed entanglement of a channel and showed
that it is an upper bound on $P^{\leftrightarrow}(\mathcal{N})$ for any
channel $\mathcal{N}$ (see also \cite{Wilde2016}). This result thus
established a strong limitation for quantum key distribution protocols as
discussed in \cite{TGW14Nat}. Following this development, by building on the
notion of relative entropy of entanglement \cite{VP98}\ and the fact that this
quantity is also an upper bound on the distillable key of a bipartite state
\cite{HHHO05,HHHO09}, the authors of \cite{PLOB15}\ defined a channel's
relative entropy of entanglement and stated that it is an upper bound on
$P^{\leftrightarrow}(\mathcal{N})$ for any channel $\mathcal{N}$ that has a
\textquotedblleft teleportation symmetry\textquotedblright\ identified in
\cite[Section~V]{BDSW96} and extended in \cite{NFC09,PLOB15}.
It is an open question to determine whether the relative entropy of
entanglement is an upper bound on the two-way assisted private capacity of a
general quantum channel.

Both of the aforementioned upper bounds on $P^{\leftrightarrow}(\mathcal{N})$
critically rely upon the notion of a private state \cite{HHHO05,HHHO09}. To
motivate this notion, consider that the ultimate goal of a $P^{\leftrightarrow
}$\ protocol is to generate a secret-key state of the following form:%
\begin{equation}
\left(  \mathcal{M}_{A}\otimes\mathcal{M}_{B}\right)  \left(  \gamma
_{ABE}\right)  =\frac{1}{K}\sum_{i=0}^{K-1}|i\rangle\langle i|_{A}%
\otimes|i\rangle\langle i|_{B}\otimes\sigma_{E},
\label{eq:secret-key-state-1-tri}%
\end{equation}
where the $A$ system is possessed by Alice, $B$ by Bob, $E$ by the
eavesdropper, $K$ is the number of key values, $\gamma_{ABE}$ is some state on
systems $ABE$, $\mathcal{M}(\cdot)=\sum_{i}|i\rangle\langle i|(\cdot
)|i\rangle\langle i|$ is a projective measurement channel with $\{|i\rangle
\}_{i}$ an orthonormal basis, and $\sigma_{E}$ is some state on system $E$.
The state in \eqref{eq:secret-key-state-1-tri}\ is such that the systems $A$
and $B$ are perfectly correlated (i.e., maximally classically correlated), and
the value of the key is uniformly random and independent of Eve's system $E$.
The main observation of \cite{HHHO05,HHHO09} is that, in principle, every step
of a $P^{\leftrightarrow}$\ protocol can be purified, and since these steps
are conducted in the laboratories of Alice and Bob, these parties could
possess purifying systems of $\gamma_{ABE}$ (call them $A^{\prime}$ and
$B^{\prime}$), such that $\gamma_{ABA^{\prime}B^{\prime}E}$ is a pure state
satisfying $\operatorname{Tr}_{A^{\prime}B^{\prime}}\{\gamma_{ABA^{\prime
}B^{\prime}E}\}=\gamma_{ABE}$. By employing purification theorems of quantum
information theory, the authors of \cite{HHHO05,HHHO09} showed that the
reduced state of $\gamma_{ABA^{\prime}B^{\prime}E}$ on the systems
$ABA^{\prime}B^{\prime}$ has the following form:%
\begin{equation}
\gamma_{ABA^{\prime}B^{\prime}}=U_{ABA^{\prime}B^{\prime}}(\Phi_{AB}%
\otimes\theta_{A^{\prime}B^{\prime}})U_{ABA^{\prime}B^{\prime}}^{\dag},
\label{eq:intro-private-state}%
\end{equation}
where $\Phi_{AB}$ is a maximally entangled state, $U_{ABA^{\prime}B^{\prime}}$
is a special kind of unitary called a \textquotedblleft
twisting,\textquotedblright\ and $\theta_{A^{\prime}B^{\prime}}$ is an
arbitrary state (see Section~\ref{sec:prelim}\ for more details). Such a state
is now known as a bipartite private state and is fully equivalent to the state
in \eqref{eq:secret-key-state-1-tri}\ in the aforementioned sense. This
observation thus allows for a perspective change which is helpful for
analyzing private communication protocols:\ one can eliminate the eavesdropper
from the analysis, revising the goal of such a protocol to be the production
of states of the form in \eqref{eq:intro-private-state}, and this allows for
using the powerful tools of entanglement theory \cite{H42007}\ to analyze
secret-key rates.

Not only did the results of \cite{HHHO05,HHHO09} provide a conceptually
different method for understanding privacy in the quantum setup, but they also
showed how there are fundamental differences between entanglement distillation
and secret-key distillation protocols. Indeed, the strongest demonstration of
this difference was the realization that there exist quantum channels that
have zero capacity to send quantum information and yet can generate private
information at a non-zero rate \cite{HHHLO08,PhysRevLett.100.110502}. This in
turn led to the discovery of the superactivation effect
\cite{science2008smith,SSY11}:\ two quantum channels each having zero quantum
capacity can be used together to have a non-zero quantum capacity, by taking
advantage of the intricate interplay between privacy and coherence.

In all of the above theoretical analyses, the statements made are asymptotic
in nature, applying exclusively to the situation in which a large number of
independent and identical channel uses are available. While these works have
provided interesting bounds and are conceptually rich, they are somewhat
removed from practical situations in which the number of channel uses is
limited. However, some recent works have aimed to bridge this gap for the case
of quantum communication \cite{BD10,MW13,TWW14,BDL15,TBR15}, giving more
refined bounds on what is possible and impossible for a limited number of
channel uses. One goal of the present paper is to bridge the gap for private communication.

Similar to the results from \cite{TGW14Nat,PLOB15,Goodenough2015,Wilde2016},
the bounds given in this paper can be used to assess the performance of
quantum key distribution protocols, as first suggested in \cite{L15}. In
particular, one prominent experimental goal has been to build a quantum
repeater \cite{SST11,LST09}, which is a device that could be inserted between
two nodes in a given network to increase the rates of secret key generation.
One way to assess the performance of such a repeater is that it should be able
to exceed the limitations of the network that hold whenever the repeater is
not present \cite{L15}, and this has been hailed as one of the main
applications of the bounds from
\cite{TGW14Nat,PLOB15,Goodenough2015,Wilde2016}. However, since these bounds
are asymptotic in nature, they have limited applicability to protocols using a
channel a finite number of times. On the other hand, the bounds given in this
paper can be used to assess the performance of practical, non-asymptotic
protocols for certain channels.

\paragraph{Summary of results.}

In this paper, we establish several converse bounds on the private
transmission capabilities of a quantum channel. The main conceptual
development is a so-called \textquotedblleft meta-converse\textquotedblright%
\ bound for private communication, which is a general upper bound that can be
translated to several regimes of interest (the idea of a \textquotedblleft
meta-converse\textquotedblright\ has its roots in the seminal work in
\cite{polyanskiy10}). In particular, we can use the meta-converse to establish
that a channel's relative entropy of entanglement is a strong converse rate
for private communication, meaning that if the communication rate of a
sequence of protocols exceeds this amount, then the probability of a
protocol's failure tends to one exponentially fast in the number of channel
uses. The result builds strongly on the approach from \cite{TWW14} (see also
\cite{MW13} for progress on a strong converse for private capacity of
degradable channels).

We also use the meta-converse bound to establish second-order converse bounds
for private communication. In this regime, one fixes the error parameter and
asks what is the maximum rate of private communication possible. Here we again
find an upper bound in terms of quantities related to the relative entropy of
entanglement, but this bound applies only to channels with certain symmetry.
For some channels with sufficient symmetry, we establish exact
characterizations of the second-order coding rate (and even finer)\ by
combining our upper bounds with the lower bounds from \cite{TBR15}. Finally,
we can use the method to establish finite blocklength converse bounds for all
single-mode phase-insensitive bosonic channels, and as a consequence, we find
that the weak-converse bounds presented in \cite{PLOB15} are in fact
strong-converse bounds for two-way assisted private communication. As a special case, we establish that
the two-way assisted, unconstrained private and quantum capacities of the
pure-loss and quantum-limited amplifier channels satisfy the strong converse property.

The rest of the paper gives details of our results. In the next section
(Section~\ref{sec:prelim}), we recall many facts from quantum information
theory which are needed for the developments, and we establish the notation
used in the later parts. Section~\ref{sec:private-states} reviews private
states \cite{HHHO05,HHHO09}, and then Section~\ref{sec:priv-class-comm}\ gives
definitions of secret-key transmission protocols and their non-asymptotic
achievable rates. Section~\ref{sec:meta-converse} establishes the general
meta-converse bound for any private communication protocol. In
Section~\ref{sec:strong-converse}, we use the meta-converse and prior
developments in \cite{TWW14}\ to prove that a channel's relative entropy of
entanglement is a strong converse rate for private communication. If a channel
is \textquotedblleft teleportation simulable\textquotedblright\ (defined
later),\ then the same quantity is a strong converse rate for private
communication assisted by public classical communication. In
Section~\ref{sec:second-order}, we use the meta-converse to establish
second-order expansions of relative entropy of entanglement bounds on private
communication rates (this is for channels with sufficient symmetry).
Section~\ref{sec:examples} then gives several examples of channels for which
we have precise characterizations of their private transmission capabilities,
including the qubit dephasing channel, the qubit erasure channel, and any
entanglement-breaking channel. Section~\ref{sec:gaussian} establishes finite
blocklength converse bounds for phase-insensitive bosonic channels. We finally
conclude in Section~\ref{sec:concl} with a summary and some open questions.

\section{Preliminaries}

\label{sec:prelim}

\subsection{Quantum information}

Much of the background on quantum information theory reviewed here is
available in \cite{W15book}. Let $\mathcal{L}(\mathcal{H})$ denote the algebra
of bounded linear operators acting on a Hilbert space $\mathcal{H}$. Let
$\mathcal{L}_{+}(\mathcal{H})$ denote the subset of positive semi-definite
operators. We also write $X\geq0$ if $X\in\mathcal{L}_{+}(\mathcal{H})$.
An\ operator $\rho$ is in the set $\mathcal{D}(\mathcal{H})$\ of density
operators (or states) if $\rho\in\mathcal{L}_{+}(\mathcal{H})$ and Tr$\left\{
\rho\right\}  =1$. An\ operator $\rho$ is in the set $\mathcal{D}_{\leq
}(\mathcal{H})$\ of subnormalized density operators if $\rho\in\mathcal{L}%
_{+}(\mathcal{H})$ and Tr$\left\{  \rho\right\}  \leq1$. The tensor product of
two Hilbert spaces $\mathcal{H}_{A}$ and $\mathcal{H}_{B}$ is denoted by
$\mathcal{H}_{A}\otimes\mathcal{H}_{B}$ or $\mathcal{H}_{AB}$.\ Given a
multipartite density operator $\rho_{AB}\in\mathcal{D}(\mathcal{H}_{A}%
\otimes\mathcal{H}_{B})$, we unambiguously write $\rho_{A}=\operatorname{Tr}%
_{B}\{\rho_{AB}\}$ for the reduced density operator on system $A$. We use
$\rho_{AB}$, $\sigma_{AB}$, $\tau_{AB}$, $\omega_{AB}$, etc.~to denote general
density operators in $\mathcal{D}(\mathcal{H}_{A}\otimes\mathcal{H}_{B})$,
while $\psi_{AB}$, $\varphi_{AB}$, $\phi_{AB}$, etc.~denote rank-one density
operators (pure states) in $\mathcal{D}(\mathcal{H}_{A}\otimes\mathcal{H}%
_{B})$ (with it implicit, clear from the context, and the above convention
implying that $\psi_{A}$, $\varphi_{A}$, $\phi_{A}$ may be mixed if $\psi
_{AB}$, $\varphi_{AB}$, $\phi_{AB}$ are pure). A purification $|\phi^{\rho
}\rangle_{RA}\in\mathcal{H}_{R}\otimes\mathcal{H}_{A}$ of a state $\rho_{A}%
\in\mathcal{D}(\mathcal{H}_{A})$ is such that $\rho_{A}=\operatorname{Tr}%
_{R}\{|\phi^{\rho}\rangle\langle\phi^{\rho}|_{RA}\}$. As is conventional, we
often say that a unit vector $|\psi\rangle$ is a pure state or a pure-state
vector (while also saying that $|\psi\rangle\langle\psi|$ is a pure state). An
extension of a state $\rho_{A}\in\mathcal{S}\left(  \mathcal{H}_{A}\right)  $
is some state $\rho_{RA}\in\mathcal{S}\left(  \mathcal{H}_{R}\otimes
\mathcal{H}_{A}\right)  $ such that $\operatorname{Tr}_{R}\left\{  \rho
_{RA}\right\}  =\rho_{A}$. Often, an identity operator is implicit if we do
not write it explicitly (and should be clear from the context). We employ the
shorthand supp$(A)$ and ker$(A)$ to refer to the support and kernel of an
operator $A$, respectively.

Let $\{|i\rangle_{A}\}$ be an orthonormal basis (i.e., the standard basis)
associated to a Hilbert space $\mathcal{H}_{A}$, and let $\{|i\rangle_{B}\}$
be defined similarly for $\mathcal{H}_{B}$. If these spaces are
finite-dimensional and their dimensions are equal ($\dim(\mathcal{H}_{A}%
)=\dim(\mathcal{H}_{B})=d$), then we define the maximally entangled state
vector $|\Phi\rangle_{AB}\in\mathcal{H}_{A}\otimes\mathcal{H}_{B}$ as%
\begin{equation}
|\Phi\rangle_{AB}\equiv\frac{1}{\sqrt{d}}\sum_{i=0}^{d-1}|i\rangle_{A}%
\otimes|i\rangle_{B}.
\end{equation}
A state $\sigma_{AB}\in\mathcal{D}(\mathcal{H}_{A}\otimes\mathcal{H}_{B})$ is
separable if it can be written in the following form \cite{W89}:%
\begin{equation}
\sigma_{AB}=\sum_{x}p_{X}(x)|\phi^{x}\rangle\langle\phi^{x}|_{A}%
\otimes|\varphi^{x}\rangle\langle\varphi^{x}|_{B},
\end{equation}
where $p_{X}$ is a probability distribution and $\{|\phi^{x}\rangle_{A}\}$ and
$\{|\varphi^{x}\rangle_{B}\}$ are sets of pure-state vectors. Let
$\mathcal{S}(A\!:\!B)$ denote the set of separable states acting on
$\mathcal{H}_{A}\otimes\mathcal{H}_{B}$. Note that%
\begin{equation}
\mathcal{S}(A\!:\!B)=\operatorname{conv}\{|\phi\rangle\langle\phi|_{A}%
\otimes|\varphi\rangle\langle\varphi|_{B}:|\phi\rangle_{A}\in\mathcal{H}%
_{A},|\varphi\rangle_{B}\in\mathcal{H}_{B},\left\Vert |\phi\rangle
_{A}\right\Vert _{2}=\left\Vert |\varphi\rangle_{B}\right\Vert _{2}=1\},
\end{equation}
where $\operatorname{conv}$ denotes the convex hull.

A linear map $\mathcal{N}_{A\rightarrow B}:\mathcal{L}(\mathcal{H}%
_{A})\rightarrow\mathcal{L}(\mathcal{H}_{B})$\ is positive if $\mathcal{N}%
_{A\rightarrow B}\left(  \sigma_{A}\right)  \in\mathcal{L}_{+}(\mathcal{H}%
_{B})$ whenever $\sigma_{A}\in\mathcal{L}_{+}(\mathcal{H}_{A})$. Let id$_{A}$
denote the identity map acting on a system $A$. A linear map $\mathcal{N}%
_{A\rightarrow B}$ is completely positive if the map id$_{R}\otimes
\mathcal{N}_{A\rightarrow B}$ is positive for a reference system $R$ of
arbitrary size. A linear map $\mathcal{N}_{A\rightarrow B}$ is
trace-preserving if $\operatorname{Tr}\left\{  \mathcal{N}_{A\rightarrow
B}\left(  \tau_{A}\right)  \right\}  =\operatorname{Tr}\left\{  \tau
_{A}\right\}  $ for all input operators $\tau_{A}\in\mathcal{L}(\mathcal{H}%
_{A})$. It is trace non-increasing if $\operatorname{Tr}\left\{
\mathcal{N}_{A\rightarrow B}\left(  \tau_{A}\right)  \right\}  \leq
\operatorname{Tr}\left\{  \tau_{A}\right\}  $ for all $\tau_{A}\in
\mathcal{L}_{+}(\mathcal{H}_{A})$. A quantum channel is a linear map which is
completely positive and trace-preserving (CPTP). Every quantum channel has a
Kraus representation as $\mathcal{N}_{A\rightarrow B}( \tau_{A}) =\sum
_{x}E^{x}\tau_{A}(E^{x})^{\dag}$ where the Kraus operators $\{E^{x}\}$ satisfy
$\sum_{x}(E^{x})^{\dag}E^{x}=I_{A}$. A positive operator-valued measure (POVM)
is a set $\left\{  \Lambda^{m}\right\}  $ of positive semi-definite operators
such that $\sum_{m}\Lambda^{m}=I$. An isometry $U:\mathcal{H}\rightarrow
\mathcal{H}^{\prime}$ is a linear map such that $U^{\dag}U=I_{\mathcal{H}}$.
An isometric extension $U_{A\rightarrow BE}^{\mathcal{N}}$ of a quantum
channel $\mathcal{N}_{A\rightarrow B}$ (i.e., its Stinespring dilation
\cite{S55}) is a linear map that satisfies the following:%
\begin{align}
\operatorname{Tr}_{E}\!\left\{  U_{A\rightarrow BE}^{\mathcal{N}}\rho
_{A}(U_{A\rightarrow BE}^{\mathcal{N}})^{\dag}\right\}   &  =\mathcal{N}%
_{A\rightarrow B}(\rho_{A}),\\
U_{\mathcal{N}}^{\dagger}U_{\mathcal{N}}  &  =I_{A},\\
U_{\mathcal{N}}U_{\mathcal{N}}^{\dagger}  &  =\Pi_{BE},
\end{align}
for all states $\rho_{A}\in\mathcal{D}(\mathcal{H}_{A})$, where $\Pi_{BE}$ is
a projection onto a subspace of the Hilbert space $\mathcal{H}_{B}%
\otimes\mathcal{H}_{E}$. We define an isometric channel $\mathcal{U}%
_{A\rightarrow BE}^{\mathcal{N}}$ from the isometry $U_{A\rightarrow
BE}^{\mathcal{N}}$ as%
\begin{equation}
\mathcal{U}_{A\rightarrow BE}^{\mathcal{N}}(X_{A})=U_{A\rightarrow
BE}^{\mathcal{N}}X_{A}(U_{A\rightarrow BE}^{\mathcal{N}})^{\dag},
\end{equation}
where $X_{A}\in\mathcal{L}(\mathcal{H}_{A})$.

The trace distance between two quantum states $\rho,\sigma\in\mathcal{D}%
(\mathcal{H})$\ is equal to $\left\Vert \rho-\sigma\right\Vert _{1}$, where
$\left\Vert C\right\Vert _{1}\equiv\operatorname{Tr}\{\sqrt{C^{\dag}C}\}$ for
any operator $C$. It has a direct operational interpretation in terms of the
distinguishability of these states. That is, if $\rho$ or $\sigma$ are
prepared with equal probability and the task is to distinguish them via some
quantum measurement, then the optimal success probability in doing so is equal
to $\left(  1+\left\Vert \rho-\sigma\right\Vert _{1}/2\right)  /2$. The
fidelity is defined as $F(\rho,\sigma)\equiv\left\Vert \sqrt{\rho}\sqrt
{\sigma}\right\Vert _{1}^{2}$ \cite{U76}, and more generally we can use the
same formula to define $F(P,Q)$ if $P,Q\in\mathcal{L}_{+}(\mathcal{H})$.
Uhlmann's theorem states that \cite{U76}%
\begin{equation}
F(\rho_{A},\sigma_{A})=\max_{U}\left\vert \langle\phi^{\sigma}|_{RA}%
U_{R}\otimes I_{A}|\phi^{\rho}\rangle_{RA}\right\vert ^{2},
\label{eq:uhlmann-thm}%
\end{equation}
where $|\phi^{\rho}\rangle_{RA}$ and $|\phi^{\sigma}\rangle_{RA}$ are fixed
purifications of $\rho_{A}$ and $\sigma_{A}$, respectively, and the
optimization is with respect to all unitaries $U_{R}$. The same statement
holds more generally for $P,Q\in\mathcal{L}_{+}(\mathcal{H})$. The fidelity is
invariant with respect to isometries and monotone non-decreasing with respect
to channels. The sine distance or $C$-distance between two quantum states
$\rho,\sigma\in\mathcal{D}(\mathcal{H})$ was defined as
\begin{equation}
C(\rho,\sigma)\equiv\sqrt{1-F(\rho,\sigma)}%
\end{equation}
and proven to be a metric in \cite{R02,R03,GLN04,R06}. It was
later~\cite{TCR09} (under the name \textquotedblleft purified
distance\textquotedblright) shown to be a metric on subnormalized states
$\rho,\sigma\in\mathcal{D}_{\leq}(\mathcal{H})$ via the embedding
\begin{equation}
P(\rho,\sigma)\equiv C(\rho\oplus\left[  1-\operatorname{Tr}\{\rho\}\right]
,\sigma\oplus\left[  1-\operatorname{Tr}\{\sigma\}\right]  ) \,.
\label{eq:purified-distance}%
\end{equation}


\subsection{Quantum channels with symmetries}

Let $G$ be a finite group, and for every $g\in G$, let $g\rightarrow U_{A}(g)$
and $g\rightarrow V_{B}(g)$ be unitary representations acting on the input and
output spaces of a quantum channel $\mathcal{N}_{A\rightarrow B}$,
respectively. Then a\ quantum channel $\mathcal{N}_{A\rightarrow B}$\ is
covariant with respect to these representations if the following relation
holds for all input density operators $\rho_{A}\in\mathcal{D}(A)$ and group
elements $g\in G$ \cite{H02}:%
\begin{equation}
\mathcal{N}_{A\rightarrow B}\!\left(  U_{A}(g)\rho_{A}U_{A}^{\dag}(g)\right)
=V_{B}(g)\mathcal{N}_{A\rightarrow B}(\rho_{A})V_{B}^{\dag}(g).
\end{equation}

\begin{definition}
[Covariant channel]\label{def:covariant-channel}A quantum channel is covariant
if it is covariant with respect to a group which has a representation $U(g)$
on $\mathcal{H}_{A}$ that is a unitary one-design, the latter meaning that the
channel $\frac{1}{|G|}\sum_{g\in G}U(g)(\cdot)U(g)^{\dag}$ always outputs the
maximally mixed state.
\end{definition}

The teleportation protocol is a basic primitive in quantum information
\cite{PhysRevLett.70.1895}. We say that a channel is \textquotedblleft
teleportation-simulable\textquotedblright\ with associated state $\omega_{AB}$
if it can be realized by the action of the teleportation protocol on one share
of a bipartite state $\omega_{AB}$ \cite[Section~V]{BDSW96}. That is, a
channel $\mathcal{N}_{A^{\prime}\rightarrow B}$ is teleportation-simulable
with associated state $\omega_{AB}$ if there exists a state $\omega_{AB}%
\in\mathcal{D}(\mathcal{H}_{A}\otimes\mathcal{H}_{B})$, with $\mathcal{H}%
_{A}\simeq\mathcal{H}_{A^{\prime}}$, such that for all $\rho_{A^{\prime}}%
\in\mathcal{D}(\mathcal{H}_{A^{\prime}})$
\begin{equation}
\mathcal{N}_{A^{\prime}\rightarrow B}(\rho_{A^{\prime}})=\mathcal{T}%
_{AA^{\prime}B\rightarrow B}(\rho_{A^{\prime}}\otimes\omega_{AB}),
\label{eq:TP-covariant}%
\end{equation}
where $\mathcal{T}_{AA^{\prime}B\rightarrow B}$ is a channel corresponding to
a general teleportation protocol \cite{Werner01} (note here that the
correction operations might need to be adapted for the output space of the
channel $\mathcal{N}_{A^{\prime}\rightarrow B}$, which could be different from
the input space). The advantage of channels possessing this symmetry is that
any protocol involving $\mathcal{N}_{A^{\prime}\rightarrow B}$ can be replaced
by one which involves Alice and Bob sharing $\omega_{AB}$ and performing
quantum teleportation to simulate $\mathcal{N}_{A^{\prime}\rightarrow B}$.

The following proposition establishes that every covariant channel is
teleportation-simulable. Its proof is given in Appendix~\ref{app:cov->TP-sim}
and extends earlier developments in \cite{Werner01} and \cite[Eqs.~(53)--(56)]%
{LM15}.

\begin{proposition}
\label{prop:cov->TP-sim}If a quantum channel $\mathcal{N}_{A\rightarrow B}%
$\ is covariant (as given in Definition~\ref{def:covariant-channel}), then it
is teleportation simulable with associated state $\omega_{AB} \equiv
\mathcal{N}_{A^{\prime}\to B}(\Phi_{AA^{\prime}})$.
\end{proposition}

\subsection{Local operations and (public) classical communication}

A quantum instrument\ is a quantum channel that accepts a quantum system as
input and outputs two systems:\ a classical one and a quantum one
\cite{DL70,D76,Ozawa1984}. More formally, a quantum instrument is a collection
$\{\mathcal{N}^{x}\}$\ of completely positive trace non-increasing maps, such
that the sum map $\sum_{x}\mathcal{N}^{x}$ is a quantum channel. We can write
the action of a quantum instrument on an input density operator $\rho
\in\mathcal{D}(\mathcal{H})$ as the following quantum channel:%
\begin{equation}
\rho\rightarrow\sum_{x}\mathcal{N}^{x}(\rho)\otimes|x\rangle\langle x|,
\end{equation}
where $\left\{  |x\rangle\right\}  _{x} $ is an orthonormal basis labeling the
classical output of the instrument.

It is common in quantum communication theory to consider the framework of
local operations and classical communication (LOCC) \cite{BDSW96,CLMOW14},
which consists of particular interactions between two parties usually called
Alice and Bob. A round of LOCC\ (or LOCC\ channel) consists of a finite number
of compositions of the following:

\begin{enumerate}
\item Alice performs a quantum instrument, which has both a quantum and
classical output. She forwards the classical output to Bob, who then performs
a quantum channel conditioned on the classical data received. This sequence of
actions corresponds to the following channel:%
\begin{equation}
\sum_{x}\mathcal{E}_{A}^{x}\otimes\mathcal{F}_{B}^{x},
\label{eq-em:LOCC-channel}%
\end{equation}
where $\{\mathcal{E}_{A}^{x}\}$ is a collection of completely positive maps
such that $\sum_{x}\mathcal{E}_{A}^{x}$ is a quantum channel and
$\{\mathcal{F}_{B}^{x}\}$ is a collection of quantum channels.

\item The situation is reversed, with Bob performing the initial instrument,
who forwards the classical data to Alice, who then performs a quantum channel
conditioned on the classical data. This sequence of actions corresponds to a
channel of the form in \eqref{eq-em:LOCC-channel}, with the $A$ and $B$ labels switched.
\end{enumerate}

\noindent The framework of local operations and public communication (LOPC) is
essentially the same as LOCC, except that the terminology implies that there
is a third party Eve (an eavesdropper)\ who receives a copy of all of the
classical data exchanged between Alice and Bob.

A channel is separable \cite{Rains2008} if it has Kraus operators of the form
$C_{A}^{x}\otimes D_{B}^{x}$, where $\sum_{x}(C_{A}^{x})^{\dag}C_{A}%
^{x}\otimes(D_{B}^{x})^{\dag}D_{B}^{x}=I_{AB}$. Every LOCC\ channel is
separable, but the opposite is not always true \cite{PhysRevA.59.1070}. A
channel is separability preserving \cite{HN03,BP2010} if it preserves the set
of separable states. The swap operator is an example of a
separability-preserving channel that is not a separable channel.

\subsection{Private states}

\label{sec:private-states}Private states are an essential notion for our
analysis \cite{HHHO05,HHHO09}, and we review their basics here.

\begin{definition}
\label{def:tripartite-key-state-1-def}A tripartite key state $\gamma_{ABE}%
\in\mathcal{D}(\mathcal{H}_{ABE})$\ contains $\log K$ bits of secret key if
there exists a state $\sigma_{E}\in\mathcal{D}(\mathcal{H}_{E})$ and a
projective measurement channel $\mathcal{M}(\cdot)=\sum_{i}|i\rangle\langle
i|(\cdot)|i\rangle\langle i|$, where $\{\vert i \rangle\}_{i}$ is an
orthonormal basis, such that%
\begin{equation}
\left(  \mathcal{M}_{A}\otimes\mathcal{M}_{B}\right)  \left(  \gamma
_{ABE}\right)  =\frac{1}{K}\sum_{i=0}^{K-1}|i\rangle\langle i|_{A}%
\otimes|i\rangle\langle i|_{B}\otimes\sigma_{E}.
\label{eq:tri-party-priv-state}%
\end{equation}

\end{definition}

That is, we see that the systems $A$ and $B$ are maximally classically
correlated, and the key value is uniformly random and independent of the $E$
system. Physically, we can think of the $A$ system as being in Alice's
laboratory, $B$ in Bob's, and $E$ in Eve's. We also think of Alice and Bob as
two honest parties and Eve as a malicious eavesdropper whose system should
ideally be independent of the key systems possessed by Alice and Bob.

Purifying such a state $\gamma_{ABE}$ with two systems $A^{\prime}$ and
$B^{\prime}$, thinking of $A^{\prime}$ as being available to Alice and
$B^{\prime}$ as being available to Bob (or alternatively simply as not being
available to Eve), and tracing out the $E$ system then leads to the notion of
a bipartite private state $\gamma_{ABA^{\prime}B^{\prime}}$
\cite{HHHO05,HHHO09}. As shown in \cite{HHHO05,HHHO09}, any such state
$\gamma_{ABA^{\prime}B^{\prime}}\in\mathcal{D}(\mathcal{H}_{ABA^{\prime
}B^{\prime}})$ takes a canonical form:

\begin{definition}
\label{def:private-state-bi} A bipartite private state $\gamma_{A^{\prime
\prime}B^{\prime\prime}} \in\mathcal{D}(\mathcal{H}_{A^{\prime\prime}%
B^{\prime\prime}})$ contains $\log K$ bits of secret key if $\mathcal{H}%
_{A^{\prime\prime}} = \mathcal{H}_{A} \otimes\mathcal{H}_{A^{\prime}}$ and
$\mathcal{H}_{B^{\prime\prime}} = \mathcal{H}_{B} \otimes\mathcal{H}%
_{B^{\prime}}$ such that $\gamma_{ABA^{\prime}B^{\prime}}\in\mathcal{D}%
(\mathcal{H}_{ABA^{\prime}B^{\prime}})$ has the following form:%
\begin{equation}
\gamma_{ABA^{\prime}B^{\prime}}=U_{ABA^{\prime}B^{\prime}}(\Phi_{AB}%
\otimes\theta_{A^{\prime}B^{\prime}})U_{ABA^{\prime}B^{\prime}}^{\dag},
\label{eq:bi-party-priv-state}%
\end{equation}
where $\Phi_{AB}$ is a maximally entangled state of Schmidt rank $K$,
$U_{ABA^{\prime}B^{\prime}}$ is a \textquotedblleft twisting\textquotedblright%
\ unitary of the form%
\begin{equation}
U_{ABA^{\prime}B^{\prime}}=\sum_{i,j=0}^{K-1}|i\rangle\langle i|_{A}%
\otimes|j\rangle\langle j|_{B}\otimes U_{A^{\prime}B^{\prime}}^{ij},
\label{eq:twisting-unitary}%
\end{equation}
with each $U_{A^{\prime}B^{\prime}}^{ij}$ a unitary, and $\theta_{A^{\prime
}B^{\prime}}\in\mathcal{D}(\mathcal{H}_{A^{\prime}B^{\prime}})$.
\end{definition}

The systems $A^{\prime}$ and $B^{\prime}$ are called the \textquotedblleft
shield\textquotedblright\ systems because they, along with the twisting
unitary, can help to protect the key in systems $A$ and $B$ from any party
possessing a purification of $\gamma_{ABA^{\prime}B^{\prime}}$. Such bipartite
private states are in one-to-one correspondence with the tripartite key states
given in \eqref{eq:tri-party-priv-state} \cite{HHHO05,HHHO09}. That is, for
every state $\gamma_{ABE}$ of the form in \eqref{eq:tri-party-priv-state}, we
can find a state of the form in \eqref{eq:bi-party-priv-state} and vice versa.
We summarize this as the following proposition:

\begin{proposition}
[\cite{HHHO05,HHHO09}]Bipartite private states and tripartite key states are
equivalent. That is, for $\gamma_{ABA^{\prime}B^{\prime}}$ a bipartite private
state, $\gamma_{ABE}$ is a tripartite key state for any purification
$\gamma_{ABA^{\prime}B^{\prime}E}$ of $\gamma_{ABA^{\prime}B^{\prime}}$.
Conversely, for any tripartite key state $\gamma_{ABE}$ and any purification
$\gamma_{ABA^{\prime}B^{\prime}E}$ of it, $\gamma_{ABA^{\prime}B^{\prime}}$ is
a bipartite private state.
\end{proposition}

This correspondence takes on a more physical form (reviewed in
Section~\ref{sec:priv-class-comm}), which is that any tripartite protocol
whose aim it is to extract tripartite key states of the form in
\eqref{eq:tri-party-priv-state} is in one-to-one correspondence with a
bipartite protocol whose aim it is to extract bipartite private states of the
form in \eqref{eq:bi-party-priv-state} \cite{HHHO05,HHHO09}.

\begin{definition}
\label{def:approx-priv}A state $\rho_{ABE}\in\mathcal{D}(\mathcal{H}_{ABE})$
is an $\varepsilon$-approximate tripartite key state if there exists a
tripartite key state $\gamma_{ABE}$ of the form in
\eqref{eq:tri-party-priv-state} such that%
\begin{equation}
F(\rho_{ABE},\gamma_{ABE})\geq1-\varepsilon, \label{eq:approx-tri-priv-state}%
\end{equation}
where $\varepsilon\in\left[  0,1\right]  $. Similarly, a state $\rho
_{ABA^{\prime}B^{\prime}}\in\mathcal{D}(\mathcal{H}_{ABA^{\prime}B^{\prime}})$
is an $\varepsilon$-approximate bipartite private state if there exists a
bipartite private state $\gamma_{ABA^{\prime}B^{\prime}}\in\mathcal{D}%
(\mathcal{H}_{ABA^{\prime}B^{\prime}})$ of the form in
\eqref{eq:tri-party-priv-state} such that%
\begin{equation}
F(\rho_{ABA^{\prime}B^{\prime}},\gamma_{ABA^{\prime}B^{\prime}})\geq
1-\varepsilon. \label{eq:approx-bi-priv-state}%
\end{equation}

\end{definition}

Approximate tripartite key states are in one-to-one correspondence with
approximate bipartite private states \cite[Theorem~5]{HHHO09}, as summarized below:

\begin{proposition}
[\cite{HHHO05,HHHO09}]If $\rho_{ABA^{\prime}B^{\prime}}$ is an $\varepsilon
$-approximate bipartite key state with $K$ key values, then Alice and Bob hold
an $\varepsilon$-approximate tripartite key state with $K$ key values. The
converse statement is true as well.
\end{proposition}

\section{Secret key transmission and generation over quantum channels}

\label{sec:priv-class-comm}In this section, we define secret-key transmission
and generation codes and corresponding measures of their performance. We also
review the identification from \cite{HHHO05,HHHO09}, which shows how a
tripartite key distillation protocol is in one-to-one correspondence with a
bipartite private state distillation protocol.

\subsection{Secret-key transmission codes}

\label{sec:private-comm-codes}Given is a quantum channel $\mathcal{N}%
_{A^{\prime}\rightarrow B}$. Let $\mathcal{N}_{A^{\prime}\rightarrow
B}^{\otimes n}$ denote the tensor-product channel, $U_{A^{\prime}\rightarrow
BE}^{\mathcal{N}}$ an isometric extension of $\mathcal{N}_{A^{\prime
}\rightarrow B}$, and $\mathcal{U}_{A^{\prime}\rightarrow BE}^{\mathcal{N}}$
the associated isometric channel. A secret-key transmission protocol for $n$
channel uses consists of a triple $\{\left\vert K\right\vert ,\mathcal{E}%
,\mathcal{D}\}$, where $\left\vert K\right\vert $ is the size of the secret
key to be generated, $\mathcal{E}_{K^{\prime}\rightarrow A^{\prime n}}$ is the
encoder (a CPTP\ map), and $\mathcal{D}_{B^{n}\rightarrow\hat{K}}$ is the
decoder (another CPTP map). The protocol begins with a third party preparing a
maximally classically correlated state $\overline{\Phi}_{KK^{\prime}}$ of the
following form:%
\begin{equation}
\overline{\Phi}_{KK^{\prime}}\equiv\frac{1}{\left\vert K\right\vert }%
\sum_{i=0}^{\left\vert K\right\vert -1}|i\rangle\langle i|_{K}\otimes
|i\rangle\langle i|_{K^{\prime}},
\end{equation}
and then sending the $K^{\prime}$ system to Alice. Alice then inputs the
$K^{\prime}$ system to an encoder $\mathcal{E}_{K^{\prime}\rightarrow
A^{\prime n}}$, transmits the $A^{\prime n}$ systems through the tensor-power
channel $(\mathcal{U}_{A^{\prime}\rightarrow BE}^{\mathcal{N}})^{\otimes n}$,
and the receiver Bob applies the decoder $\mathcal{D}_{B^{n}\rightarrow\hat
{K}}$ to the systems~$B^{n}$. The state at the end of the protocol is as
follows:%
\begin{equation}
\rho_{K\hat{K}E^{n}}\equiv(\mathcal{D}_{B^{n}\rightarrow\hat{K}}%
\circ(\mathcal{U}_{A^{\prime}\rightarrow BE}^{\mathcal{N}})^{\otimes n}%
\circ\mathcal{E}_{K^{\prime}\rightarrow A^{\prime n}})(\overline{\Phi
}_{KK^{\prime}}).
\end{equation}
Figure~\ref{fig:private-code}\ depicts such a protocol.\begin{figure}[ptb]
\begin{center}
\includegraphics[
width=4.0638in
]{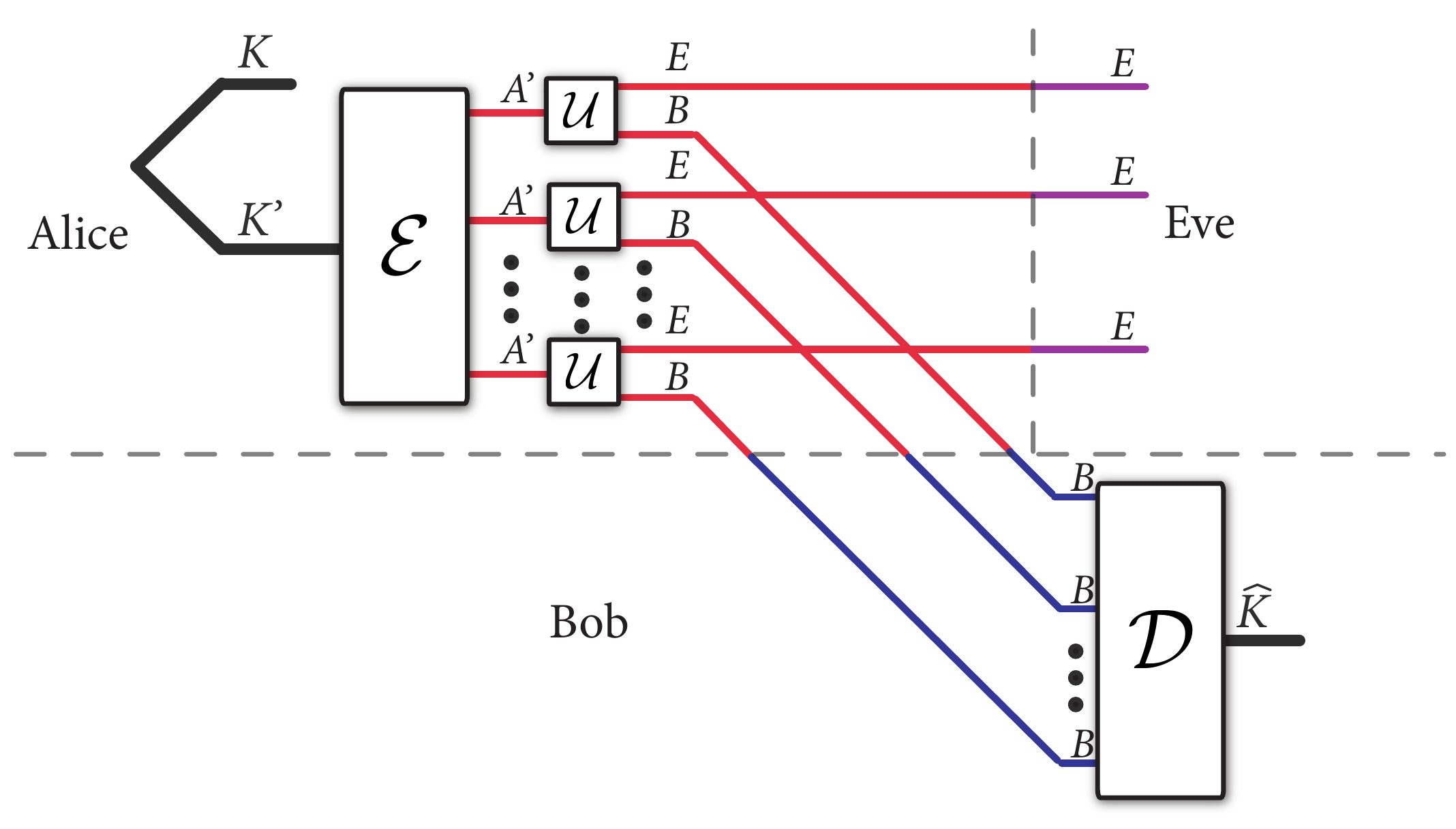}
\end{center}
\caption{A protocol for secret-key transmission over a quantum channel.}%
\label{fig:private-code}%
\end{figure}

A triple $\left(  n,P,\varepsilon\right)  $ consists of the number $n$\ of
channel uses, the rate $P$ of secret-key transmission, and the error
$\varepsilon\in\left[  0,1\right]  $. Such a triple is achievable on
$\mathcal{N}_{A^{\prime}\rightarrow B}$\ if there exists a secret-key
transmission protocol $\{\left\vert K\right\vert ,\mathcal{E},\mathcal{D}\}$
and some state $\omega_{E^{n}}\in\mathcal{D}(\mathcal{H}_{E^{n}})$ such that
$\frac{1}{n}\log\left\vert K\right\vert \geq P$ and
\begin{equation}
F(\overline{\Phi}_{K\hat{K}}\otimes\omega_{E^{n}},\rho_{K\hat{K}E^{n}}%
)\geq1-\varepsilon. \label{eq:code-performance-fid}%
\end{equation}
Thus, the goal of such a secret-key transmission protocol is to realize an
$\varepsilon$-approximate tripartite secret-key state as defined in \eqref{eq:approx-tri-priv-state}.

Note that the above definition of secret-key transmission combines the error
probability and the security parameter into a single parameter $\varepsilon$,
in contrast to the definitions from
\cite{ieee2005dev,1050633,MW13,HTW14,Win16}. Doing so is consistent with the
definition of private capacity or distillable key from
\cite{HHHO05,HHHO09,C06} and turns out to be beneficial for the developments
in this paper. Furthermore, we argue in Appendix~\ref{app:priv-def-str-conv}%
\ how a converse bound according to the above definition of privacy gives a
converse bound according to quantum generalizations of the privacy definition
from \cite{HTW14}.

As mentioned before Definition~\ref{def:approx-priv}, it is possible to purify
a secret-key transmission protocol \cite{HHHO05,HHHO09}, such that every step
is performed coherently and the ultimate goal is to realize a private
bipartite state $\gamma_{K_{A}K_{B}S_{A}S_{B}}$, where we now denote the key
systems by $K$ and the shield systems by $S$. In the class of protocols
discussed above, this consists of replacing each step with the following:

\begin{enumerate}
\item A third party preparing a purification of the state $\overline{\Phi
}_{KK^{\prime}}$, which is a \textquotedblleft GHZ\ state\textquotedblright%
\ that we denote by $|\Phi^{\operatorname{GHZ}}\rangle_{KK^{\prime}M}%
\equiv\left\vert K\right\vert ^{-1/2}\sum_{i}|i\rangle_{K}\otimes
|i\rangle_{K^{\prime}}\otimes|i\rangle_{M}$, and giving the $K^{\prime}$
system to Alice,

\item Alice performing an isometric extension of the encoder $\mathcal{E}%
_{K^{\prime}\rightarrow A^{\prime n}}$, denoted by $\mathcal{U}_{K^{\prime
}\rightarrow A^{\prime n}A^{\prime\prime}}^{\mathcal{E}}$,

\item Bob performing an isometric extension of the decoder $\mathcal{D}%
_{B^{n}\rightarrow\hat{K}}$, denoted by $\mathcal{U}_{B^{n}\rightarrow\hat
{K}B^{\prime\prime}}^{\mathcal{D}}$.
\end{enumerate}

\noindent Figure~\ref{fig:private-prot-purified}\ depicts such a purified
version of a secret-key transmission protocol. \begin{figure}[ptb]
\begin{center}
\includegraphics[
width=3.3304in
]{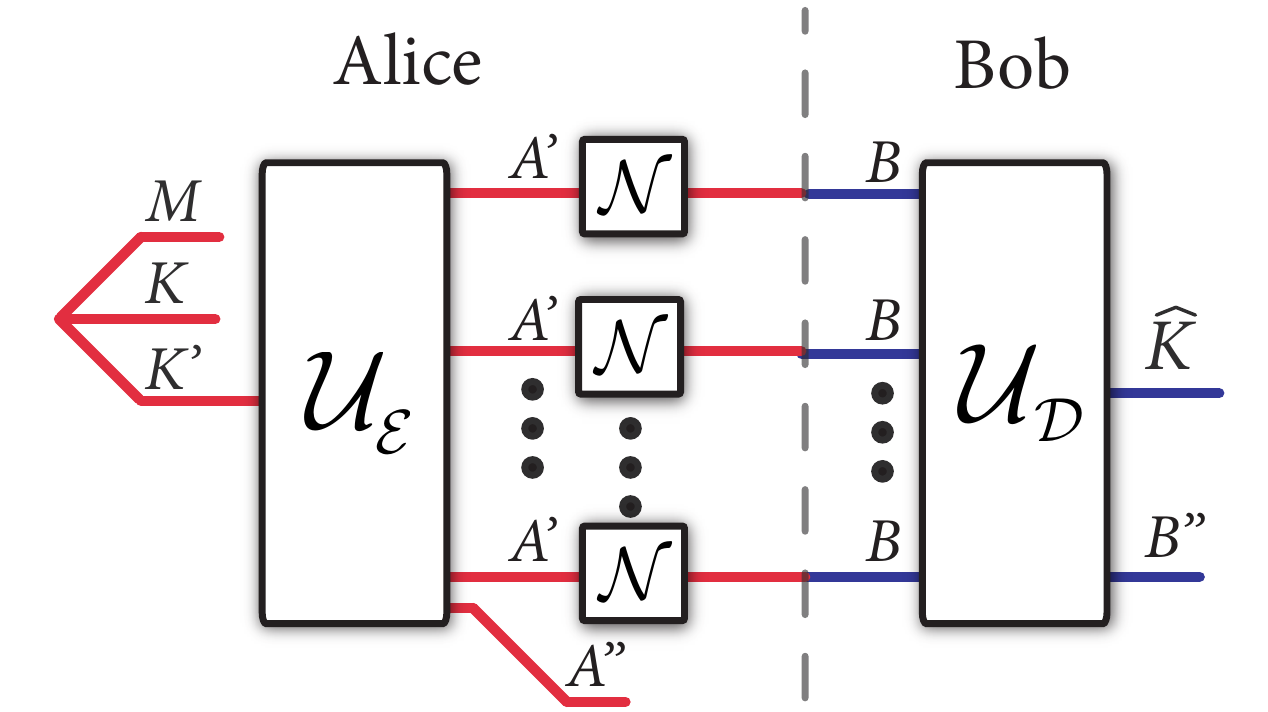}
\end{center}
\caption{A purified version of a secret-key transmission protocol.}%
\label{fig:private-prot-purified}%
\end{figure}By employing \cite[Theorem~5]{HHHO09}, we find that
\eqref{eq:code-performance-fid} implies that%
\begin{equation}
F(\gamma_{K_{A}K_{B}S_{A}S_{B}},\rho_{K\hat{K}MA^{\prime\prime}B^{\prime
\prime}})\geq1-\varepsilon,
\end{equation}
for some private state $\gamma_{K_{A}K_{B}S_{A}S_{B}}$, where we make the
identifications $K_{A}\equiv K$, $K_{B}\equiv\hat{K}$, $S_{A}\equiv
MA^{\prime\prime}$, and $S_{B}\equiv B^{\prime\prime}$, and%
\begin{equation}
\rho_{K\hat{K}MA^{\prime\prime}B^{\prime\prime}E^n }\equiv(\mathcal{U}%
_{B^{n}\rightarrow\hat{K}B^{\prime\prime}}^{\mathcal{D}}\circ(\mathcal{U}%
_{A^{\prime}\rightarrow BE}^{\mathcal{N}})^{\otimes n}\circ\mathcal{U}%
_{K^{\prime}\rightarrow A^{\prime n}A^{\prime\prime}}^{\mathcal{E}}%
)(\Phi_{KK^{\prime}M}^{\operatorname{GHZ}}).
\end{equation}

\subsection{Secret-key generation codes}

A secret-key generation protocol is defined similarly as above, with some key
differences however. The main difference is that the objective is secret key
\textit{generation}. As such, a secret-key generation protocol for $n$ channel
uses consists of a triple $\{\left\vert K\right\vert ,\varrho_{KA^{\prime n}%
},\mathcal{D}\}$, where $\left\vert K\right\vert $ is the size of the secret
key to be generated, $\varrho_{KA^{\prime n}}$ is the initial state, and
$\mathcal{D}_{B^{n}\rightarrow\hat{K}}$ is the decoder. Such a protocol begins
with Alice preparing the state $\varrho_{KA^{\prime n}}$, sending the
$A^{\prime n}$ systems through the tensor-power channel $(\mathcal{U}%
_{A^{\prime}\rightarrow BE}^{\mathcal{N}})^{\otimes n}$, and the receiver Bob
applies the decoder $\mathcal{D}_{B^{n}\rightarrow\hat{K}}$ to the
systems~$B^{n}$. The state at the end of the protocol is as follows:%
\begin{equation}
\varrho_{K\hat{K}E^{n}}\equiv(\mathcal{D}_{B^{n}\rightarrow\hat{K}}%
\circ(\mathcal{U}_{A^{\prime}\rightarrow BE}^{\mathcal{N}})^{\otimes
n})(\varrho_{KA^{\prime n}}).
\end{equation}

A triple $\left(  n,P,\varepsilon\right)  $ for secret-key generation consists
of the number $n$\ of channel uses, the rate $P$ of secret-key generation, and
the error $\varepsilon\in\left[  0,1\right]  $. Such a triple is achievable on
$\mathcal{N}_{A^{\prime}\rightarrow B}$\ for secret-key generation if there
exists a secret-key generation protocol $\{\left\vert K\right\vert
,\varrho_{KA^{\prime n}},\mathcal{D}\}$ and some state $\omega_{E^{n}}%
\in\mathcal{D}(\mathcal{H}_{E^{n}})$ such that $\frac{1}{n}\log\left\vert
K\right\vert \geq P$ and%
\begin{equation}
F(\overline{\Phi}_{K\hat{K}}\otimes\omega_{E^{n}},\varrho_{K\hat{K}E^{n}}%
)\geq1-\varepsilon.
\end{equation}
Thus, the goal of such a secret-key generation protocol is to generate an
$\varepsilon$-approximate tripartite secret-key state as defined in \eqref{eq:approx-tri-priv-state}.

Any secret-key generation protocol can be purified as discussed in the
previous section, such that the goal is to generate an $\varepsilon
$-approximate bipartite private state.

Finally, note that the encoder and decoder for an $(n,P,\varepsilon)$
secret-key transmission protocol can be used directly to realize an
$(n,P,\varepsilon)$ secret-key generation protocol, simply by setting the
initial state $\varrho_{KA^{\prime n}}$ for the secret-key generation protocol
equal to $\mathcal{E}_{K^{\prime}\rightarrow A^{\prime n}}(\overline{\Phi
}_{KK^{\prime}})$. The converse realization is not possible without the
assistance of another resource.

\subsection{Non-asymptotic achievable regions}

\label{sec:non-asym-region}\textbf{Unassisted protocols.} The non-asymptotic
private achievable region of a quantum channel is the union of all triples
$\left(  n,P,\varepsilon\right)  $ for secret-key transmission, and we are
interested in understanding two different boundaries of this region, defined
as%
\begin{align}
\hat{P}_{\mathcal{N}}(n,\varepsilon)  &  \equiv\max\left\{  P:\left(
n,P,\varepsilon\right)  \text{ is achievable for }\mathcal{N}\right\}  ,\\
\hat{\varepsilon}_{\mathcal{N}}(n,P)  &  \equiv\min\left\{  \varepsilon
:\left(  n,P,\varepsilon\right)  \text{ is achievable for }\mathcal{N}%
\right\}  .
\end{align}
In this paper, we investigate both of these boundaries. The first boundary
$\hat{P}_{\mathcal{N}}(n,\varepsilon)$ identifies how the rate can change as a
function of $n$ for fixed error $\varepsilon$, and second-order coding rates
can characterize this boundary for sufficiently large $n$. The second boundary
$\hat{\varepsilon}_{\mathcal{N}}(n,P)$ identifies how the error can change as
a function of $n$ for fixed rate $P$, and error exponents and strong converse
exponents characterize this boundary (in this paper we focus exclusively on
bounds on strong converse exponents).

\textbf{LOPC/LOCC-assisted protocols.} We can extend all of the above
definitions to the case in which Alice and Bob employ classical communication
to aid in their goal of establishing a secret key. We call such a protocol a
secret-key-agreement protocol. The most general such protocol in the
tripartite picture consists of rounds of local operations and public
communication (LOPC) interleaved between every channel use. By purifying every
operation and the classical data communicated in such a protocol, we can
describe such a protocol in the bipartite picture, which consists of local
operations and classical communication (LOCC) \cite{HHHO05,HHHO09}.
Figure~\ref{fig:LOCC-protocol}\ depicts such an LOCC-assisted protocol. The
output of such a protocol is compared via the fidelity with a private state,
and by \cite[Theorem~5]{HHHO09}, it meets the same fidelity requirement as the
original tripartite formulation. We define a triple $\left(  n,P,\varepsilon
\right)  $ to be achievable if there exists a secret-key agreement protocol of
the above form that generates an $\varepsilon$-approximate tripartite secret
key state in the tripartite picture or, equivalently, an $\varepsilon
$-approximate bipartite private state in the bipartite picture. We define the
achievable rate region as before as the union of all achievable rate triples,
and we are interested in the boundaries, defined as%
\begin{align}
\hat{P}_{\mathcal{N}}^{\leftrightarrow}(n,\varepsilon)  &  \equiv\max\left\{
P:\left(  n,P,\varepsilon\right)  \text{ is achievable for }\mathcal{N}\text{
using }{\leftrightarrow}\right\}  ,\\
\hat{\varepsilon}_{\mathcal{N}}^{\leftrightarrow}(n,P)  &  \equiv\min\left\{
\varepsilon:\left(  n,P,\varepsilon\right)  \text{ is achievable for
}\mathcal{N}\text{ using }{\leftrightarrow}\right\}  ,
\end{align}
where ${\leftrightarrow}$ indicates that the protocol is LOCC-assisted.

There is no difference in the performance of secret-key transmission and
secret-key generation protocols whenever classical communication is available
for free. To see this, consider that any $(n,P,\varepsilon)$ LOCC-assisted
secret-key transmission protocol realizes an $(n,P,\varepsilon)$ LOCC-assisted
secret-key generation protocol, for reasons similar to those that we discussed
previously. When classical communication is available for free, any
$(n,P,\varepsilon)$ LOCC-assisted secret-key generation protocol is an
instance of an $(n,P,\varepsilon)$ LOCC-assisted secret-key transmission
protocol. This follows by an application of the well known one-time pad
protocol. That is, suppose that the secret-key generation protocol produces an
$\varepsilon$-approximate key state of size $\left\vert K\right\vert $ shared
between Alice and Bob. Then if the third party had given the system
$K^{\prime}$ of $\overline{\Phi}_{KK^{\prime}}$ to Alice, she and Bob could
employ a one-time pad protocol, using their generated $\varepsilon
$-approximate key as a resource, in order for Alice to transmit the
$K^{\prime}$ system. The effect is to realize an $\varepsilon$-approximate key
state shared between the third party and Bob (due to the monotonicity of
fidelity with respect to any quantum channel; in this case, the one-time pad
protocol is a particular quantum channel). The resulting protocol is then an
$(n,P,\varepsilon)$ LOCC-assisted secret-key transmission
protocol.\begin{figure}[ptb]
\begin{center}
\includegraphics[
width=6.5673in
]{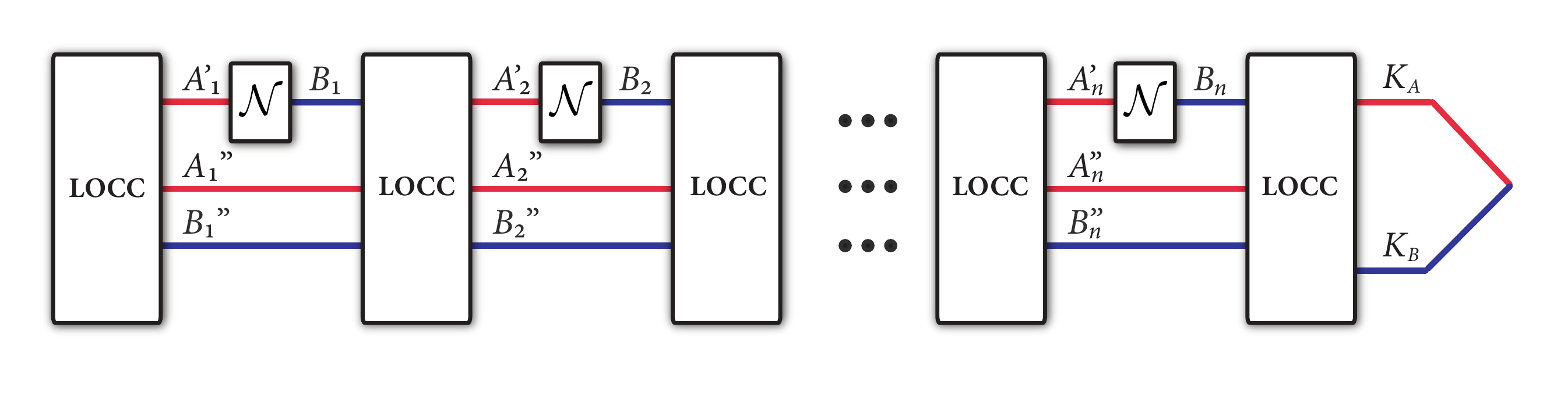}
\end{center}
\caption{An LOCC-assisted secret-key agreement protocol. For $i\in\left\{
1,\ldots,n\right\}  $, the $A_{i}^{\prime\prime}$ and $B_{i}^{\prime\prime}$
systems are scratch registers of arbitrary size that Alice and Bob can use in
between every channel use. The output of the protocol are systems $K_{A}$ and
$K_{B}$ which contain the key realized by the protocol.}%
\label{fig:LOCC-protocol}%
\end{figure}

\textbf{CPPP-assisted protocols.} Similarly, we can define a class of
protocols that consist of a single round of LOPC/LOCC, using the channel $n$
times, and a final round of LOPC/LOCC. We call these classical pre- and
post-processing (CPPP) protocols and define the corresponding boundaries as%
\begin{align}
\hat{P}_{\mathcal{N}}^{\operatorname{cppp}}(n,\varepsilon)  &  \equiv
\max\left\{  P:(n,P,\varepsilon)\text{ is achievable for }\mathcal{N}\text{
using cppp}\right\}  ,\\
\hat{\varepsilon}_{\mathcal{N}}^{\operatorname{cppp}}(n,P)  &  \equiv
\min\left\{  \varepsilon:(n,P,\varepsilon)\text{ is achievable for
}\mathcal{N}\text{ using cppp}\right\}  .
\end{align}
For similar reasons as discussed directly above, there is no difference in the
performance of secret-key transmission and secret-key generation protocols
that are CPPP-assisted.

From the definitions, we see that the following inequalities hold:%
\begin{align}
\hat{P}_{\mathcal{N}}(n,\varepsilon)  &  \leq\hat{P}_{\mathcal{N}%
}^{\operatorname{cppp}}(n,\varepsilon)\leq\hat{P}_{\mathcal{N}}%
^{\leftrightarrow}(n,\varepsilon),\label{eq:finite-P-hierarchy}\\
\hat{\varepsilon}_{\mathcal{N}}(n,P)  &  \geq\hat{\varepsilon}_{\mathcal{N}%
}^{\operatorname{cppp}}(n,P)\geq\hat{\varepsilon}_{\mathcal{N}}%
^{\leftrightarrow}(n,P),
\end{align}
because adding an extra resource can only help to increase the rate or reduce
the error.

\bigskip The structure of protocols involving adaptive LOCC\ can simplify
immensely for channels that are teleportation-simulable. This was realized in
\cite[Section~V]{BDSW96} for discrete-variable channels and extended in
\cite{NFC09,PLOB15} to continuous-variable bosonic channels (a general review
of the idea is available in \cite{PLOB15}). Since secret-key agreement
protocols in the tripartite picture can be recast as LOCC\ protocols in the
bipartite picture, this reduction by teleportation applies to them as well.
For such channels, a secret-key agreement protocol can be simulated by one in
which every channel use is replaced by Alice and Bob sharing the state
$\omega_{AB}$ from \eqref{eq:TP-covariant}\ and then performing the
teleportation protocol. This allows for Alice to take \textquotedblleft full
control\textquotedblright\ of the channel, delaying or advancing its use at
will. What this translates to for an adaptive secret-key agreement protocol is
that all of the adaptive rounds of LOCC\ can be delayed until the very end of
the protocol, such that the resulting protocol is a special kind of
CPPP-assisted protocol. Thus, for a teleportation-simulable channel
$\mathcal{N}_{A^{\prime}\rightarrow B}^{\operatorname{TP}}$ with associated
state $\omega_{AB}$, any secret-key agreement protocol can be simulated by one
which leads to a state of the following form:%
\begin{equation}
\Lambda_{A^{n}B^{n}\rightarrow K_{A}K_{B}}(\omega_{AB}^{\otimes n}),
\label{eq:SKA-TPC-protocols}%
\end{equation}
where $\Lambda_{A^{n}B^{n}\rightarrow K_{A}K_{B}}$ is an LOCC\ channel with
$K_{A}$ and $K_{B}$ the key systems generated for Alice and Bob, respectively.
For several teleportation-simulable channels of interest, it suffices to take
$\omega_{AB}=\mathcal{N}_{A^{\prime}\rightarrow B}^{\operatorname{TP}}%
(\Phi_{AA^{\prime}})$, in which Alice prepares a maximally entangled state and
sends one share of it through the channel. These observations were applied in
\cite[Section~V]{BDSW96} to the case of quantum communication protocols and in
\cite{PLOB15} to private communication after definitions of private capacity
were settled in \cite{ieee2005dev,1050633}.

\subsection{Relation to entanglement transmission}

\label{sec:q-comm}Entanglement transmission achieves the task of secret key
transmission. That is, an entanglement transmission code of rate $R$ and
fidelity $\geq1-\varepsilon$ can serve as a secret-key transmission code with
the same rate and fidelity. We formalize this relation now by defining
entanglement transmission codes, which are strongly related to quantum
communication codes \cite{BKN98}. An entanglement transmission code is a
triple $\left\{  \left\vert M\right\vert ,\mathcal{E},\mathcal{D}\right\}  $,
which consists of an encoding $\mathcal{E}_{M^{\prime}\rightarrow A^{\prime
n}}$ and a decoding $\mathcal{D}_{B^{n}\rightarrow\hat{M}}$. We say that a
triple $\left(  n,R,\varepsilon\right)  $ is achievable if $\frac{1}{n}%
\log\left\vert M\right\vert \geq R$ and%
\begin{equation}
F(\Phi_{M\hat{M}},(\mathcal{D}_{B^{n}\rightarrow\hat{M}}\circ\mathcal{N}%
_{A^{\prime}\rightarrow B}^{\otimes n}\circ\mathcal{E}_{M^{\prime}\rightarrow
A^{\prime n}})(\Phi_{MM^{\prime}}))\geq1-\varepsilon.
\end{equation}
Since $\Phi_{M\hat{M}}$ is a particular kind of private state with trivial
twisting unitary and trivial shield systems (called \textquotedblleft abelian
twisting\textquotedblright\ in \cite[Section~VI-A]{HHHO09}), such a code is
achievable for secret-key transmission in the bipartite picture of privacy.
Due to the relation between the bipartite picture of privacy and the
tripartite picture \cite[Theorem~5]{HHHO09}, the triple $\left(
n,R,\varepsilon\right)  $ is achievable in the tripartite picture as well. As
a consequence, we can always give lower bounds on secret-key transmission
rates in terms of entanglement transmission rates.

We can extend all of the various definitions given previously for secret-key
transmission to the case of entanglement transmission (see \cite{TBR15}\ for
details), and this gives us all of the quantities from the previous section,
with $P$ replaced by $Q$. The relation of the non-asymptotic $Q$ quantities to
privacy is summarized as follows:%
\begin{equation}
\hat{Q}_{\mathcal{N}}(n,\varepsilon)\leq\hat{P}_{\mathcal{N}}(n,\varepsilon
),\ \ \ \ \ \ \hat{Q}_{\mathcal{N}}^{\operatorname{cppp}}(n,\varepsilon
)\leq\hat{P}_{\mathcal{N}}^{\operatorname{cppp}}(n,\varepsilon
),\ \ \ \ \ \ \hat{Q}_{\mathcal{N}}^{\leftrightarrow}(n,\varepsilon)\leq
\hat{P}_{\mathcal{N}}^{\leftrightarrow}(n,\varepsilon).
\label{eq:finite-Q-to-finite-P-1}%
\end{equation}

\section{General (meta-converse) bounds}

\label{sec:meta-converse}

\subsection{Information measures for the general (meta-converse) bounds}

\label{sec:one-shot-hypo-quantities}The general meta-converse bound in
Section~\ref{sec:meta-converse-statements} is given in terms of the following
quantity, defined for $\rho\in\mathcal{D}(\mathcal{H})$, $\sigma\in
\mathcal{L}_{+}(\mathcal{H})$, and $\varepsilon\in\lbrack0,1]$ as%
\begin{equation}
D_{H}^{\varepsilon}(\rho\Vert\sigma)\equiv-\log\left[  \inf\{\operatorname{Tr}%
\{\Lambda\sigma\}:0\leq\Lambda\leq I\wedge\operatorname{Tr}\{\Lambda\rho
\}\geq1-\varepsilon\}\right]  .
\end{equation}
If $\sigma$ is a quantum state, $D_{H}^{\varepsilon}(\rho\Vert\sigma)$ has an
interpretation as the optimal exponent of the Type~II error in a hypothesis
test to distinguish $\rho$ from $\sigma$, given the constraint that the Type~I
error should not exceed $\varepsilon$ \cite{HP91}. This quantity was used
effectively in one-shot quantum information theory \cite{BD10,BD11,WR12} and
given the name \textquotedblleft hypothesis testing relative
entropy\textquotedblright\ in \cite{WR12}. The hypothesis testing relative
entropy is monotone non-increasing with respect to quantum channels as well
(see, e.g., \cite{WR12}). By inspecting the definition, one can see that the
following statement holds for $\varepsilon\in\lbrack0,1)$%
\begin{equation}
\rho=\sigma\ \ \ \ \ \ \Rightarrow\ \ \ \ \ \ D_{H}^{\varepsilon}(\rho
\Vert\sigma)=-\log\left(  1-\varepsilon\right)
.\label{eq:D_H-rho-sigma-equal}%
\end{equation}
That is, the conditions $\rho=\sigma$ and $\operatorname{Tr}\{\Lambda
\rho\}\geq1-\varepsilon$ imply that $\operatorname{Tr}\{\Lambda\sigma
\}\geq1-\varepsilon$ for all $\Lambda$, and we can take $\Lambda
=(1-\varepsilon)I$ to achieve this bound.

From this quantity follows an information measure \cite[Definition~4]{BD11}
closely related to the relative entropy of entanglement \cite{VP98}:%
\begin{equation}
E_{R}^{\varepsilon}(A;B)_{\rho}\equiv\inf_{\sigma_{AB}\in\mathcal{S}%
(A:B)}D_{H}^{\varepsilon}(\rho_{AB}\Vert\sigma_{AB}).
\label{eq:hypo-rel-ent-ent}%
\end{equation}
This quantity is an LOCC\ monotone, meaning that%
\begin{equation}
E_{R}^{\varepsilon}(A;B)_{\rho}\geq E_{R}^{\varepsilon}(A^{\prime};B^{\prime
})_{\omega},
\end{equation}
for $\omega_{A^{\prime}B^{\prime}}\equiv\Lambda_{AB\rightarrow A^{\prime
}B^{\prime}}(\rho_{AB})$, with $\Lambda_{AB\rightarrow A^{\prime}B^{\prime}}$
an LOCC\ channel. This follows because the underlying quantity $D_{H}%
^{\varepsilon}$ is monotone non-increasing with respect to quantum channels
and the set of separable states is closed under LOCC\ channels (see
\cite[Lemma~1]{BD11} for an explicit proof). More generally, $E_{R}%
^{\varepsilon}(A;B)_{\rho}$ is monotone non-increasing with respect to
separability-preserving channels for the same reasons. We can extend the
definition in \eqref{eq:hypo-rel-ent-ent} to be a function of a quantum
channel~$\mathcal{N}_{A^{\prime}\rightarrow B}$:%
\begin{equation}
E_{R}^{\varepsilon}(\mathcal{N})\equiv\sup_{|\psi\rangle_{AA^{\prime}}%
\in\mathcal{H}_{AA^{\prime}}}E_{R}^{\varepsilon}(A;B)_{\rho},
\end{equation}
where $\rho_{AB}\equiv\mathcal{N}_{A^{\prime}\rightarrow B}(\psi_{AA^{\prime}%
})$. Note that it suffices to perform the optimization with respect to pure
states due to the fact that $D_{H}^{\varepsilon}$ satisfies the data
processing inequality. The quantity $E_{R}^{\varepsilon}(\mathcal{N})$ (and
later related ones) will play an important role in establishing upper bounds
on the private transmission capabilities of a quantum channel.

\subsection{Privacy test}

Here we define a \textquotedblleft privacy test\textquotedblright\ as a method
for testing whether a given bipartite state is private. It forms an essential
component of the general meta-converse bound given in
Section~\ref{sec:meta-converse-statements}. In some sense, this notion is
already implicit in the developments of \cite[Eqns.~(282)--(284)]{HHHO09} and
is stated even more explicitly in \cite{PhysRevLett.100.110502,HHHLO08}. We
state the notion here concretely for completeness.

\begin{definition}
[Privacy test]\label{def:privacy-test}Let $\gamma_{ABA^{\prime}B^{\prime}}%
\in\mathcal{D}(\mathcal{H}_{ABA^{\prime}B^{\prime}})$ be a bipartite private
state as given in Definition~\ref{def:private-state-bi}. A privacy test
corresponding to $\gamma_{ABA^{\prime}B^{\prime}}$ (a $\gamma$-privacy
test)\ is defined as the following dichotomic measurement:%
\begin{equation}
\left\{  \Pi_{ABA^{\prime}B^{\prime}},I_{ABA^{\prime}B^{\prime}}%
-\Pi_{ABA^{\prime}B^{\prime}}\right\}  ,
\end{equation}
where $\Pi_{ABA^{\prime}B^{\prime}}\equiv U_{ABA^{\prime}B^{\prime}}\left(
\Phi_{AB}\otimes I_{A^{\prime}B^{\prime}}\right)  U_{ABA^{\prime}B^{\prime}%
}^{\dag}$ and $U_{ABA^{\prime}B^{\prime}}$ is the unitary specified in \eqref{eq:twisting-unitary}.
\end{definition}

If one has access to the systems $ABA^{\prime}B^{\prime}$ of a bipartite state
$\rho_{ABA^{\prime}B^{\prime}}$ and has a description of $\gamma_{ABA^{\prime
}B^{\prime}}$ satisfying \eqref{eq:approx-bi-priv-state}, then the $\gamma
$-privacy test decides whether $\rho_{ABA^{\prime}B^{\prime}}$ is a private
state with respect to $\gamma_{ABA^{\prime}B^{\prime}}$. The first outcome
corresponds to the decision \textquotedblleft yes, it is a $\gamma$-private
state,\textquotedblright\ and the second outcome corresponds to
\textquotedblleft no.\textquotedblright\ Physically, this test is just
untwisting the purported private state and projecting onto a maximally
entangled state. The following lemma states that the probability for an
$\varepsilon$-approximate bipartite private state to pass the $\gamma$-privacy
test is high:

\begin{lemma}
\label{lem:pass-privacy-test}Let $\varepsilon\in\left[  0,1\right]  $ and let
$\rho_{ABA^{\prime}B^{\prime}}\in\mathcal{D}(\mathcal{H}_{ABA^{\prime
}B^{\prime}})$ be an $\varepsilon$-approximate private state as given in
Definition~\ref{def:approx-priv}, with $\gamma_{ABA^{\prime}B^{\prime}}$
satisfying \eqref{eq:approx-bi-priv-state}. The probability for $\rho
_{ABA^{\prime}B^{\prime}}$ to pass the $\gamma$-privacy test is never smaller
than $1-\varepsilon$:%
\begin{equation}
\operatorname{Tr}\{\Pi_{ABA^{\prime}B^{\prime}}\rho_{ABA^{\prime}B^{\prime}%
}\}\geq1-\varepsilon, \label{eq:priv-state-pass-test}%
\end{equation}
where $\Pi_{ABA^{\prime}B^{\prime}}$ is defined as above.
\end{lemma}

\begin{proof}
One can see this bound explicitly by inspecting the following steps:
\begin{align}
\operatorname{Tr}\{\Pi_{ABA^{\prime}B^{\prime}}\rho_{ABA^{\prime}B^{\prime}%
}\}  &  =\langle\Phi|_{AB}\operatorname{Tr}_{A^{\prime}B^{\prime}%
}\{U_{ABA^{\prime}B^{\prime}}^{\dag}\rho_{ABA^{\prime}B^{\prime}%
}U_{ABA^{\prime}B^{\prime}}\}|\Phi\rangle_{AB}\\
&  =F(\Phi_{AB},\operatorname{Tr}_{A^{\prime}B^{\prime}}\{U_{ABA^{\prime
}B^{\prime}}^{\dag}\rho_{ABA^{\prime}B^{\prime}}U_{ABA^{\prime}B^{\prime}%
}\})\\
&  \geq F(\Phi_{AB}\otimes\theta_{A^{\prime}B^{\prime}},U_{ABA^{\prime
}B^{\prime}}^{\dag}\rho_{ABA^{\prime}B^{\prime}}U_{ABA^{\prime}B^{\prime}})\\
&  =F(U_{ABA^{\prime}B^{\prime}}(\Phi_{AB}\otimes\theta_{A^{\prime}B^{\prime}%
})U_{ABA^{\prime}B^{\prime}}^{\dag},\rho_{ABA^{\prime}B^{\prime}})\\
&  =F(\gamma_{ABA^{\prime}B^{\prime}},\rho_{ABA^{\prime}B^{\prime}}%
)\geq1-\varepsilon.
\end{align}
The steps follow as a consequence of several properties of the fidelity
recalled in Section~\ref{sec:prelim}.
\end{proof}

\medskip

For completeness, we think it is worthwhile to recall the brief proof of
\eqref{eq:priv-test-separable} below from \cite[Eqns.~(282)--(284)]{HHHO09}.

\begin{lemma}
[{\cite[Eqn.~(281)]{HHHO09}}]\label{lem:fail-privacy-test}For a separable
state $\sigma_{ABA^{\prime}B^{\prime}}\in\mathcal{S}(AA^{\prime}%
\!:\!BB^{\prime})$, the probability of passing any $\gamma$-privacy test is
never larger than $1/K$:%
\begin{equation}
\operatorname{Tr}\{\Pi_{ABA^{\prime}B^{\prime}}\sigma_{ABA^{\prime}B^{\prime}%
}\}\leq\frac{1}{K}\ , \label{eq:priv-test-separable}%
\end{equation}
where $K$ is the number of values that the secret key can take (i.e.,
$K=\dim(\mathcal{H}_{A})=\dim(\mathcal{H}_{B})$).
\end{lemma}

\begin{proof}
The idea is to begin by establishing the bound for any pure product state
$|\phi\rangle_{AA^{\prime}}\otimes|\varphi\rangle_{BB^{\prime}}$. We can
expand these states with respect to the standard bases of $A$ and $B$ as
follows:%
\begin{equation}
|\phi\rangle_{AA^{\prime}}\otimes|\varphi\rangle_{BB^{\prime}}=\left[
\sum_{i=0}^{K-1}\alpha_{i}|i\rangle_{A}\otimes|\phi_{i}\rangle_{A^{\prime}%
}\right]  \otimes\left[  \sum_{j=0}^{K-1}\beta_{j}|j\rangle_{B}\otimes
|\varphi_{j}\rangle_{B^{\prime}}\right]  ,
\end{equation}
where $\sum_{i=0}^{K-1}\left\vert \alpha_{i}\right\vert ^{2}=\sum_{j=0}%
^{K-1}\left\vert \beta_{j}\right\vert ^{2}=1$. A few steps of calculation then
lead to the following equalities:%
\begin{align}
&  \!\!\!\!\!\!\operatorname{Tr}\{\Pi_{ABA^{\prime}B^{\prime}}|\phi
\rangle\langle\phi|_{AA^{\prime}}\otimes|\varphi\rangle\langle\varphi
|_{BB^{\prime}}\}\nonumber\\
&  =\operatorname{Tr}\{U_{ABA^{\prime}B^{\prime}}\left(  \Phi_{AB}\otimes
I_{A^{\prime}B^{\prime}}\right)  U_{ABA^{\prime}B^{\prime}}^{\dag}|\phi
\rangle\langle\phi|_{AA^{\prime}}\otimes|\varphi\rangle\langle\varphi
|_{BB^{\prime}}\}\\
&  =\frac{1}{K}\sum_{i,j=0}^{K-1}\alpha_{i}\beta_{i}\alpha_{j}^{\ast}\beta
_{j}^{\ast}\langle\xi_{j}|\xi_{i}\rangle_{A^{\prime}B^{\prime}},
\end{align}
where $|\xi_{i}\rangle_{A^{\prime}B^{\prime}}\equiv(U_{A^{\prime}B^{\prime}%
}^{ii})^{\dag}|\phi_{i}\rangle_{A^{\prime}}|\varphi_{i}\rangle_{B^{\prime}}$
is a quantum state. The desired bound is then equivalent to%
\begin{equation}
\sum_{i,j=0}^{K-1}\alpha_{i}\beta_{i}\alpha_{j}^{\ast}\beta_{j}^{\ast}%
\langle\xi_{j}|\xi_{i}\rangle_{A^{\prime}B^{\prime}}\leq1.
\end{equation}
Setting $\alpha_{i}=\sqrt{p_{i}}e^{i\theta_{i}}$ and $\beta_{i}=\sqrt{q_{i}%
}e^{i\eta_{i}}$, we find that%
\begin{align}
\sum_{i,j=0}^{K-1}\alpha_{i}\beta_{i}\alpha_{j}^{\ast}\beta_{j}^{\ast}%
\langle\xi_{j}|\xi_{i}\rangle_{A^{\prime}B^{\prime}}  &  =\left\vert
\sum_{i,j=0}^{K-1}\sqrt{p_{i}q_{i}p_{j}q_{j}}e^{i\left(  \theta_{i}+\eta
_{i}-\theta_{j}-\eta_{j}\right)  }\langle\xi_{j}|\xi_{i}\rangle_{A^{\prime
}B^{\prime}}\right\vert \\
&  \leq\sum_{i,j=0}^{K-1}\sqrt{p_{i}q_{i}p_{j}q_{j}}\left\vert \langle\xi
_{j}|\xi_{i}\rangle_{A^{\prime}B^{\prime}}\right\vert \leq\sum_{i,j=0}%
^{K-1}\sqrt{p_{i}q_{i}p_{j}q_{j}}\\
&  =\left[  \sum_{i=0}^{K-1}\sqrt{p_{i}q_{i}}\right]  ^{2}\leq1,
\end{align}
where the last inequality holds for all probability distributions (this is
just the statement that the classical fidelity cannot exceed one). The above
reasoning thus establishes \eqref{eq:priv-test-separable} for pure product
states, and the bound for general separable states follows because every such
state can be written as a convex combination of pure product states.
\end{proof}

\bigskip

The bounds in \eqref{eq:priv-state-pass-test} and
\eqref{eq:priv-test-separable} are the core ones underlying all of our
converse bounds in this paper.

\subsection{Statements of general (meta-converse) bounds}

\label{sec:meta-converse-statements}

We now establish some general bounds on the achievable regions discussed in
Section~\ref{sec:non-asym-region}.

\begin{theorem}
\label{prop:meta-conv-priv} Let $\mathcal{N}_{A^{\prime}\rightarrow B}$ be a
quantum channel. Then for any fixed $\varepsilon\in\left(  0,1\right)  $, the
achievable region with CPPP assistance satisfies%
\begin{equation}
\hat{P}_{\mathcal{N}}^{\operatorname{cppp}}(1,\varepsilon)\leq E_{R}%
^{\varepsilon}(\mathcal{N}).
\end{equation}

\end{theorem}

\begin{proof}
Consider any CPPP-assisted protocol that achieves a rate $\hat{P}%
_{\mathcal{N}}^{\operatorname{cppp}}(1,\varepsilon)\equiv\hat{P}$, formulated
in the bipartite picture as discussed in the previous section. Let
$\omega_{A_{0}A^{\prime}B_{0}}$ denote the state generated by the first round
of LOCC. Note that $\omega_{A_{0}A^{\prime}B_{0}}$ is a separable
state:\ $\omega_{A_{0}A^{\prime}B_{0}}\in\mathcal{S}(A_{0}A^{\prime}%
\!:\!B_{0})$. The $A^{\prime}$ system of this state gets sent through the
channel $\mathcal{N}_{A^{\prime}\rightarrow B}$, leading to the state%
\begin{equation}
\theta_{A_{0}BB_{0}}\equiv\mathcal{N}_{A^{\prime}\rightarrow B}(\omega
_{A_{0}A^{\prime}B_{0}}).
\end{equation}
Alice and Bob apply an LOCC decoder $\mathcal{D}_{A_{0}BB_{0}\rightarrow
K_{A}K_{B}S_{A}S_{B}}$, which consists of a round of LOCC, leading to the
state%
\begin{equation}
\omega_{K_{A}K_{B}S_{A}S_{B}}\equiv\mathcal{D}_{A_{0}BB_{0}\rightarrow
K_{A}K_{B}S_{A}S_{B}}(\theta_{A_{0}BB_{0}}).
\end{equation}
By assumption we have that%
\begin{equation}
F(\gamma_{K_{A}K_{B}S_{A}S_{B}},\omega_{K_{A}K_{B}S_{A}S_{B}})\geq
1-\varepsilon,
\end{equation}
for some private state $\gamma_{K_{A}K_{B}S_{A}S_{B}}$. By
Lemma~\ref{lem:pass-privacy-test}, there is a projector $\Pi_{K_{A}K_{B}%
S_{A}S_{B}}$ corresponding to a $\gamma$-privacy test of the form in
Definition~\ref{def:privacy-test}, such that%
\begin{equation}
\operatorname{Tr}\{\Pi_{K_{A}K_{B}S_{A}S_{B}}\omega_{K_{A}K_{B}S_{A}S_{B}%
}\}\geq1-\varepsilon.
\end{equation}
From Lemma~\ref{lem:fail-privacy-test}, we have that%
\begin{equation}
\operatorname{Tr}\{\Pi_{K_{A}K_{B}S_{A}S_{B}}\sigma_{K_{A}K_{B}S_{A}S_{B}%
}\}\leq2^{-\hat{P}},
\end{equation}
for any separable state $\sigma_{K_{A}K_{B}S_{A}S_{B}}\in\mathcal{S}%
(K_{A}S_{A}\!:\!K_{B}S_{B})$. Thus, this test is feasible for $D_{H}%
^{\varepsilon}(\omega\Vert\sigma)$ and we find that%
\begin{equation}
\hat{P}\leq D_{H}^{\varepsilon}(\omega_{K_{A}K_{B}S_{A}S_{B}}\Vert
\sigma_{K_{A}K_{B}S_{A}S_{B}})\label{eq:critical-bound}%
\end{equation}
for any separable state $\sigma_{K_{A}K_{B}S_{A}S_{B}}\in\mathcal{S}%
(K_{A}S_{A}\!:\!K_{B}S_{B})$. Let $\tau_{A_{0}B}\in\mathcal{S}(A_{0}\!:\!B)$.
From the quasi-convexity of $D_{H}^{\varepsilon}$ we find that there exist
pure states $\psi_{A_{0}A^{\prime}}$ and $\varphi_{B_{0}}$ such that%
\begin{align}
D_{H}^{\varepsilon}(\mathcal{N}_{A^{\prime}\rightarrow B}(\psi_{A_{0}%
A^{\prime}})\Vert\tau_{A_{0}B}) &  =D_{H}^{\varepsilon}(\mathcal{N}%
_{A^{\prime}\rightarrow B}(\psi_{A_{0}A^{\prime}})\otimes\varphi_{B_{0}}%
\Vert\tau_{A_{0}B}\otimes\varphi_{B_{0}})\\
&  \geq D_{H}^{\varepsilon}(\mathcal{N}_{A^{\prime}\rightarrow B}%
(\omega_{A_{0}A^{\prime}B_{0}})\Vert\tau_{A_{0}B}\otimes\varphi_{B_{0}})\\
&  \geq D_{H}^{\varepsilon}(\omega_{K_{A}K_{B}S_{A}S_{B}}\Vert\sigma
_{K_{A}K_{B}S_{A}S_{B}})\\
&  \geq\hat{P},
\end{align}
where we take $\sigma_{K_{A}K_{B}S_{A}S_{B}}=\mathcal{D}_{A_{0}BB_{0}%
\rightarrow K_{A}K_{B}S_{A}S_{B}}(\tau_{A_{0}B}\otimes\varphi_{B_{0}})$. The
first equality follows because $D_{H}^{\varepsilon}$ is invariant with respect
to tensoring in the same state on an extra system (doing so does not change
the constrained Type~II\ error in a quantum hypothesis test). The first
inequality follows from quasi-convexity of $D_{H}^{\varepsilon}$. The second
inequality follows from the monotonicity of $D_{H}^{\varepsilon}$ with respect
to quantum channels. Since the decoder $\mathcal{D}_{A_{0}BB_{0}\rightarrow
K_{A}K_{B}S_{A}S_{B}}$ is an LOCC\ channel, we can conclude that
$\sigma_{K_{A}K_{B}S_{A}S_{B}}\in\mathcal{S}(K_{A}S_{A}\!:\!K_{B}S_{B})$. The
final inequality follows from \eqref{eq:critical-bound}. Since the inequality
holds for any choice $\tau_{A_{0}B}\in\mathcal{S}(A_{0}\!:\!B)$, we can
conclude that%
\begin{equation}
E_{R}^{\varepsilon}(\mathcal{N}_{A^{\prime}\rightarrow B}(\psi_{A_{0}%
A^{\prime}}))\geq\hat{P}.
\end{equation}
Optimizing over all input states $\psi_{A_{0}A^{\prime}}$, we can conclude the
statement of the proposition.
\end{proof}

The above theorem immediately leads to the following bound for any quantum
channel $\mathcal{N}$:%
\begin{equation}
\hat{P}_{\mathcal{N}}^{\operatorname{cppp}}(n,\varepsilon)\leq\frac{1}{n}%
E_{R}^{\varepsilon}(\mathcal{N}^{\otimes n}%
).\label{eq:bound-rel-ent-ent-any-channel-cppp}%
\end{equation}
For a teleportation-simulable channel $\mathcal{N}_{A^{\prime}\rightarrow
B}^{\operatorname{TP}}$ as defined in \eqref{eq:TP-covariant} with associated
state $\omega_{AB}$, combining \eqref{eq:SKA-TPC-protocols}\ and a proof
similar to that for Theorem~\ref{prop:meta-conv-priv}\ leads to the following
bound:%
\begin{equation}
\hat{P}_{\mathcal{N}}^{\leftrightarrow}(n,\varepsilon)\leq\frac{1}{n}%
E_{R}^{\varepsilon}(A^{n};B^{n})_{\omega^{\otimes n}}%
.\label{eq:two-way-meta-conv}%
\end{equation}
Note that this latter bound will lead us to a complete proof of the main
result presented in \cite{PLOB15},
in addition to other more refined statements. In contrast to the approach presented in \cite{PLOB15}, our result has no dependence on the dimension of the shield systems. This is in particular beneficial for the treatment of quantum Gaussian channels (discussed in Section~\ref{sec:gaussian}).

\section{Relative entropy of entanglement as a strong converse rate}

\label{sec:strong-converse}In this section, we prove two strong converse
theorems for private communication. Before doing so, we review various
definitions of private capacities and strong converse rates for private
communication, and we also review several R\'{e}nyi entropic measures that
play a role in establishing the strong converse theorems.

\subsection{Definitions of private capacities and strong converse rates}

A rate $r$ is achievable for secret-key transmission over the channel
$\mathcal{N}$ if there exists a sequence of secret-key transmission protocols
$\{(n,P_{n},\varepsilon_{n})\}_{n\in\mathbb{N}}$, such that%
\begin{equation}
\liminf_{n\rightarrow\infty}P_{n}\geq r\qquad\text{and}\qquad\lim
_{n\rightarrow\infty}\varepsilon_{n}=0.
\end{equation}
The private capacity of $\mathcal{N}$, denoted $P(\mathcal{N})$, is equal to
the supremum of all achievable rates \cite{ieee2005dev,1050633}. Equivalently,
we have that%
\begin{equation}
P(\mathcal{N})=\lim_{\varepsilon\rightarrow0}\liminf_{n\rightarrow\infty}%
\hat{P}_{\mathcal{N}}(n,\varepsilon).
\end{equation}
Analogously, the CPPP-assisted private capacity of $\mathcal{N}$, denoted
$P_{\operatorname{cppp}}(\mathcal{N})$, is equal to the supremum of all
CPPP-assisted achievable rates, and we have a similar definition for
$P_{\leftrightarrow}(\mathcal{N})$. Similarly,%
\begin{align}
P_{\operatorname{cppp}}(\mathcal{N})  &  =\lim_{\varepsilon\rightarrow
0}\liminf_{n\rightarrow\infty}\hat{P}_{\mathcal{N}}^{\operatorname{cppp}%
}(n,\varepsilon),\\
P_{\leftrightarrow}(\mathcal{N})  &  =\lim_{\varepsilon\rightarrow0}%
\liminf_{n\rightarrow\infty}\hat{P}_{\mathcal{N}}^{\leftrightarrow
}(n,\varepsilon).
\end{align}

On the other hand, $r$ is a strong converse rate for secret-key transmission
if for every sequence of secret-key transmission protocols $\{(n,P_{n}%
,\varepsilon_{n})\}_{n\in\mathbb{N}}$ as above, we have%
\begin{equation}
\liminf_{n\rightarrow\infty}P_{n}>r\qquad\implies\qquad\lim_{n\rightarrow
\infty}\varepsilon_{n}=1.
\end{equation}
The strong converse private capacity, denoted $P^{\dagger}(\mathcal{N})$, is
equal to the infimum of all strong converse rates. Analogously, the
CPPP-assisted strong converse private capacity of $\mathcal{N}$, denoted
$P_{\operatorname{cppp}}^{\dagger}(\mathcal{N})$, is equal to the infimum of
all CPPP-assisted strong converse rates (and similarly for $P_{\leftrightarrow
}^{\dag}(\mathcal{N})$). The following inequalities hold by definition:%
\begin{align}
P(\mathcal{N})  &  \leq P_{\operatorname{cppp}}(\mathcal{N})\leq
P_{\leftrightarrow}(\mathcal{N}),\\
P^{\dagger}(\mathcal{N})  &  \leq P_{\operatorname{cppp}}^{\dagger
}(\mathcal{N})\leq P_{\leftrightarrow}^{\dagger}(\mathcal{N}%
),\label{eq:strong-converse-ordering}\\
P(\mathcal{N})  &  \leq P^{\dagger}(\mathcal{N}),\\
P_{\operatorname{cppp}}(\mathcal{N})  &  \leq P_{\operatorname{cppp}}^{\dag
}(\mathcal{N}),\\
P_{\leftrightarrow}(\mathcal{N})  &  \leq P_{\leftrightarrow}^{\dagger
}(\mathcal{N}).
\end{align}
We argue in Appendix~\ref{app:priv-def-str-conv}\ how a strong converse rate
according to the above definitions of privacy is a strong converse rate
according to quantum generalizations of the definitions from \cite{HTW14} (see
Appendix~\ref{app:priv-def-str-conv} for specifics).

Finally, we say that a channel $\mathcal{N}$ satisfies the strong converse
property for private communication if $P(\mathcal{N})=P^{\dagger}%
(\mathcal{N})$. Similarly, we say that a channel $\mathcal{N}$ satisfies the
strong converse property for CPPP-assisted private communication if
$P_{\operatorname{cppp}}(\mathcal{N})=P_{\operatorname{cppp}}^{\dagger
}(\mathcal{N})$, and a similar statement if $P_{\leftrightarrow}%
(\mathcal{N})=P_{\leftrightarrow}^{\dagger}(\mathcal{N})$.

For the capacities, we find that%
\begin{align}
Q(\mathcal{N}) &  \leq P(\mathcal{N}),\ \ \ \ \ \ Q_{\operatorname{cppp}%
}(\mathcal{N})\leq P_{\operatorname{cppp}}(\mathcal{N}%
),\ \ \ \ \ \ Q_{\leftrightarrow}(\mathcal{N})\leq P_{\leftrightarrow
}(\mathcal{N}),\\
Q^{\dag}(\mathcal{N}) &  \leq P^{\dagger}(\mathcal{N}%
),\ \ \ \ \ \ Q_{\operatorname{cppp}}^{\dag}(\mathcal{N})\leq
P_{\operatorname{cppp}}^{\dagger}(\mathcal{N}),\ \ \ \ \ \ Q_{\leftrightarrow
}^{\dag}(\mathcal{N})\leq P_{\leftrightarrow}^{\dagger}(\mathcal{N}),\\
Q(\mathcal{N}) &  \leq Q^{\dagger}(\mathcal{N}%
),\ \ \ \ \ \ Q_{\operatorname{cppp}}(\mathcal{N})\leq Q_{\operatorname{cppp}%
}^{\dag}(\mathcal{N}),\ \ \ \ \ \ Q_{\leftrightarrow}(\mathcal{N})\leq
Q_{\leftrightarrow}^{\dagger}(\mathcal{N}),
\end{align}
where the quantum capacity $Q$ and the strong converse quantum capacity
$Q^{\dagger}$ are defined analogously to $P$ and $P^{\dagger}$. In summary,
any lower bound on a rate of entanglement transmission for a given scenario is
a lower bound for secret-key transmission in the same scenario.

\subsection{R\'{e}nyi relative entropies and related measures}

Let $\rho\in\mathcal{D}(\mathcal{H})$ and $\sigma\in\mathcal{L}_{+}%
(\mathcal{H})$. The quantum relative entropy $D(\rho\Vert\sigma)$ is defined
as \cite{U62}%
\begin{equation}
D(\rho\Vert\sigma)\equiv\left\{
\begin{array}
[c]{cc}%
\operatorname{Tr}\{\rho(\log\rho-\log\sigma)\} & \text{if }\operatorname{supp}%
(\rho)\subseteq\operatorname{supp}(\sigma)\\
+\infty & \text{else}%
\end{array}
\right.  .
\end{equation}
Throughout we take the logarithm (denoted by $\log$) to be base two unless
stated otherwise. The relative entropy $D(\rho\Vert\sigma)$ is monotone with
respect to quantum channels \cite{Lindblad1975,U77}, in the sense that%
\begin{equation}
D(\rho\Vert\sigma)\geq D(\mathcal{N}(\rho)\Vert\mathcal{N}(\sigma)),
\end{equation}
for $\mathcal{N}$ a quantum channel.

The sandwiched R\'{e}nyi relative entropy is defined for $\alpha\in
(0,1)\cup(1,\infty)$ as \cite{MDSFT13,WWY13}:%
\begin{equation}
\widetilde{D}_{\alpha}(\rho\Vert\sigma)\equiv\left\{
\begin{array}
[c]{cc}%
\frac{2\alpha}{\alpha-1}\log\left\Vert \sigma^{\left(  1-\alpha\right)
/2\alpha}\rho^{1/2}\right\Vert _{2\alpha} &
\begin{array}
[c]{c}%
\text{if }\left(  \alpha\in(0,1)\wedge\operatorname{supp}(\rho)\not \perp
\operatorname{supp}(\sigma)\right) \\
\vee\left(  \operatorname{supp}(\rho)\subseteq\operatorname{supp}%
(\sigma)\right)
\end{array}
\\
+\infty & \text{else}%
\end{array}
\right.  ,
\end{equation}
where $\Vert A \Vert_{p} \equiv[\operatorname{Tr} \{ |A|^{p}\}]^{1/p}$ is the
$p$-norm of an operator $A$ for $p \geq1$ and $|A| \equiv\sqrt{A^{\dag}A}$
(note that we define $\Vert A \Vert_{p}$ as above even for $p \in[0,1)$ when
it is not a norm). Both the quantum and sandwiched relative entropies are
additive in the following sense:%
\begin{align}
D(\rho_{0}\otimes\rho_{1}\Vert\sigma_{0}\otimes\sigma_{1})  &  =D(\rho
_{0}\Vert\sigma_{0})+D(\rho_{1}\Vert\sigma_{1}%
),\label{eq:additivity-rel-ent-1}\\
\widetilde{D}_{\alpha}(\rho_{0}\otimes\rho_{1}\Vert\sigma_{0}\otimes\sigma
_{1})  &  =\widetilde{D}_{\alpha}(\rho_{0}\Vert\sigma_{0})+\widetilde
{D}_{\alpha}(\rho_{1}\Vert\sigma_{1}), \label{eq:additivity-rel-ent-2}%
\end{align}
where $\rho_{i}\in\mathcal{D}(\mathcal{H}_{i})$ and $\sigma_{i}\in
\mathcal{L}_{+}(\mathcal{H}_{i})$ for $i\in\left\{  0,1\right\}  $. The
following limits hold \cite{MDSFT13,WWY13}%
\begin{equation}
\lim_{\alpha\rightarrow1}\widetilde{D}_{\alpha}(\rho\Vert\sigma)=D(\rho
\Vert\sigma),\ \ \ \ \ \ \lim_{\alpha\rightarrow\infty}\widetilde{D}_{\alpha
}(\rho\Vert\sigma)=D_{\max}(\rho\Vert\sigma),
\end{equation}
where $D_{\max}(\rho\Vert\sigma)\equiv2\log\left\Vert \sigma^{-1/2}\rho
^{1/2}\right\Vert _{\infty}$ \cite{D09}. The quantity $\widetilde{D}_{\alpha}$
is monotone with respect to quantum channels \cite{FL13}, in the sense that%
\begin{equation}
\widetilde{D}_{\alpha}(\rho\Vert\sigma)\geq\widetilde{D}_{\alpha}%
(\mathcal{N}(\rho)\Vert\mathcal{N}(\sigma)),
\end{equation}
for $\mathcal{N}$ a quantum channel and $\alpha\in\lbrack1/2,1)\cup(1,\infty
]$. The quantity $\widetilde{D}_{\alpha}$ is also monotone with respect to the
R\'{e}nyi parameter \cite{MDSFT13,B13monotone}:\ for $1<\alpha<\beta$, the
following inequality holds%
\begin{equation}
D(\rho\Vert\sigma)\leq\widetilde{D}_{\alpha}(\rho\Vert\sigma)\leq\widetilde
{D}_{\beta}(\rho\Vert\sigma). \label{eq:renyi-param-mono}%
\end{equation}

The following inequality relates $D_{H}^{\varepsilon}(\rho\Vert\sigma)$ to
$\widetilde{D}_{\alpha}(\rho\Vert\sigma)$ for $\alpha\in(1,\infty)$ and
$\varepsilon\in(0,1)$:%
\begin{equation}
D_{H}^{\varepsilon}(\rho\Vert\sigma)\leq\widetilde{D}_{\alpha}(\rho\Vert
\sigma)+\frac{\alpha}{\alpha-1}\log\left(  1/\left(  1-\varepsilon\right)
\right)  . \label{dh-to-dalpha-1}%
\end{equation}
This inequality is implicit in the literature \cite{HP91,N01,ON00} (an
explicit proof for the interested reader is available as \cite[Lemma~5]{CMW14}).

The relative entropy of entanglement of a state $\rho_{AB}\in\mathcal{D}%
(\mathcal{H}_{AB})$ is defined as \cite{VP98}%
\begin{equation}
E_{R}(A;B)_{\rho}\equiv\min_{\sigma_{AB}\in\mathcal{S}(A:B)}D(\rho_{AB}%
\Vert\sigma_{AB}), \label{eq:rel-ent-ent-state}%
\end{equation}
and we define a related R\'{e}nyi relative entropy of entanglement as well:%
\begin{equation}
\widetilde{E}_{R,\alpha}(A;B)_{\rho}\equiv\inf_{\sigma_{AB}\in\mathcal{S}%
(A:B)}\widetilde{D}_{\alpha}(\rho_{AB}\Vert\sigma_{AB}).
\end{equation}
Note that an alternative definition of R\'{e}nyi relative entropy of
entanglement has already been given in \cite{Sh14}, in terms of the R\'{e}nyi
relative entropy defined in \cite{P86}. The relative entropies of entanglement
are LOCC\ monotones and more generally separability-preserving monotones, as
defined and justified previously in Section~\ref{sec:one-shot-hypo-quantities}%
.\ The following subadditivity relations hold:%
\begin{align}
E_{R}(A_{0}A_{1};B_{0}B_{1})_{\rho^{0}\otimes\rho^{1}}  &  \leq E_{R}%
(A_{0};B_{0})_{\rho^{0}}+E_{R}(A_{1};B_{1})_{\rho^{1}},\\
\widetilde{E}_{R,\alpha}(A_{0}A_{1};B_{0}B_{1})_{\rho^{0}\otimes\rho^{1}}  &
\leq\widetilde{E}_{R,\alpha}(A_{0};B_{0})_{\rho^{0}}+\widetilde{E}_{R,\alpha
}(A_{1};B_{1})_{\rho^{1}}, \label{eq:sub-add-renyi-rel-ent-ent}%
\end{align}
where $\rho_{A_{i}B_{i}}^{i}\in\mathcal{D}(\mathcal{H}_{A_{i}B_{i}})$ for
$i\in\left\{  0,1\right\}  $. These follow from the additivity relations in
\eqref{eq:additivity-rel-ent-1}--\eqref{eq:additivity-rel-ent-2} and because
the separable states in $\mathcal{S}(A_{0}A_{1}\!:\!B_{0}B_{1})$ considered
for the infima on the left-hand side need not be a tensor product. We extend
these definitions to be functions of a quantum channel $\mathcal{N}%
_{A^{\prime}\rightarrow B}$, which we call the channel's relative entropy of
entanglement:%
\begin{align}
E_{R}(\mathcal{N})  &  \equiv\sup_{|\psi\rangle_{AA^{\prime}}\in
\mathcal{H}_{AA^{\prime}}}E_{R}(A;B)_{\rho},\label{eq:rel-ent-ent-ch}\\
\widetilde{E}_{R,\alpha}(\mathcal{N})  &  \equiv\sup_{|\psi\rangle
_{AA^{\prime}}\in\mathcal{H}_{AA^{\prime}}}\widetilde{E}_{R,\alpha}%
(A;B)_{\rho},
\end{align}
where $\rho_{AB}\equiv\mathcal{N}_{A^{\prime}\rightarrow B}(\psi_{AA^{\prime}%
})$.

By a standard continuity argument (see, e.g.,
\cite{MH11,TWW14,CMW14,MO14,DW15}), the following limits hold%
\begin{align}
\lim_{\alpha\rightarrow1}\widetilde{E}_{R,\alpha}(A;B)_{\rho}  &
=E_{R}(A;B)_{\rho},\label{eq:renyi-rel-ent-ent-limit-1}\\
\lim_{\alpha\rightarrow1}\widetilde{E}_{R,\alpha}(\mathcal{N})  &
=E_{R}(\mathcal{N}). \label{eq:renyi-rel-ent-ent-limit-1-channel}%
\end{align}
Monotonicity of these quantities with respect to the R\'{e}nyi parameter
follows from \eqref{eq:renyi-param-mono}: for $1<\alpha<\beta$, the following
inequalities hold%
\begin{align}
E_{R}(A;B)_{\rho}  &  \leq\widetilde{E}_{R,\alpha}(A;B)_{\rho}\leq
\widetilde{E}_{R,\beta}(A;B)_{\rho},\label{eq:renyi-rel-ent-ent-mono}\\
E_{R}(\mathcal{N})  &  \leq\widetilde{E}_{R,\alpha}(\mathcal{N})\leq
\widetilde{E}_{R,\beta}(\mathcal{N}).
\label{eq:renyi-rel-ent-ent-mono-channel}%
\end{align}
The inequality in \eqref{dh-to-dalpha-1}\ allows us to relate $\widetilde
{E}_{R,\alpha}$ to $E_{R}^{\varepsilon}$ for $\alpha\in(1,\infty)$ and
$\varepsilon\in\left(  0,1\right)  $:%
\begin{align}
E_{R}^{\varepsilon}(A;B)_{\rho}  &  \leq\widetilde{E}_{R,\alpha}(A;B)_{\rho
}+\frac{\alpha}{\alpha-1}\log\left(  1/\left(  1-\varepsilon\right)  \right)
,\label{eq:rel-ent-hyp-to-rel-ent-ent-alpha-state}\\
E_{R}^{\varepsilon}(\mathcal{N})  &  \leq\widetilde{E}_{R,\alpha}%
(\mathcal{N})+\frac{\alpha}{\alpha-1}\log\left(  1/\left(  1-\varepsilon
\right)  \right)  . \label{eq:rel-ent-hyp-to-rel-ent-ent-alpha-channel}%
\end{align}
These latter two inequalities are helpful for obtaining the strong converse
theorems given below.

\subsection{Statements of strong converse results}

We begin by establishing the following strong converse theorem for
teleportation-simulable channels:

\begin{theorem}
\label{thm:tp-cov-str-conv}If a channel $\mathcal{N}^{\operatorname{TP}}$ is
teleportation-simulable as defined in \eqref{eq:TP-covariant} (with associated
state $\omega_{AB}$), then $E_{R}(A;B)_{\omega}$ is a strong converse rate for
two-way assisted private communication:%
\begin{equation}
P_{\leftrightarrow}^{\dagger}(\mathcal{N}^{\operatorname{TP}})\leq
E_{R}(A;B)_{\omega}.
\end{equation}

\end{theorem}

\begin{proof}
Let $\alpha\in(1,\infty)$. A consequence of \eqref{eq:two-way-meta-conv} and a
rewriting of \eqref{eq:rel-ent-hyp-to-rel-ent-ent-alpha-state} is the
following bound on the optimal fidelity of any two-way assisted protocol for a
teleportation-simulable channel$~\mathcal{N}^{\operatorname{TP}}$:%
\begin{align}
1-\hat{\varepsilon}_{\mathcal{N}}^{\leftrightarrow}(n,P^{\leftrightarrow})  &
\leq2^{-n\left(  \frac{\alpha-1}{\alpha}\right)  \left(  P^{\leftrightarrow
}-\frac{1}{n}\widetilde{E}_{R,\alpha}(A^{n};B^{n})_{\omega^{\otimes n}%
}\right)  }\\
&  \leq2^{-n\left(  \frac{\alpha-1}{\alpha}\right)  \left(  P^{\leftrightarrow
}-\widetilde{E}_{R,\alpha}(A;B)_{\omega}\right)  },
\end{align}
where the second inequality follows from \eqref{eq:sub-add-renyi-rel-ent-ent},
i.e, the subadditivity of $\widetilde{E}_{R,\alpha}$ with respect to
tensor-product states. Thus, if $P^{\leftrightarrow}>E_{R}(A;B)_{\omega}$,
then by \eqref{eq:renyi-rel-ent-ent-mono} and
\eqref{eq:renyi-rel-ent-ent-limit-1}, there exists $\alpha>1$ such that
$P^{\leftrightarrow}>\widetilde{E}_{R,\alpha}(A;B)_{\omega}$ and so the
optimal error $\hat{\varepsilon}_{\mathcal{N}}^{\leftrightarrow}%
(n,P^{\leftrightarrow})$ increases exponentially fast to one with exponent
$\left(  \frac{\alpha-1}{\alpha}\right)  (P^{\leftrightarrow}-\widetilde
{E}_{R,\alpha}(A;B)_{\omega})$.
\end{proof}

Next we establish that a channel's relative entropy of entanglement from
\eqref{eq:rel-ent-ent-ch} is a strong converse rate for CPPP-assisted private
communication (and thus for unassisted private communication as well by \eqref{eq:strong-converse-ordering}).

\begin{theorem}
\label{thm:rel-ent-ent-strong-conv-cppp}For any channel $\mathcal{N}$, its
relative entropy of entanglement is a strong converse rate for CPPP-assisted
private communication:%
\begin{equation}
P_{\operatorname{cppp}}^{\dagger}(\mathcal{N})\leq E_{R}(\mathcal{N}).
\end{equation}

\end{theorem}

We do not give a detailed proof of the theorem above, because it follows from
several results already available in \cite{TWW14}. Here we merely collect the
needed statements and give a proof sketch. A proof for
Proposition~\ref{prop:covariance} below follows by the same proof given for
\cite[Proposition~2]{TWW14}:

\begin{proposition}
\label{prop:covariance} Let $\mathcal{N}_{A^{\prime}\rightarrow B}$ be a
quantum channel that is covariant with respect to a group $G$ (as defined in
Section~\ref{sec:prelim})\ and let $\rho_{A^{\prime}}\in\mathcal{D}%
(\mathcal{H}_{A^{\prime}})$, $\phi_{AA^{\prime}}^{\rho} \in\mathcal{H}_{A}
\otimes\mathcal{H}_{A^{\prime}}$ be a purification of $\rho_{A^{\prime}}$, and
$\rho_{AB}=\mathcal{N}_{A^{\prime}\rightarrow B}(\phi_{AA^{\prime}}^{\rho})$.
Let $\bar{\rho}_{A^{\prime}}$ denote the group expectation of $\rho
_{A^{\prime}}$, i.e.,%
\begin{equation}
\bar{\rho}_{A^{\prime}}=\frac{1}{\left\vert G\right\vert }\sum_{g}%
U_{A^{\prime}}(g)\rho_{A^{\prime}}U_{A^{\prime}}^{\dag}(g),
\end{equation}
and let $\phi_{AA^{\prime}}^{\bar{\rho}}$ be a purification of $\bar{\rho
}_{A^{\prime}}$ and $\bar{\rho}_{AB}=\mathcal{N}_{A^{\prime}\rightarrow
B}(\phi_{AA^{\prime}}^{\bar{\rho}})$. Then, for $\alpha\in\left(
1,\infty\right)  $ and $\varepsilon\in\left(  0,1\right)  $, the following
inequalities hold%
\begin{align}
E_{R}(A;B)_{\bar{\rho}}  &  \geq E_{R}(A;B)_{{\rho}},\\
E_{R}^{\varepsilon}(A;B)_{\bar{\rho}}  &  \geq E_{R}^{\varepsilon}%
(A;B)_{{\rho}},\\
\widetilde{E}_{R,\alpha}(A;B)_{\bar{\rho}}  &  \geq\widetilde{E}_{R,\alpha
}(A;B)_{{\rho}}.
\end{align}

\end{proposition}

The following theorem is a consequence of Proposition~\ref{prop:covariance},
the permutation covariance of any i.i.d.~channel, and an application of the
quantum de Finetti theorem, in the form of the postselection technique
\cite{CKR09} (see \cite[Theorem~6]{TWW14} for an explicit proof of the theorem below).

\begin{theorem}
\label{thm:weak-subadditivity} Let $\mathcal{N}_{A^{\prime}\rightarrow B}$ be
a quantum channel. For all $\alpha>1$ and $n\in\mathbb{N}$, we have
\begin{equation}
\widetilde{E}_{R,\alpha}(\mathcal{N}^{\otimes n})\leq n\widetilde{E}%
_{R,\alpha}(\mathcal{N})+\frac{\alpha\left\vert A^{\prime}\right\vert ^{2}%
}{\alpha-1}\log n. \label{eq:weak-subadd}%
\end{equation}

\end{theorem}

To arrive at the statement in Theorem~\ref{thm:rel-ent-ent-strong-conv-cppp},
we note that
\eqref{eq:bound-rel-ent-ent-any-channel-cppp}\ and\ \eqref{eq:rel-ent-hyp-to-rel-ent-ent-alpha-channel}\ lead
to the following bound on the optimal fidelity of any CPPP-assisted protocol:%
\begin{align}
1-\hat{\varepsilon}_{\mathcal{N}}^{\operatorname{cppp}}(n,P)  &
\leq2^{-\left(  \frac{\alpha-1}{\alpha}\right)  \left(  nP-\widetilde
{E}_{R,\alpha}(\mathcal{N}^{\otimes n})\right)  }\\
&  \leq2^{-\left(  \frac{\alpha-1}{\alpha}\right)  \left(  nP-n\widetilde
{E}_{R,\alpha}(\mathcal{N})-\frac{\alpha\left\vert A^{\prime}\right\vert ^{2}%
}{\alpha-1}\log n\right)  }\\
&  =n^{\left\vert A^{\prime}\right\vert ^{2}}2^{-n\left(  \frac{\alpha
-1}{\alpha}\right)  \left(  P-\widetilde{E}_{R,\alpha}(\mathcal{N})\right)  },
\end{align}
where the second inequality follows from Theorem~\ref{thm:weak-subadditivity}.
Thus, if $P>E_{R}(\mathcal{N})$, then by
\eqref{eq:renyi-rel-ent-ent-mono-channel} and
\eqref{eq:renyi-rel-ent-ent-limit-1-channel}, there exists $\alpha>1$ such
that $P>\widetilde{E}_{R,\alpha}(\mathcal{N})$ and so the optimal error
$\hat{\varepsilon}_{\mathcal{N}}^{\operatorname{cppp}}(n,P)$ increases
exponentially fast to one with exponent $\left(  \frac{\alpha-1}{\alpha
}\right)  (P-\widetilde{E}_{R,\alpha}(\mathcal{N}))$ (the polynomial prefactor
$n^{\left\vert A^{\prime}\right\vert ^{2}}$ does not contribute to the
exponent). Thus Theorem~\ref{thm:rel-ent-ent-strong-conv-cppp} follows.

\begin{remark}
The following regularized versions of the bounds in
Theorems~\ref{thm:tp-cov-str-conv} and \ref{thm:rel-ent-ent-strong-conv-cppp}
hold, by applying the same argument as given in \cite[Theorem~8]{TWW14}:%
\begin{align}
P_{\leftrightarrow}^{\dagger}(\mathcal{N}^{\operatorname{TP}})  &  \leq
\inf_{\ell\geq1}\frac{1}{\ell}E_{R}(A^{\ell};B^{\ell})_{\omega^{\otimes\ell}%
},\\
P_{\operatorname{cppp}}^{\dagger}(\mathcal{N})  &  \leq\inf_{\ell\geq1}%
\frac{1}{\ell}E_{R}(\mathcal{N}^{\otimes\ell}),
\end{align}
where $\mathcal{N}^{\operatorname{TP}}$ is a teleportation-simulable channel
with associated state $\omega_{AB}$ and $\mathcal{N}$ is an arbitrary channel.
\end{remark}

\subsection{Strong converses for particular channels}

Two particular channels of interest for which we can establish the strong
converse property for their private capacities are generalized dephasing
channels and quantum erasure channels. A generalized dephasing channel is any
channel with an isometric extension of the form%
\begin{equation}
U_{A\rightarrow BE}^{\mathcal{N}}\equiv\sum_{x=0}^{d-1}|x\rangle_{B}\langle
x|_{A}\otimes\left\vert \psi_{x}\right\rangle _{E}, \label{eq:deph}%
\end{equation}
where the states $\left\vert \psi_{x}\right\rangle $ are arbitrary (not
necessarily orthonormal). Hence, the specification for such a channel is as
follows:%
\begin{equation}
\mathcal{N}(\rho)=\sum_{x,y=0}^{d-1}\langle x|_{A}\rho|y\rangle_{A}%
\ \left\langle \psi_{y}|\psi_{x}\right\rangle \ |x\rangle\langle y|_{B}.
\end{equation}

\begin{proposition}
\label{prop:I_c=R-for-H}Let $\mathcal{N}$ be a generalized dephasing channel
of the form~\eqref{eq:deph}. Then%
\begin{equation}
I_{c}(\mathcal{N})=P(\mathcal{N})=P^{\dag}(\mathcal{N})=P_{\operatorname{cppp}%
}(\mathcal{N})=P_{\operatorname{cppp}}^{\dag}(\mathcal{N})=E_{R}(\mathcal{N}),
\end{equation}
where $I_{c}(\mathcal{N})$ is the coherent information of the channel, defined
in \eqref{eq:max-coh-info}.
\end{proposition}

\noindent A proof for the above proposition proceeds exactly as in the proof
of \cite[Proposition~10]{TWW14}.

\medskip

A quantum erasure channel is defined as follows:%
\begin{equation}
\mathcal{E}_{A^{\prime}\rightarrow B}^{p}:\rho_{A^{\prime}}\mapsto
(1-p)\rho_{B}+p|e\rangle\langle e|_{B}, \label{eq:erasure}%
\end{equation}
where $p\in\lbrack0,1]$ is the erasure probability, $\rho_{B}$ is an isometric
embedding of $\rho_{A^{\prime}}$ into $\mathcal{H}_{B}$, and $|e\rangle$ is a
quantum state orthogonal to $\rho_{B}$. The cppp- and two-way assisted private
capacity of this channel was presented in \cite[Section~IV]{Goodenough2015} and
\cite{PLOB15}\ to be equal to $(1-p)\log|A|$. This channel is
teleportation-simulable, and the associated state $\omega_{AB}$ in this case
can be taken as $\omega_{AB}=\mathcal{E}_{A^{\prime}\rightarrow B}^{p}%
(\Phi_{AA^{\prime}})$. Applying Theorem~\ref{thm:tp-cov-str-conv} and the
method of proof from \cite[Proposition~11]{TWW14}, we obtain the following:

\begin{proposition}
Let $\mathcal{E}^{p}$ be an erasure channel of the form~\eqref{eq:erasure}.
Then%
\begin{equation}
P_{\operatorname{cppp}}(\mathcal{E}^{p})=P_{\operatorname{cppp}}^{\dagger
}(\mathcal{E}^{p})=P_{\leftrightarrow}(\mathcal{E}^{p})=P_{\leftrightarrow
}^{\dagger}(\mathcal{E}^{p})=(1-p)\log|A|.
\end{equation}

\end{proposition}

\section{Second-order expansions for private communication}

\label{sec:second-order}A recent goal of research in quantum information
theory has been to determine second-order characterizations of various quantum
communication tasks \cite{TH12,li12,TT13,DTW14,BDL15,TBR15}. In this program,
the goal is to determine the highest rate of communication achievable for a
given task when constrained to meet a fixed (constant) error probability and
with a sufficiently large number of channel uses available. What one finds
here is called the \textquotedblleft Gaussian approximation,\textquotedblright%
\ which often serves as a good guideline for rates that are achievable at a
fixed error and finite blocklength. Thus, this research goal is especially
important nowadays given that experimentalists have limited control over
quantum systems, only being able to manipulate small numbers of qubits coherently.

Two of the main tools which are consistently used in a second-order analysis
are the quantum relative entropy variance and a second-order expansion of
$D_{H}^{\varepsilon}$. The quantum relative entropy variance $V(\rho
\Vert\sigma)$ is defined as \cite{li12,TH12}%
\begin{equation}
V(\rho\Vert\sigma)\equiv\operatorname{Tr}\{\rho\left[  \log\rho-\log
\sigma-D(\rho\Vert\sigma)\right]  ^{2}\}.
\end{equation}
whenever $\operatorname{supp}(\rho)\subseteq\operatorname{supp}(\sigma)$. The
following second order expansion holds for $n\propto1/\varepsilon^{2}$ and
$\mathcal{H}$ a finite-dimensional Hilbert space \cite{li12,TH12}:%
\begin{equation}
D_{H}^{\varepsilon}(\rho^{\otimes n}\Vert\sigma^{\otimes n})=nD(\rho
\Vert\sigma)+\sqrt{nV(\rho\Vert\sigma)}\Phi^{-1}(\varepsilon)+O(\log n).
\label{eq:2nd-order-expand}%
\end{equation}
In the above, we have used the cumulative distribution function for a standard
normal random variable:%
\begin{equation}
\Phi(a)\equiv\frac{1}{\sqrt{2\pi}}\int_{-\infty}^{a}dx\,\exp\left(
-x^{2}/2\right)  , \label{eq:cumul-gauss}%
\end{equation}
and its inverse, defined as $\Phi^{-1}(\varepsilon)\equiv\sup\left\{
a\in\mathbb{R}\,|\,\Phi(a)\leq\varepsilon\right\}  $. It should be clear from
the context whether $\Phi$ refers to the maximally entangled state or
\eqref{eq:cumul-gauss}. A recent alternative proof of
\eqref{eq:2nd-order-expand} is available in \cite{DPR15}.

\subsection{Converse (upper)\ bounds}

If a quantum channel is teleportation-simulable, then it is possible to give a
second-order expansion of the upper bounds from
\eqref{eq:bound-rel-ent-ent-any-channel-cppp} and
\eqref{eq:two-way-meta-conv}, by employing \eqref{eq:2nd-order-expand}.
(Recall from Proposition~\ref{prop:cov->TP-sim} that every covariant channel
is teleportation-simulable.) Before doing so, we define the following variance
quantity:%
\begin{equation}
V_{E_{R}}^{\varepsilon}(A;B\mathcal{)}_{\rho}\equiv\left\{
\begin{array}
[c]{cc}%
\sup_{\sigma_{AB^{\prime}}\in\Pi_{\mathcal{S}}}V(\rho_{AB}\Vert\sigma_{AB}) &
\text{for }\varepsilon<1/2\\
\inf_{\sigma_{AB}\in\Pi_{\mathcal{S}}}V(\rho_{AB}\Vert\sigma_{AB}) & \text{for
}\varepsilon\geq1/2
\end{array}
\right.  ,
\end{equation}
where $\Pi_{\mathcal{S}}\subseteq\mathcal{S}(A\!:\!B)$ is the set of separable
states achieving the minimum in $E_{R}(A;B\mathcal{)}_{\rho}$. This leads us
to the following theorem, as an immediate consequence of the above reasoning:

\begin{theorem}
If a quantum channel $\mathcal{N}_{A^{\prime}\rightarrow B}$\ is
teleportation-simulable with associated state $\omega_{AB}$, then%
\begin{equation}
\hat{P}_{\mathcal{N}}(n,\varepsilon)\leq\hat{P}_{\mathcal{N}}%
^{\operatorname{cppp}}(n,\varepsilon)\leq\hat{P}_{\mathcal{N}}%
^{\leftrightarrow}(n,\varepsilon)\leq E_{R}(A;B\mathcal{)}_{\omega}%
+\sqrt{\frac{V_{E_{R}}^{\varepsilon}(A;B\mathcal{)}_{\omega}}{n}}\Phi
^{-1}(\varepsilon)+O\!\left(  \frac{\log n}{n}\right)  .
\end{equation}

\end{theorem}

\subsection{Achievable rates and lower bounds}

\label{sec:lower-2nd-order-ach}As discussed in Section~\ref{sec:q-comm},
entanglement transmission achieves the task of secret-key transmission, and
this relationship allows for giving lower bounds on achievable secret-key
transmission rates by employing known lower bounds on achievable entanglement
transmission rates. For some channels of interest, this approach leads to a
tight second-order characterization of their private transmission
capabilities. We now briefly review some known lower bounds on achievable
entanglement transmission rates \cite{BDL15,TBR15}. Two quantities which arise
in such a setting are the conditional quantum entropy and conditional entropy
variance \cite{TH12},\ defined for $\rho_{AB}\in\mathcal{D}(\mathcal{H}_{AB})$
as%
\begin{align}
H(A|B)_{\rho}  &  \equiv-D(\rho_{AB}\Vert I_{A}\otimes\rho_{B}),\\
V(A|B)_{\rho}  &  \equiv V(\rho_{AB}\Vert I_{A}\otimes\rho_{B}).
\end{align}
The coherent information is defined as $I(A\rangle B)_{\rho}\equiv
-H(A|B)_{\rho}$ \cite{PhysRevA.54.2629}\ and its corresponding variance is
$V(A\rangle B)_{\rho}\equiv V(A|B)_{\rho}$. Using these quantities, the
maximum entanglement transmission rate $\hat{Q}_{\mathcal{N}}(n,\varepsilon
)$\ possible has the following general lower bound for $\varepsilon\in(0,1)$
\cite{BDL15,TBR15}:%
\begin{equation}
\hat{Q}_{\mathcal{N}}(n,\varepsilon)\geq\hat{Q}_{\operatorname{lower}%
,\mathcal{N}}(n,\varepsilon)\equiv I_{c}(\mathcal{N})+\sqrt{\frac
{V_{c}^{\varepsilon}(\mathcal{N)}}{n}}\Phi^{-1}(\varepsilon)+O\!\left(
\frac{\log n}{n}\right)  , \label{eq:2nd-order-q-char}%
\end{equation}
where $I_{c}(\mathcal{N})$ is the channel's coherent information:%
\begin{equation}
I_{c}(\mathcal{N})\equiv\max_{|\psi\rangle_{AA^{\prime}}\in\mathcal{H}%
_{AA^{\prime}}}I(A\rangle B)_{\theta}, \label{eq:max-coh-info}%
\end{equation}
$\theta_{AB}\equiv\mathcal{N}_{A^{\prime}\rightarrow B}(\psi_{AA^{\prime}})$,
and $V_{c}^{\varepsilon}(\mathcal{N)}$ is the channel's conditional entropy
variance:%
\begin{equation}
V_{c}^{\varepsilon}(\mathcal{N)\equiv}\left\{
\begin{array}
[c]{cc}%
\min_{\psi_{AA^{\prime}}\in\Pi}V(A\rangle B)_{\theta} & \text{for }%
\varepsilon<1/2\\
\max_{\psi_{AA^{\prime}}\in\Pi}V(A\rangle B)_{\theta} & \text{for }%
\varepsilon\geq1/2
\end{array}
\right.  .
\end{equation}
The set $\Pi\subseteq\mathcal{D}(\mathcal{H}_{AA^{\prime}})$ is the set of all
states achieving the maximum in \eqref{eq:max-coh-info}. For channels with
sufficient symmetry, such as covariant generalized dephasing channels, the
characterization in \eqref{eq:2nd-order-q-char}\ is tight, in the sense that
$\hat{Q}_{\mathcal{N}}(n,\varepsilon)=\hat{Q}_{\operatorname{lower}%
,\mathcal{N}}(n,\varepsilon)$ for sufficiently large $n$ \cite{TBR15}. Due to
\eqref{eq:finite-P-hierarchy} and \eqref{eq:finite-Q-to-finite-P-1}, we can
also conclude that the maximum rates possible for secret-key transmission have
the same lower bound for $\varepsilon\in(0,1)$:%
\begin{equation}
\hat{P}_{\mathcal{N}}^{\leftrightarrow}(n,\varepsilon)\geq\hat{P}%
_{\mathcal{N}}^{\operatorname{cppp}}(n,\varepsilon)\geq\hat{P}_{\mathcal{N}%
}(n,\varepsilon)\geq\hat{Q}_{\operatorname{lower},\mathcal{N}}(n,\varepsilon).
\end{equation}

The protocol for achieving the lower bound in
\eqref{eq:2nd-order-q-char}\ does not utilize forward or backward classical
communication in any way. Proposition~\ref{prop:one-shot-ent-dist-lower} below
gives a lower bound on the one-shot distillable entanglement of a bipartite
state $\rho_{AB}$. Such protocols allow for classical communication assistance
in one direction. Note that lower bounds on one-shot distillable entanglement
have previously appeared in the literature \cite{Berta08,BD10a}, but the bound
given below allows for a tighter characterization when we later consider the
i.i.d.~case and second-order expansions. We give a proof of
Proposition~\ref{prop:one-shot-ent-dist-lower}\ in
Appendix~\ref{app:one-shot-ent-dist-lower}.

\begin{definition}
The smooth conditional max-entropy $H_{\max}^{\xi}(A|B)_{\rho}$\ of a
bipartite state $\rho_{AB}$ is defined for $\xi\in\lbrack0,1)$ as%
\begin{equation}
H_{\max}^{\xi}(A|B)_{\rho}\equiv\inf_{\widetilde{\rho}_{AB}\in\mathcal{B}%
^{\xi}(\rho_{AB})}\sup_{\sigma_{B}\in\mathcal{D}(\mathcal{H}_{B})}\log
F(\widetilde{\rho}_{AB},I_{A}\otimes\sigma_{B}),
\end{equation}
where%
\begin{equation}
\mathcal{B}^{\xi}(\rho_{AB})\equiv\left\{  \rho_{AB}^{\prime}\in
\mathcal{D}_{\leq}(\mathcal{H}_{AB}):P(\rho_{AB},\rho_{AB}^{\prime})\leq
\xi\right\}  ,
\end{equation}
with $P$ denoting the purified distance in \eqref{eq:purified-distance}.
\end{definition}

\begin{proposition}
\label{prop:one-shot-ent-dist-lower}Let $\rho_{AB}\in\mathcal{D}%
(\mathcal{H}_{AB})$, $\varepsilon\in\left[  0,1\right]  $, and $\eta\in
\lbrack0,\sqrt{\varepsilon})$. Then there exists a one-way entanglement
distillation protocol $\Lambda_{AB\rightarrow A^{\prime}B^{\prime}}$,
utilizing classical communication from Alice to Bob, such that%
\begin{equation}
F(\Phi_{A^{\prime}B^{\prime}},\Lambda_{AB\rightarrow A^{\prime}B^{\prime}%
}(\rho_{AB}))\geq1-\varepsilon,
\end{equation}
where $\Phi_{A^{\prime}B^{\prime}}$ is a maximally entangled state\ of Schmidt
rank $d$ and%
\begin{equation}
\log d=-H_{\max}^{\sqrt{\varepsilon}-\eta}(A|B)_{\rho}-4\log\!\left(  \frac
{1}{\eta}\right)  .
\end{equation}

\end{proposition}

One strategy for generating entanglement or secret key by means of a quantum
channel is for

\begin{enumerate}
\item Alice to prepare $n$ copies of a given state $\psi_{AA^{\prime}}$,

\item Alice to send the $n$ systems labeled by $A^{\prime}$ through the
memoryless channel $\mathcal{N}_{A^{\prime}\rightarrow B}^{\otimes n}$, which
leads to $n$ copies of a bipartite state $\theta_{AB}\equiv\mathcal{N}%
_{A^{\prime}\rightarrow B}(\psi_{AA^{\prime}})$, and then for

\item Alice and Bob to perform entanglement distillation by means of backward
classical communication from Bob to Alice.
\end{enumerate}

\noindent Since the classical communication is now from Bob to Alice, the
number of $\varepsilon$-approximate ebits that they can generate using this
method is equal to $-H_{\max}^{\sqrt{\varepsilon}-\eta}(B^{n}|A^{n}%
)_{\theta^{\otimes n}}-4\log\!\left(  \frac{1}{\eta}\right)  $, by applying
Proposition~\ref{prop:one-shot-ent-dist-lower}.

Steps 1-3 above realize an entanglement generation protocol. If the goal is
entanglement transmission, Alice and Bob could subsequently perform quantum
teleportation \cite{PhysRevLett.70.1895}, using forward classical
communication from Alice to Bob, to transmit any system Alice possesses to
Bob. By the monotonicity of the fidelity with respect to quantum channels (and
the teleportation protocol realizing a channel), the fidelity of entanglement
transmission is $\geq1-\varepsilon$ if the fidelity of entanglement generation
is $\geq1-\varepsilon$.

Using standard methods for second-order expansions \cite{TH12}, we find that%
\begin{equation}
\hat{P}_{\mathcal{N}}^{\operatorname{cppp}}(n,\varepsilon)\geq\hat
{Q}_{\mathcal{N}}^{\operatorname{cppp}}(n,\varepsilon)\geq
I_{\operatorname{rev}}(\mathcal{N})+\sqrt{\frac{V_{\operatorname{rev}%
}^{\varepsilon}(\mathcal{N)}}{n}}\Phi^{-1}(\varepsilon)+O\!\left(  \frac{\log
n}{n}\right)  \label{eq:rev-low-bound}%
\end{equation}
where $I_{\operatorname{rev}}(\mathcal{N})$ is the channel's reverse coherent
information (see \cite[Section~5.3]{DJKR06}\ and \cite{GPLS09}):%
\begin{equation}
I_{\operatorname{rev}}(\mathcal{N})\equiv\max_{|\psi\rangle_{AA^{\prime}}%
\in\mathcal{H}_{AA^{\prime}}}I(B\rangle A)_{\theta}, \label{eq:rev-coh-info}%
\end{equation}
$\theta_{AB}\equiv\mathcal{N}_{A^{\prime}\rightarrow B}(\psi_{AA^{\prime}})$,
and $V_{\operatorname{rev}}^{\varepsilon}(\mathcal{N)}$ is the channel's
reverse conditional entropy variance:%
\begin{equation}
V_{\operatorname{rev}}^{\varepsilon}(\mathcal{N)\equiv}\left\{
\begin{array}
[c]{cc}%
\min_{\psi_{AA^{\prime}}\in\Pi_{\operatorname{rev}}}V(B\rangle A)_{\theta} &
\text{for }\varepsilon<1/2\\
\max_{\psi_{AA^{\prime}}\in\Pi_{\operatorname{rev}}}V(B\rangle A)_{\theta} &
\text{for }\varepsilon\geq1/2
\end{array}
\right.  .
\end{equation}
The set $\Pi_{\operatorname{rev}}\subseteq\mathcal{D}(\mathcal{H}_{AA^{\prime
}})$ is the set of all states achieving the maximum in \eqref{eq:rev-coh-info}.

Covariant dephasing channels and quantum erasure channels are two classes of
channels for which we have tight second-order characterizations of their
private transmission capabilities, due to their symmetries. In the next
section, we provide characterizations of the private transmission capabilities
of qubit versions of these channels that are tight to even the third order.

\section{Channels with higher-order characterizations}

\label{sec:examples}This section details several examples of channels for
which we can provide precise characterizations of their private communication
abilities. Some of the results rely heavily on those from \cite{TBR15}, which
in turn rely upon recent results from classical information theory (see the
references given in \cite{TBR15}). For this reason, in these cases we keep the
exposition brief and point to \cite{TBR15} for details.

\subsection{Qubit dephasing channel}

The qubit dephasing channel performs the following transformation on an input
qubit density operator:%
\begin{equation}
\mathcal{Z}^{\gamma}:\rho\longmapsto\left(  1-\gamma\right)  \rho+\gamma Z\rho
Z,
\end{equation}
where $\gamma\in(0,1)$ is the dephasing parameter and $Z$ is shorthand for the
Pauli $\sigma_{Z}$ operator. This channel is teleportation-simulable because
it arises from the action of the teleportation protocol on the state
$\mathcal{Z}_{A^{\prime}\rightarrow B}^{\gamma}(\Phi_{AA^{\prime}}^{+})$,
where $|\Phi^{+}\rangle_{AA^{\prime}}\equiv(|00\rangle_{AA^{\prime}%
}+|11\rangle_{AA^{\prime}})/\sqrt{2}$. As a consequence, the upper bound from
\eqref{eq:two-way-meta-conv} applies. By using the same method given in
\cite{TBR15}, which relates to results for the binary symmetric channel from
classical information theory, we can establish a third-order upper bound. We
can also follow the achievability strategy from \cite{TBR15} to establish a
matching lower bound. This leads to the following result:

\begin{proposition}
For the qubit dephasing channel $\mathcal{Z}^{\gamma}$ with $\gamma\in(0,1)$,
the boundary $\hat{P}(n;\varepsilon)$ satisfies
\begin{equation}
\hat{P}(n,\varepsilon)=\hat{P}^{\operatorname{cppp}}(n,\varepsilon)=\hat
{P}^{\leftrightarrow}(n,\varepsilon)=1-h(\gamma)+\sqrt{\frac{v(\gamma)}{n}%
}\,\Phi^{-1}(\varepsilon)+\frac{\log n}{2n}+O\!\left(  \frac{1}{n}\right)  \,,
\end{equation}
where $\Phi$ is the cumulative standard Gaussian distribution, $h(\gamma)$
denotes the binary entropy and $v(\gamma)$ the corresponding variance, defined
as
\begin{align}
h(\gamma)  &  \equiv-\gamma\log\gamma-(1-\gamma)\log(1-\gamma) ,\\
v(\gamma)  &  \equiv\gamma(\log\gamma+h(\gamma))^{2}+(1-\gamma)(\log
(1-\gamma)+h(\gamma))^{2}.
\end{align}

\end{proposition}

Thus, for this channel, there is no difference between its private and quantum
transmission capabilities.

\subsection{Qubit erasure channel}

The qubit erasure channel, defined for $\left\vert A\right\vert =2$ in
\eqref{eq:erasure}, is another example of a channel for which we can obtain a
precise characterization. This channel is teleportation-simulable because it
arises from the action of the teleportation protocol on $\mathcal{E}%
_{A^{\prime}\rightarrow B}^{p}(\Phi_{AA^{\prime}}^{+})$. Thus, we can apply
the upper bound from \eqref{eq:two-way-meta-conv} and the same reasoning from
\cite{TBR15} to establish a precise upper bound on the rates of private
communication achievable with LOCC-assistance. Also, the same achievability
protocol from \cite{TBR15} gives a lower bound that matches this upper bound,
giving us the following:

\begin{proposition}
For the qubit erasure channel $\mathcal{E}^{p}$ with $p\in(0,1)$, the boundary
$\hat{P}_{\mathcal{E}^{p}}^{\leftrightarrow}(n,\varepsilon)$ satisfies
\begin{equation}
\varepsilon=\sum_{l=n-k+1}^{n}{\binom{n}{l}}p^{l}(1-p)^{n-l}\left(
1-2^{n\left(  1-\hat{P}_{\mathcal{E}^{p}}^{\leftrightarrow}(n,\varepsilon
)\right)  -l}\right)  \,. \label{eq:bec-exact}%
\end{equation}
The same equation holds for $\hat{P}_{\mathcal{E}^{p}}^{\operatorname{cppp}%
}(n,\varepsilon)$. Moreover, the following expansion holds%
\begin{equation}
\hat{P}_{\mathcal{E}^{p}}^{\operatorname{cppp}}(n,\varepsilon)=\hat
{P}_{\mathcal{E}^{p}}^{\leftrightarrow}(n,\varepsilon)=1-p+\sqrt{\frac
{p(1-p)}{n}}\Phi^{-1}(\varepsilon)+O\!\left(  \frac{1}{n}\right)  \,.
\end{equation}

\end{proposition}

Thus, again for this channel, there is no difference between its private and
quantum transmission capabilities. The statement given above provides a strong
refinement of the recent results in \cite{Goodenough2015}
for the qubit erasure channel, which established $1-p$ as the two-way assisted
private capacity of the qubit erasure channel.

\subsection{Entanglement-breaking channels}

\label{sec:EB-channels}Entanglement-breaking channels have played an important
role in the development of quantum information theory \cite{HSR03}, in some
sense being the class of channels which are most similar to a classical
channel. A formal definition is that a channel $\mathcal{N}_{A^{\prime
}\rightarrow B}$ is entanglement breaking if the state $\mathcal{N}%
_{A^{\prime}\rightarrow B}(\rho_{AA^{\prime}})$ is separable regardless of the
input state $\rho_{AA^{\prime}}$. As shown in \cite{HSR03}, it suffices to
check this condition for a single input:\ the maximally entangled state
$\Phi_{AA^{\prime}}$.

Entanglement-breaking channels can be simulated by local operations and
classical communication:\ every such channel can be simulated by a measurement
followed by a preparation of a state conditioned on the measurement outcome
\cite{HSR03,Holevo2008}. As a consequence, any $P^{\leftrightarrow}$ protocol
using an entanglement-breaking channel $n$ times can only generate a separable
state at the end of the protocol. Applying the same method of proof as in
Theorem~\ref{prop:meta-conv-priv}\ and the observation from
\eqref{eq:D_H-rho-sigma-equal}, we find that the following bound holds for any
entanglement-breaking channel $\mathcal{N}$:%
\begin{equation}
\hat{P}_{\mathcal{N}}^{\leftrightarrow}(n,\varepsilon)\leq-\frac{1}{n}%
\log\left(  1-\varepsilon\right)  .
\end{equation}
Thus, for these channels, the first, second, and third order terms all vanish,
implying that such channels have essentially no capability to transmit private
information. Given \eqref{eq:finite-Q-to-finite-P-1}, the same upper bound
holds for $\hat{Q}_{\mathcal{N}}^{\leftrightarrow}(n,\varepsilon)$, a result
already obtained by the transposition bound method from \cite[Section~IV-C]%
{Mueller-Hermes2015}.

\section{Quantum Gaussian channels}

\label{sec:gaussian} Quantum Gaussian channels are an important model for
communication in realistic settings, such as free space and fiber-optic
communication. A relevant subclass are the phase-insensitive channels, which
add noise equally to the position and momentum quadrature of a given bosonic
mode. Several results are now known for the various capacities of these
channels (see the reviews in \cite{S09,WPGCRSL12,HG12}\ and see
\cite{GGCH14,PLOB15}\ for more recent developments).

In this section, we are interested in establishing bounds on the private and
quantum communication capabilities of three kinds of phase-insensitive bosonic
channels: the thermalizing channels, the amplifier channels, and the additive
noise channels. Each of these are defined respectively by the following
Heisenberg input-output relations:%
\begin{align}
\hat{b}  &  =\sqrt{\eta}\hat{a}+\sqrt{1-\eta}\hat{e}%
,\label{eq:thermal-channel}\\
\hat{b}  &  =\sqrt{G}\hat{a}+\sqrt{G-1}\hat{e}^{\dag}%
,\label{eq:amplifier-channel}\\
\hat{b}  &  =\hat{a}+\left(  x+ip\right)  /\sqrt{2},
\label{eq:additive-noise-channel}%
\end{align}
where $\hat{a}$, $\hat{b}$, and $\hat{e}$ are the field-mode annihilation
operators for the sender's input, the receiver's output, and the environment's
input of these channels, respectively.

The channel in \eqref{eq:thermal-channel} is a thermalizing channel, in which
the environmental mode is prepared in a thermal state $\theta(N_{B})$\ of mean
photon number $N_{B}\geq0$, defined as%
\begin{equation}
\theta(N_{B})\equiv\frac{1}{N_{B}+1}\sum_{n=0}^{\infty}\left(  \frac{N_{B}%
}{N_{B}+1}\right)  ^{n}|n\rangle\langle n|,
\end{equation}
where $\left\{  |n\rangle\right\}  _{n=0}^{\infty}$ is the orthonormal,
photonic number-state basis. When $N_{B}=0$, $\theta(N_{B})$ reduces to the
vacuum state, in which case the resulting channel in
\eqref{eq:thermal-channel} is called the pure-loss channel---it is said to be
quantum-limited in this case because the environment is injecting the minimum
amount of noise allowed by quantum mechanics. The parameter $\eta\in\left[
0,1\right]  $ is the transmissivity of the channel, representing the average
fraction of photons making it from the input to the output of the channel. The
channel in \eqref{eq:thermal-channel} is entanglement-breaking when $\left(
1-\eta\right)  N_{B}\geq\eta$ \cite{Holevo2008}. Let $\mathcal{L}_{\eta,N_{B}%
}$ denote this channel, and we make the further abbreviation $\mathcal{L}%
_{\eta}\equiv\mathcal{L}_{\eta,N_{B}=0}$ when it is the pure-loss channel.

The channel in \eqref{eq:amplifier-channel} is an amplifier channel, and the
parameter $G\geq1$ is its gain. For this channel, the environment is prepared
in the thermal state $\theta(N_{B})$. If $N_{B}=0$, the amplifier channel is
said to be quantum-limited for a similar reason as stated above. The channel
in \eqref{eq:amplifier-channel} is entanglement-breaking when $\left(
G-1\right)  N_{B}\geq1$ \cite{Holevo2008}. Let $\mathcal{A}_{G,N_{B}}$ denote
this channel, and we make the further abbreviation $\mathcal{A}_{G}%
\equiv\mathcal{A}_{G,N_{B}=0}$ when it is the quantum-limited amplifier channel.

Finally, the channel in \eqref{eq:additive-noise-channel} is an additive noise
channel, representing a quantum generalization of the classical additive white
Gaussian noise channel. In \eqref{eq:additive-noise-channel}, $x$ and $p$ are
zero-mean Gaussian random variables each having variance $\xi\geq0$. The
channel in \eqref{eq:additive-noise-channel} is entanglement-breaking when
$\xi\geq1$ \cite{Holevo2008}. Let $\mathcal{W}_{\xi}$ denote this channel.
Kraus representations for the channels in
\eqref{eq:thermal-channel}--\eqref{eq:additive-noise-channel}\ are available
in \cite{ISS11}.

For our purposes, it suffices to consider the three kinds of channels given
above. All other phase-insensitive Gaussian channels are
entanglement-breaking, and they thus have their private and quantum
communication abilities severely limited as discussed in
Section~\ref{sec:EB-channels}. Figure~1 in \cite{GGCH14} is helpful for
understanding the various phase-insensitive quantum Gaussian channels (the
channels given in
\eqref{eq:thermal-channel}--\eqref{eq:additive-noise-channel} all fall within
the white-shaded area in that figure).

Recently, the authors of \cite{PLOB15}\ presented several upper bounds on the
two-way assisted private capacities of these channels:%
\begin{align}
P_{\leftrightarrow}(\mathcal{L}_{\eta,N_{B}})  &  \leq-\log\!\left(  \left(
1-\eta\right)  \eta^{N_{B}}\right)  -g(N_{B}),\label{eq:loss-upper-bnd}\\
P_{\leftrightarrow}(\mathcal{A}_{G,N_{B}})  &  \leq\log\!\left(
\frac{G^{N_{B}+1}}{G-1}\right)  -g(N_{B}),\\
P_{\leftrightarrow}(\mathcal{W}_{\xi})  &  \leq\frac{\xi-1}{\ln2}-\log\xi,
\label{eq:additive-upper-bnd}%
\end{align}
where%
\begin{equation}
g(x)\equiv\left(  x+1\right)  \log\left(  x+1\right)  -x\log x
\end{equation}
is the quantum entropy of a thermal state with mean photon number $x\geq0$. If
the channels are entanglement breaking (specific parameter values discussed
above), then the upper bound can be taken as zero.

When considering capacities of communication, it is a common convention to
impose an energy constraint on the modes being input to the channel uses. This
constraint acknowledges the practical reality that a communication scheme
cannot consume an arbitrary amount of energy in any given protocol. However,
it is known that the quantum and private capacities of quantum Gaussian
channels are bounded even if an infinite amount of energy is available
\cite{HW01,TGW14IEEE} (this is in contrast to the classical capacity
\cite{HW01}). Thus, we can consider both the unconstrained capacity (with no
constraint on energy consumption) and the constrained capacity in these
scenarios. Note that the bounds in
\eqref{eq:loss-upper-bnd}--\eqref{eq:additive-upper-bnd} apply to both the
constrained and unconstrained capacities.

Regarding lower bounds on the capacities, an achievable rate for both quantum
and private data transmission is given by the reverse coherent information (as
discussed in Section~\ref{sec:lower-2nd-order-ach}):%
\begin{equation}
I(B\rangle A)_{\omega}\equiv H(A)_{\omega}-H(AB)_{\omega},
\end{equation}
where $\omega_{AB}\equiv\mathcal{N}_{A^{\prime}\rightarrow B}(\psi
_{AA^{\prime}})$. Alternatively, the coherent information $I(A\rangle
B)_{\omega}$ is an achievable rate as discussed in
Section~\ref{sec:lower-2nd-order-ach}. Evaluating these quantities for the
channels of interest in
\eqref{eq:thermal-channel}--\eqref{eq:additive-noise-channel}\ leads to the
following lower bounds on the unconstrained two-way assisted quantum
capacities \cite{PGBL09}:%
\begin{align}
-\log(1-\eta)-g(N_{B})  &  \leq Q_{\leftrightarrow}(\mathcal{L}_{\eta,N_{B}%
}),\\
\log\!\left(  \frac{G}{G-1}\right)  -g(N_{B})  &  \leq Q_{\leftrightarrow
}(\mathcal{A}_{G,N_{B}}),\\
-1/\ln2-\log\xi &  \leq Q_{\leftrightarrow}(\mathcal{W}_{\xi}).
\end{align}
For the pure-loss and quantum-limited amplifier channels, we thus have an
exact characterization of their unconstrained capacities \cite{PLOB15}:%
\begin{align}
Q_{\leftrightarrow}(\mathcal{L}_{\eta})  &  =P_{\leftrightarrow}%
(\mathcal{L}_{\eta})=-\log\left(  1-\eta\right)  ,
\label{eq:pure-loss-capacity-inf-energy}\\
Q_{\leftrightarrow}(\mathcal{A}_{G})  &  =P_{\leftrightarrow}(\mathcal{A}%
_{G})=\log\!\left(  \frac{G}{G-1}\right)  . \label{eq:amp-capacity-inf-energy}%
\end{align}

The following theorem refines the upper bounds in
\eqref{eq:loss-upper-bnd}--\eqref{eq:additive-upper-bnd}, which in turn
establishes the rates as strong converse rates and solidifies the claims of
\cite{PLOB15}:

\begin{theorem}
\label{thm:bosonic-bounds}Let $V_{\mathcal{L}_{\eta,N_{B}}}$, $V_{\mathcal{A}%
_{G,N_{B}}}$, and $V_{\mathcal{W}_{\xi}}$ be the unconstrained relative
entropy variances of the thermalizing, amplifier, and additive-noise channels,
respectively:%
\begin{align}
V_{\mathcal{L}_{\eta,N_{B}}}  &  \equiv N_{B}(N_{B}+1)\log^{2}(\eta\left[
N_{B}+1\right]  /N_{B}),\label{eq:unc-rel-ent-var-thermal}\\
V_{\mathcal{A}_{G,N_{B}}}  &  \equiv N_{B}(N_{B}+1)\log^{2}(G^{-1}\left[
N_{B}+1\right]  /N_{B}),\label{eq:unc-rel-ent-var-amp}\\
V_{\mathcal{W}_{\xi}}  &  \equiv\left(  1-\xi\right)  ^{2}/\ln^{2}2.
\label{eq:unc-rel-ent-var-add}%
\end{align}
The following converse bounds hold for all $\varepsilon\in(0,1)$, $n\geq1$,
and $N_{B}>0$:%
\begin{align}
\hat{P}_{\mathcal{L}_{\eta,N_{B}}}^{\leftrightarrow}(n,\varepsilon)  &
\leq-\log\!\left(  \left(  1-\eta\right)  \eta^{N_{B}}\right)  -g(N_{B}%
)+\sqrt{\frac{2V_{\mathcal{L}_{\eta,N_{B}}}}{n\left(  1-\varepsilon\right)  }%
}+\frac{C(\varepsilon)}{n},\label{eq:CLT-bound-thermal}\\
\hat{P}_{\mathcal{A}_{G,N_{B}}}^{\leftrightarrow}(n,\varepsilon)  &  \leq
\log\!\left(  \frac{G^{N_{B}+1}}{G-1}\right)  -g(N_{B})+\sqrt{\frac
{2V_{\mathcal{A}_{G,N_{B}}}}{n\left(  1-\varepsilon\right)  }}+\frac
{C(\varepsilon)}{n},\label{eq:CLT-bound-amp}\\
\hat{P}_{\mathcal{W}_{\xi}}^{\leftrightarrow}(n,\varepsilon)  &  \leq\frac
{\xi-1}{\ln2}-\log\xi+\sqrt{\frac{2V_{\mathcal{W}_{\xi}}}{n\left(
1-\varepsilon\right)  }}+\frac{C(\varepsilon)}{n}, \label{eq:CLT-bound-add}%
\end{align}
where $C(\varepsilon)\equiv\log6+2\log\left(  \left[  1+\varepsilon\right]
/\left[  1-\varepsilon\right]  \right)  $. For the pure-loss channel
$\mathcal{L}_{\eta}$\ and quantum-limited amplifier channel $\mathcal{A}_{G}$,
the following bounds hold%
\begin{align}
\hat{P}_{\mathcal{L}_{\eta}}^{\leftrightarrow}(n,\varepsilon)  &  \leq
-\log(1-\eta)+\frac{C(\varepsilon)}{n},\label{eq:pure-loss-tight-bound}\\
\hat{P}_{\mathcal{A}_{G}}^{\leftrightarrow}(n,\varepsilon)  &  \leq
\log\!\left(  \frac{G}{G-1}\right)  +\frac{C(\varepsilon)}{n}.
\label{eq:amp-tight-bound}%
\end{align}

\end{theorem}

\begin{proof}
We will argue a proof of the bound \eqref{eq:CLT-bound-thermal} for the
thermalizing channel, then for the pure-loss channel, and finally the other
bounds will follow from similar reasoning. First, consider an arbitrary
$(n,P^{\leftrightarrow},\varepsilon)$ protocol for the thermalizing channel
$\mathcal{L}_{\eta,N_{B}}$. It consists of using the channel $n$~times and
interleaving rounds of LOCC\ between every channel use. Let $\zeta_{\hat
{A}\hat{B}}^{n}$ denote the final state of Alice and Bob at the end of this protocol.

By the teleportation reduction procedure from \cite[Section~V]{BDSW96}\ and
\cite{NFC09} (see also the review in \cite{PLOB15}), such a protocol can be
simulated by preparing $n$\ two-mode squeezed vacuum (TMSV) states each having
energy $\mu-1/2$ (where we think of $\mu\geq1/2$ as a very large positive
real), sending one mode of each TMSV through each channel use, and then
performing continuous-variable quantum teleportation \cite{prl1998braunstein}
to delay all of the LOCC\ operations until the end of the protocol. Let
$\rho_{\eta,N_{B}}^{\mu}$ denote the state resulting from sending one share of
the TMSV\ through the thermalizing channel, and let $\zeta_{\hat{A}\hat{B}%
}^{\prime}(n,\mu)$ denote the state at the end of the simulation. Let
$\varepsilon_{\operatorname{TP}}(n,\mu)$ denote the \textquotedblleft
infidelity\textquotedblright\ of the simulation:%
\begin{equation}
\varepsilon_{\operatorname{TP}}(n,\mu)\equiv1-F(\zeta_{\hat{A}\hat{B}}%
^{n},\zeta_{\hat{A}\hat{B}}^{\prime}(n,\mu)).
\end{equation}
Due to the fact that continuous-variable teleportation induces a perfect
quantum channel when infinite energy is available \cite{prl1998braunstein},
the following limit holds for every $n$:%
\begin{equation}
\limsup_{\mu\rightarrow\infty}\varepsilon_{\operatorname{TP}}(n,\mu)=0.
\end{equation}
Note also that $\varepsilon_{\operatorname{TP}}(n,\mu)$ is a monotone
non-increasing function of $\mu$, because the fidelity of continuous-variable
teleportation increases with increasing energy \cite{prl1998braunstein}. By
using that $\sqrt{1-F(\rho,\sigma)}$ is a distance measure for states $\rho$
and $\sigma$ (and thus obeys a triangle inequality) \cite{GLN04}, the
simulation leads to an $(n,P^{\leftrightarrow},\varepsilon(n,\mu))$ protocol
for the thermalizing channel, where%
\begin{equation}
\varepsilon(n,\mu)\equiv\min\!\left\{  1,\left[  \sqrt{\varepsilon}%
+\sqrt{\varepsilon_{\operatorname{TP}}(n,\mu)}\right]  ^{2}\right\}  .
\end{equation}
Observe that $\limsup_{\mu\rightarrow\infty}\varepsilon(n,\mu)=\varepsilon$,
so that the simulated protocol has equivalent performance to the original
protocol in the infinite-energy limit. However, the simulated protocol has a
simpler form than the original one, and since this procedure can be conducted
for any $(n,P^{\leftrightarrow},\varepsilon)$\ protocol for the channel
$\mathcal{L}_{\eta,N_{B}}$, we find that the following bound applies by
invoking reasoning similar to that needed to arrive at
\eqref{eq:two-way-meta-conv}:%
\begin{align}
\hat{P}_{\mathcal{L}_{\eta,N_{B}}}^{\leftrightarrow}(n,\varepsilon)  &
\leq\hat{P}_{\rho_{\eta,N_{B}}^{\mu}}^{\leftrightarrow}(n,\varepsilon
(n,\mu))\\
&  \leq\frac{1}{n}E_{R}^{\varepsilon(n,\mu)}(A^{n};B^{n})_{(\rho_{\eta,N_{B}%
}^{\mu})^{\otimes n}}\\
&  \leq\frac{1}{n}D_{H}^{\varepsilon(n,\mu)}((\rho_{\eta,N_{B}}^{\mu
})^{\otimes n}\Vert(\sigma_{\eta,N_{B}}^{\mu})^{\otimes n}),
\label{eq:HT-bound-thermal}%
\end{align}
where $\hat{P}_{\rho_{\eta,N_{B}}^{\mu}}^{\leftrightarrow}(n,\varepsilon
(n,\mu))$ denotes the distillable key using $n$ copies of $\rho_{\eta,N_{B}%
}^{\mu}$ up to error $\varepsilon(n,\mu)$ and the state $\sigma_{\eta,N_{B}%
}^{\mu}$ is a particular separable state chosen as in
\eqref{eq:sep-state-thermal}\ (in Appendix~\ref{app:sep-states-variances}). By
applying the Chebyshev-like bound given in Appendix~\ref{app:CLT-chebyshev},
we find that%
\begin{equation}
\frac{1}{n}D_{H}^{\varepsilon(n,\mu)}((\rho_{\eta,N_{B}}^{\mu})^{\otimes
n}\Vert(\sigma_{\eta,N_{B}}^{\mu})^{\otimes n})\leq D(\rho_{\eta,N_{B}}^{\mu
}\Vert\sigma_{\eta,N_{B}}^{\mu})+\sqrt{\frac{2V(\rho_{\eta,N_{B}}^{\mu}%
\Vert\sigma_{\eta,N_{B}}^{\mu})}{n(1-\varepsilon(n,\mu))}}+C(\varepsilon
(n,\mu))/n.
\end{equation}
Thus, the following bound holds for every $\mu$ sufficiently large (so that
$\varepsilon(n,\mu)\in(0,1)$):%
\begin{equation}
\hat{P}_{\mathcal{L}_{\eta,N_{B}}}^{\leftrightarrow}(n,\varepsilon)\leq
D(\rho_{\eta,N_{B}}^{\mu}\Vert\sigma_{\eta,N_{B}}^{\mu})+\sqrt{\frac
{2V(\rho_{\eta,N_{B}}^{\mu}\Vert\sigma_{\eta,N_{B}}^{\mu})}{n(1-\varepsilon
(n,\mu))}}+C(\varepsilon(n,\mu))/n. \label{eq:final-bound-CLT}%
\end{equation}
From the developments in \cite{PLOB15} and
Appendix~\ref{app:sep-states-variances}, we have the following expansions
about $\mu=\infty$:%
\begin{align}
D(\rho_{\eta,N_{B}}^{\mu}\Vert\sigma_{\eta,N_{B}}^{\mu})  &  =-\log\!\left(
\left(  1-\eta\right)  \eta^{N_{B}}\right)  -g(N_{B})+O(\mu^{-1}%
),\label{eq:rel-ent-inf-energy-limit}\\
V(\rho_{\eta,N_{B}}^{\mu}\Vert\sigma_{\eta,N_{B}}^{\mu})  &  =V_{\mathcal{L}%
_{\eta,N_{B}}}+O\!\left(  \mu^{-1}\right)  .
\label{eq:rel-ent-var-inf-energy-limit}%
\end{align}
We can then take the limit in \eqref{eq:final-bound-CLT}\ as $\mu
\rightarrow\infty$ to conclude the bound stated in \eqref{eq:CLT-bound-thermal}.

To recover the other bounds in \eqref{eq:CLT-bound-amp} and
\eqref{eq:CLT-bound-add}, we apply the same reasoning as above but instead use
the infinite-energy expansions of the relative entropy and the relative
entropy variance given in \cite{PLOB15}\ and
Appendix~\ref{app:sep-states-variances}, respectively.

The bound in \eqref{eq:pure-loss-tight-bound}\ for the pure-loss channel
follows similarly but requires a careful argument with appropriate limits
because the propositions in Appendix~\ref{app:CLT-chebyshev} are stated to
hold only for faithful states (i.e., positive-definite states). The main idea
of the proof below is to apply the following statement given in \cite{DTW14}%
:\ for $\rho^{\prime},\rho,\sigma\in\mathcal{D}(\mathcal{H})$, $\varepsilon
\in(0,1)$, and $\delta\in(0,1-\varepsilon)$,%
\begin{equation}
\frac{1}{2}\left\Vert \rho^{\prime}-\rho\right\Vert _{1}\leq\delta
\qquad\Rightarrow\qquad D_{H}^{\varepsilon}(\rho\Vert\sigma)\leq
D_{H}^{\varepsilon+\delta}(\rho^{\prime}\Vert\sigma). \label{eq:dh-continuity}%
\end{equation}
To this end, we can now repeat the reasoning from above for the pure-loss
channel $\mathcal{L}_{\eta}$. Let $\rho_{\eta}^{\mu}$ be the state arising
from sending one mode of a TMSV through the pure-loss channel $\mathcal{L}%
_{\eta}$. This state is $f(\delta_{n})$-close in trace distance to a faithful
state $\rho_{\eta,\delta_{n}}^{\mu}$\ that would result from sending one share
of a TMSV through a thermalizing channel $\mathcal{L}_{\eta,\delta_{n}}$,
where $\delta_{n}>0$ is a tunable parameter that we will eventually take to
zero at the end of the argument and the function $f$ is such that
$\lim_{\delta_{n}\rightarrow0}f(\delta_{n})=0$. We can choose $\delta_{n}$
small enough and the energy$~\mu$ large enough such that $\varepsilon
(n,\mu)+f(\delta_{n})\in(0,1)$ for a given $n$. Proceeding as in
\eqref{eq:HT-bound-thermal}, we pick $\sigma_{\eta,\delta_{n}}^{\mu}$ to be a
separable state \textquotedblleft tuned\textquotedblright\ for $\rho
_{\eta,\delta_{n}}^{\mu}$, chosen as in \eqref{eq:sep-state-thermal}\ (in
Appendix~\ref{app:sep-states-variances}), and we find the following for all
$n\geq1$, for all $\mu$ sufficiently large, and $\delta_{n}$ small enough:%
\begin{align}
\hat{P}_{\mathcal{L}_{\eta}}^{\leftrightarrow}(n,\varepsilon)  &  \leq\hat
{P}_{\rho_{\eta}^{\mu}}^{\leftrightarrow}(n,\varepsilon(n,\mu))\\
&  \leq\frac{1}{n}E_{R}^{\varepsilon(n,\mu)}(A^{n};B^{n})_{(\rho_{\eta}^{\mu
})^{\otimes n}}\\
&  \leq\frac{1}{n}D_{H}^{\varepsilon(n,\mu)}((\rho_{\eta}^{\mu})^{\otimes
n}\Vert(\sigma_{\eta,\delta_{n}}^{\mu})^{\otimes n})\\
&  \leq\frac{1}{n}D_{H}^{\varepsilon(n,\mu)+f(\delta_{n})}((\rho_{\eta
,\delta_{n}}^{\mu})^{\otimes n}\Vert(\sigma_{\eta,\delta_{n}}^{\mu})^{\otimes
n})\\
&  \leq D(\rho_{\eta,\delta_{n}}^{\mu}\Vert\sigma_{\eta,\delta_{n}}^{\mu
})+\sqrt{\frac{2V(\rho_{\eta,\delta_{n}}^{\mu}\Vert\sigma_{\eta,\delta_{n}%
}^{\mu})}{n(1-\varepsilon(n,\mu)-\delta_{n})}}+C(\varepsilon(n,\mu)+\delta
_{n})/n.
\end{align}
The first inequality follows from the teleportation simulation argument. The
second follows by invoking reasoning similar to that needed to arrive at
\eqref{eq:two-way-meta-conv}. The third inequality follows by picking the
separable state in $E_{R}^{\varepsilon(n,\mu)}$\ to be $(\sigma_{\eta
,\delta_{n}}^{\mu})^{\otimes n}$. The fourth inequality follows from
\eqref{eq:dh-continuity}. The last inequality follows by applying the
Chebyshev-like bound given in Appendix~\ref{app:CLT-chebyshev}. We can now
take the limit as $\mu\rightarrow\infty$, applying
\eqref{eq:rel-ent-inf-energy-limit}--\eqref{eq:rel-ent-var-inf-energy-limit},
and find that%
\begin{equation}
\hat{P}_{\mathcal{L}_{\eta}}^{\leftrightarrow}(n,\varepsilon)\leq
-\log\!\left(  \left(  1-\eta\right)  \eta^{\delta_{n}}\right)  -g(\delta
_{n})+\sqrt{\frac{2V_{\mathcal{L}_{\eta},\delta_{n}}}{n(1-\varepsilon
-\delta_{n})}}+C(\varepsilon+\delta_{n})/n.
\end{equation}
Since the above bound holds for all sufficiently small $\delta_{n}$, we can
now take the limit as $\delta_{n}\rightarrow0$, and find the following bound:%
\begin{equation}
\hat{P}_{\mathcal{L}_{\eta}}^{\leftrightarrow}(n,\varepsilon)\leq-\log\left(
1-\eta\right)  +C(\varepsilon)/n,
\end{equation}
which follows because%
\begin{align}
\lim_{\delta_{n}\rightarrow0}\left[  -\log\!\left(  \left(  1-\eta\right)
\eta^{\delta_{n}}\right)  -g(\delta_{n})\right]   &  =-\log(1-\eta),\\
\lim_{\delta_{n}\rightarrow0}V_{\mathcal{L}_{\eta},\delta_{n}}  &  =0.
\end{align}
Similar reasoning applies to get the bound in \eqref{eq:amp-tight-bound}\ for
the quantum-limited amplifier channel.
\end{proof}

\begin{corollary}
As a consequence of Theorem~\ref{thm:bosonic-bounds}, the following bounds
hold%
\begin{align}
P_{\leftrightarrow}^{\dag}(\mathcal{L}_{\eta,N_{B}})  &  \leq-\log\!\left(
\left(  1-\eta\right)  \eta^{N_{B}}\right)  -g(N_{B}),\\
P_{\leftrightarrow}^{\dag}(\mathcal{A}_{G,N_{B}})  &  \leq\log\!\left(
\frac{G^{N_{B}+1}}{G-1}\right)  -g(N_{B}),\\
P_{\leftrightarrow}^{\dag}(\mathcal{W}_{\xi})  &  \leq\frac{\xi-1}{\ln2}%
-\log\xi,
\end{align}
establishing the upper bounds in
\eqref{eq:loss-upper-bnd}--\eqref{eq:additive-upper-bnd} as strong converse
rates for the constrained and unconstrained private and quantum capacities of
these channels. Furthermore, by taking the same limit in
\eqref{eq:pure-loss-tight-bound}--\eqref{eq:amp-tight-bound} and by combining
with the statements in
\eqref{eq:pure-loss-capacity-inf-energy}--\eqref{eq:amp-capacity-inf-energy},
we can conclude that the unconstrained, two-way assisted private and quantum
capacities of the pure-loss and quantum-limited amplifier channels satisfy the
strong converse property:%
\begin{align}
Q_{\leftrightarrow}(\mathcal{L}_{\eta})  &  =Q_{\leftrightarrow}^{\dag
}(\mathcal{L}_{\eta})=P_{\leftrightarrow}(\mathcal{L}_{\eta}%
)=P_{\leftrightarrow}^{\dag}(\mathcal{L}_{\eta})=-\log(1-\eta),\\
Q_{\leftrightarrow}(\mathcal{A}_{G})  &  =Q_{\leftrightarrow}^{\dag
}(\mathcal{A}_{G})=P_{\leftrightarrow}(\mathcal{A}_{G})=P_{\leftrightarrow
}^{\dag}(\mathcal{A}_{G})=\log\!\left(  \frac{G}{G-1}\right)  .
\end{align}

\end{corollary}

We note that the corollary above improves upon the upper bound on
$Q_{\leftrightarrow}^{\dag}(\mathcal{L}_{\eta})$ and $Q_{\leftrightarrow
}^{\dag}(\mathcal{A}_{G})$ that one gets by applying the transposition bound
\cite{HW01,Mueller-Hermes2015} (in fact there cannot be any further
improvements of the result stated above due to the equalities).

\section{Conclusion}

\label{sec:concl}

We have outlined a general approach for obtaining converse bounds for the
private transmission capabilities of a quantum channel, which builds strongly
on the notion of a private state \cite{HHHO05,HHHO09} and the relative entropy
of entanglement bound therein. We first obtained a general meta-converse bound
and then applied it to obtain strong converse and second-order bounds for
private communication, building upon the methods of \cite{TWW14} and
\cite{TBR15}. For several channels of interest, we can go a step further and
give precise characterizations, as was done in \cite{TBR15}. Notable examples
include the phase-insensitive bosonic channels as well. In particular, we have
established the strong converse property for the unconstrained private and
quantum capacities of the pure-loss and quantum-limited amplifier channels, in
addition to some converse bounds for more general phase-insensitive bosonic
channels. Several of these bounds are relevant for understanding the
limitations of quantum key distribution protocols performed over such channels.

Going forward from here, it is desirable to obtain a second-order expansion of
the achievability results from \cite{ieee2005dev,1050633}. Progress in the
classical case is available in \cite{Tan12,Yang2016}, but the problem seems
generally open there as well. It is also known that there are quantum channels
which have zero quantum capacity but non-zero private capacity
\cite{HHHLO08,PhysRevLett.100.110502}, and a second-order analysis might
further elucidate this phenomenon. It would be interesting as well to prove
that the squashed entanglement of a quantum channel $\mathcal{N}$\ is an upper
bound on $Q_{\leftrightarrow}^{\dag}(\mathcal{N})$ and $P_{\leftrightarrow
}^{\dag}(\mathcal{N})$ (this question has remained open since \cite{TGW14IEEE}).

\bigskip

\textbf{Acknowledgements.} We are grateful to Siddhartha Das, Nilanjana Datta,
Eleni Diamanti, Kenneth Goodenough, Michal Horodecki, Felix Leditzky, Will
Matthews, Alexander M\"{u}ller-Hermes, Yan Pautrat, Stefano Pirandola, Joe
Renes, David Sutter, Masahiro Takeoka, and Stephanie Wehner for helpful
discussions. MB\ and MT\ thank the Hearne Institute for Theoretical Physics at
Louisiana State University for hosting them for a research visit in March
2016. MB acknowledges funding provided by the Institute for Quantum
Information and Matter, an NSF Physics Frontiers Center (NFS Grant
PHY-1125565) with support of the Gordon and Betty Moore Foundation
(GBMF-12500028). Additional funding support was provided by the ARO grant for
Research on Quantum Algorithms at the IQIM (W911NF-12-1-0521). MT is funded by
a University of Sydney Postdoctoral Fellowship and acknowledges support from
the ARC Centre of Excellence for Engineered Quantum Systems (EQUS). MMW
acknowledges the NSF under Award No.~CCF-1350397.

\appendix

\section{Covariant channels are teleportation simulable}

\label{app:cov->TP-sim}In this appendix, we give a proof of
Proposition~\ref{prop:cov->TP-sim}: any covariant channel, as defined in
Definition~\ref{def:covariant-channel}, is teleportation simulable.

Let $\mathcal{N}:\mathcal{L}(\mathcal{H}_{A})\rightarrow\mathcal{L}%
(\mathcal{H}_{B})$ be a quantum channel, and let $G$ be a group with unitary
representations $U_{A}^{g}$ and $V_{B}^{g}$ for $g\in G$, such that%
\begin{align}
\frac{1}{\left\vert G\right\vert }\sum_{g}U_{A}^{g}X_{A}\left(  U_{A}%
^{g}\right)  ^{\dag}  &  =\operatorname{Tr}\{X_{A}\}\pi_{A}%
,\label{eq:one-design-cond}\\
\mathcal{N}_{A\rightarrow B}(U_{A}^{g}X_{A}\left(  U_{A}^{g}\right)  ^{\dag})
&  =V_{B}^{g}\mathcal{N}_{A\rightarrow B}(X_{A})\left(  V_{B}^{g}\right)
^{\dag},
\end{align}
where $X_{A}\in\mathcal{L}(\mathcal{H}_{A})$ and $\pi$ denotes the maximally
mixed state. Consider that%
\begin{equation}
\frac{1}{\left\vert G\right\vert }\sum_{g}U_{A^{\prime}}^{g}\Phi_{A^{\prime}%
A}\left(  U_{A^{\prime}}^{g}\right)  ^{\dag}=\pi_{A^{\prime}}\otimes\pi_{A},
\label{eq:max-ent-cov-action}%
\end{equation}
where $\Phi$ denotes a maximally entangled state and $A^{\prime}$ is a system
isomorphic to $A$. Note that in order for $\{U_{A}^{g}\}$ to satisfy
\eqref{eq:one-design-cond}, it is necessary that $\left\vert A\right\vert
^{2}\leq\left\vert G\right\vert $ \cite{AMTW00}. Consider the POVM
$\{E_{A^{\prime}A}^{g}\}_{g}$, with $A^{\prime}$ a system isomorphic to $A$
and each element $E_{A^{\prime}A}^{g}$ defined as%
\begin{equation}
E_{A^{\prime}A}^{g}\equiv\frac{\left\vert A\right\vert ^{2}}{\left\vert
G\right\vert }U_{A^{\prime}}^{g}\Phi_{A^{\prime}A}\left(  U_{A^{\prime}}%
^{g}\right)  ^{\dag}.
\end{equation}
It follows from the fact that $\left\vert A\right\vert ^{2}\leq\left\vert
G\right\vert $ and \eqref{eq:max-ent-cov-action}\ that $\{E_{AA^{\prime}}%
^{g}\}_{g}$ is a valid POVM.

The simulation of the channel $\mathcal{N}_{A\rightarrow B}$ via teleportation
begins with a state $\rho_{A^{\prime}}$ and a shared resource $\omega
_{AB}\equiv\mathcal{N}_{A^{\prime\prime}\rightarrow B}(\Phi_{AA^{\prime\prime
}})$. The desired outcome is for Bob to receive the state $\mathcal{N}%
_{A\rightarrow B}(\rho_{A})$ and for the protocol to work independently of the
input state $\rho_{A}$. The first step is for Alice to perform the measurement
$\{E_{A^{\prime}A}^{g}\}_{g}$ on systems $A^{\prime}A$ and then send the
outcome $g$\ to Bob. Based on the outcome $g$, Bob then performs $V_{B}^{g}$.
The following analysis demonstrates that this protocol works, by simplifying
the form of the post-measurement state:%
\begin{align}
\left\vert G\right\vert \operatorname{Tr}_{AA^{\prime}}\{E_{A^{\prime}A}%
^{g}(\rho_{A^{\prime}}\otimes\omega_{AB})\}  &  =\left\vert A\right\vert
^{2}\operatorname{Tr}_{AA^{\prime}}\{U_{A^{\prime}}^{g}|\Phi\rangle
_{A^{\prime}A}\langle\Phi|_{A^{\prime}A}\left(  U_{A^{\prime}}^{g}\right)
^{\dag}(\rho_{A^{\prime}}\otimes\omega_{AB})\}\\
&  =\left\vert A\right\vert ^{2}\langle\Phi|_{A^{\prime}A}\left(
U_{A^{\prime}}^{g}\right)  ^{\dag}(\rho_{A^{\prime}}\otimes\omega
_{AB})U_{A^{\prime}}^{g}|\Phi\rangle_{A^{\prime}A}\\
&  =\left\vert A\right\vert ^{2}\langle\Phi|_{A^{\prime}A}\left(
U_{A^{\prime}}^{g}\right)  ^{\dag}\rho_{A^{\prime}}U_{A^{\prime}}^{g}%
\otimes\mathcal{N}_{A^{\prime\prime}\rightarrow B}(\Phi_{AA^{\prime\prime}%
}))|\Phi\rangle_{A^{\prime}A}\\
&  =\left\vert A\right\vert ^{2}\langle\Phi|_{A^{\prime}A}\left[  \left(
U_{A}^{g}\right)  ^{\dag}\rho_{A}U_{A}^{g}\right]  ^{\ast}\mathcal{N}%
_{A^{\prime\prime}\rightarrow B}\left(  \Phi_{AA^{\prime\prime}}\right)
|\Phi\rangle_{A^{\prime}A}. \label{eq:cov-tp-simul-block-1}%
\end{align}
The first three equalities follow by substitution and some rewriting. The
fourth equality follows from the fact that%
\begin{equation}
\langle\Phi|_{A^{\prime}A}M_{A^{\prime}}=\langle\Phi|_{A^{\prime}A}M_{A}%
^{\ast} \label{eq:ricochet-prop}%
\end{equation}
for any operator $M$ and where $*$ denotes the complex conjugate, taken with
respect to the basis in which $\vert\Phi\rangle_{A^{\prime}A}$ is defined.
Continuing, we have that%
\begin{align}
\eqref{eq:cov-tp-simul-block-1}  &  =\left\vert A\right\vert \operatorname{Tr}%
_{A}\left\{  \left[  \left(  U_{A}^{g}\right)  ^{\dag}\rho_{A}U_{A}%
^{g}\right]  ^{\ast}\mathcal{N}_{A^{\prime\prime}\rightarrow B}\left(
\Phi_{AA^{\prime\prime}}\right)  \right\} \\
&  =\left\vert A\right\vert \operatorname{Tr}_{A}\left\{  \mathcal{N}%
_{A^{\prime\prime}\rightarrow B}\left(  \left[  \left(  U_{A^{\prime\prime}%
}^{g}\right)  ^{\dag}\rho_{A^{\prime\prime}}U_{A^{\prime\prime}}^{g}\right]
^{\dag}\Phi_{AA^{\prime\prime}}\right)  \right\} \\
&  =\mathcal{N}_{A^{\prime\prime}\rightarrow B}\left(  \left[  \left(
U_{A^{\prime\prime}}^{g}\right)  ^{\dag}\rho_{A^{\prime\prime}}U_{A^{\prime
\prime}}^{g}\right]  ^{\dag}\right) \\
&  =\mathcal{N}_{A^{\prime\prime}\rightarrow B}\left(  \left(  U_{A^{\prime
\prime}}^{g}\right)  ^{\dag}\rho_{A^{\prime\prime}}U_{A^{\prime\prime}}%
^{g}\right) \\
&  =\left(  V_{B}^{g}\right)  ^{\dag}\mathcal{N}_{A^{\prime\prime}\rightarrow
B}\left(  \rho_{A^{\prime\prime}}\right)  V_{B}^{g}.
\end{align}
The first equality follows because $\left\vert A\right\vert \langle
\Phi|_{A^{\prime}A}\left(  I_{A^{\prime}}\otimes M_{AB}\right)  |\Phi
\rangle_{A^{\prime}A}=\operatorname{Tr}_{A}\{M_{AB}\}$ for any operator
$M_{AB}$. The second equality follows by applying the conjugate transpose of
\eqref{eq:ricochet-prop}. The final equality follows from the covariance
property of the channel.

Thus, if Bob finally performs the unitary $V_{B}^{g}$ upon receiving $g$ via a
classical channel from Alice, then the output of the protocol is
$\mathcal{N}_{A^{\prime\prime}\rightarrow B}\left(  \rho_{A^{\prime\prime}%
}\right)  $, so that this protocol simulates the action of the channel
$\mathcal{N}$ on the state $\rho$.

\section{Definitions of privacy and converse bounds}

\label{app:priv-def-str-conv}One of the main results of this appendix is to
show that a converse bound with the definition of privacy from
\eqref{eq:code-performance-fid} implies a converse bound for a quantum
generalization of the definition of privacy from \cite{HTW14}. Note that we
need to modify the definition in \eqref{eq:code-performance-fid} slightly as
given in Definition~\ref{def:privacy-marginal}\ below, but all of our converse
bounds apply for this modified definition of privacy.

We begin by recalling the two notions of privacy. Consider a tripartite state
$\rho_{KLE}$ of the following form:%
\begin{equation}
\rho_{KLE}=\frac{1}{\left\vert K\right\vert }\sum_{k,l}p(l|k)|k\rangle\langle
k|_{K}\otimes|l\rangle\langle l|_{L}\otimes\rho_{E}^{k,l},
\end{equation}
which is the kind of state that gets generated at the end of a secret-key
transmission protocol. Specifically, the classical variable in system $K$ is
generated uniformly at random, and at the end of the protocol the receiver
decodes it in the system $L$, which is intended to be one share of a secret
key correlated with system $K$. The system $E$ represents the eavesdropper's
system, which can be correlated with systems $K$ and $L$.

\begin{definition}
The state $\rho_{KLE}$ is an $(\varepsilon,\delta)$ Type~I\ secret-key state
for $\varepsilon,\delta\in\left[  0,1\right]  $ if%
\begin{align}
\Pr\left\{  K\neq L\right\}   &  \leq\varepsilon,\\
\frac{1}{2}\left\Vert \rho_{KE}-\pi_{K}\otimes\rho_{E}\right\Vert _{1}  &
\leq\delta,
\end{align}
where $\Pr\left\{  K\neq L\right\}  =\frac{1}{\left\vert K\right\vert }%
\sum_{l\neq k}p(l|k)$ and $\pi_{K}=I_{K}/\left\vert K\right\vert $ is the
maximally mixed state.
\end{definition}

\begin{definition}
\label{def:privacy-marginal}The state $\rho_{KLE}$ is an $\eta$ Type~II
secret-key state for $\eta\in\left[  0,1\right]  $ if%
\begin{equation}
F(\rho_{KLE},\overline{\Phi}_{KL}\otimes\rho_{E})=1-\eta,
\end{equation}
where $\overline{\Phi}_{KL}$ is the maximally classically correlated state,
defined as%
\begin{equation}
\overline{\Phi}_{KL}\equiv\frac{1}{\left\vert K\right\vert }\sum_{k}%
|k\rangle\langle k|_{K}\otimes|k\rangle\langle k|_{L}.
\end{equation}

\end{definition}

Observe that the main difference between
Definitions~\ref{def:tripartite-key-state-1-def} and
\ref{def:privacy-marginal} is that in the former, we allow for the ideal state
of the eavesdropper's system to be arbitrary, whereas in the latter, we demand
that the ideal state of the eavesdropper's system is equal to the marginal
$\rho_{E}$\ of $\rho_{KLE}$. Note that this constraint does not affect any of
our converse bounds in the main text, but here we show how it allows us to
connect to other notions of privacy in the literature.

\begin{proposition}
For the definitions given above, the following bound holds%
\begin{equation}
1-\sqrt{F(\rho_{KLE},\overline{\Phi}_{KL}\otimes\rho_{E})}\leq\Pr\{K\neq
L\}+\frac{1}{2}\left\Vert \rho_{KE}-\pi_{K}\otimes\rho_{E}\right\Vert _{1}.
\label{eq:str-conv-relation}%
\end{equation}
Thus, an $\eta$ Type~II\ secret-key state with $\eta\rightarrow1$ is an
$\left(  \varepsilon,\delta\right)  $ Type~I\ secret-key state with
$\varepsilon+\delta\rightarrow\xi\geq1$.
\end{proposition}

\begin{proof}
We follow the proof of \cite[Theorem~4.1]{PR14} closely. Let $\gamma_{KLE}$
denote the following state obtained by discarding the $L$ system of
$\rho_{KLE}$ and copying the contents of the $K$ system to the $L$ system:%
\begin{equation}
\gamma_{KLE}\equiv\frac{1}{\left\vert K\right\vert }\sum_{k,l}p(l|k)|k\rangle
\langle k|_{K}\otimes|k\rangle\langle k|_{L}\otimes\rho_{E}^{k,l}.
\end{equation}
Then by the triangle inequality consider that%
\begin{equation}
\left\Vert \rho_{KLE}-\overline{\Phi}_{KL}\otimes\rho_{E}\right\Vert _{1}%
\leq\left\Vert \rho_{KLE}-\gamma_{KLE}\right\Vert _{1}+\left\Vert \gamma
_{KLE}-\overline{\Phi}_{KL}\otimes\rho_{E}\right\Vert _{1}.
\end{equation}
The following holds%
\begin{align}
\left\Vert \gamma_{KLE}-\overline{\Phi}_{KL}\otimes\rho_{E}\right\Vert _{1}
&  =\left\Vert \gamma_{KE}\otimes|0\rangle\langle0|_{L}-\pi_{K}\otimes
|0\rangle\langle0|_{L}\otimes\rho_{E}\right\Vert _{1}\\
&  =\left\Vert \gamma_{KE}-\pi_{K}\otimes\rho_{E}\right\Vert _{1},
\end{align}
because we can perform an invertible controlled-subtraction from system $K$ to
system $L$, giving the first equality, and then we can discard the system $L$
because it does not change the trace distance. Now consider that%
\begin{align}
\left\Vert \rho_{KLE}-\gamma_{KLE}\right\Vert _{1}  &  =\left\Vert \frac
{1}{\left\vert K\right\vert }\sum_{k,l}p(l|k)|k\rangle\langle k|_{K}%
\otimes\left(  |l\rangle\langle l|_{L}-|k\rangle\langle k|_{L}\right)
\otimes\rho_{E}^{k,l}\right\Vert _{1}\\
&  \leq\frac{1}{\left\vert K\right\vert }\sum_{k,l}p(l|k)\left\Vert
|k\rangle\langle k|_{K}\otimes\left(  |l\rangle\langle l|_{L}-|k\rangle\langle
k|_{L}\right)  \otimes\rho_{E}^{k,l}\right\Vert _{1}\\
&  =\frac{1}{\left\vert K\right\vert }\sum_{k,l}p(l|k)\left\Vert
|l\rangle\langle l|_{L}-|k\rangle\langle k|_{L}\right\Vert _{1}\\
&  =\frac{2}{\left\vert K\right\vert }\sum_{k\neq l}p(l|k)=2\Pr\{K\neq L\},
\end{align}
where the inequality follows from convexity of the trace norm. Using a well
known relation between fidelity and trace distance \cite{FG98}\ and combining
with the above, we find that%
\begin{align}
1-\sqrt{F(\rho_{KLE},\overline{\Phi}_{KL}\otimes\rho_{E})}  &  \leq\frac{1}%
{2}\left\Vert \rho_{KLE}-\overline{\Phi}_{KL}\otimes\rho_{E}\right\Vert _{1}\\
&  \leq\Pr\{K\neq L\}+\frac{1}{2}\left\Vert \rho_{KE}-\pi_{K}\otimes\rho
_{E}\right\Vert _{1},
\end{align}
concluding the proof.
\end{proof}

As a consequence of the above proposition, if there is a sequence of private
communication protocols such that $\eta\rightarrow1$, then by the bound in
\eqref{eq:str-conv-relation}, we necessarily have that $\varepsilon
+\delta\rightarrow\xi\geq1$. So our approach in
Section~\ref{sec:strong-converse}\ gets strong converse rates for all
$\varepsilon$ and $\delta$ such that $\varepsilon+\delta<1$, which is the same
regime for which the authors of \cite{HTW14}\ were able to establish strong
converse rates for the classical wiretap channel.

We can also show the following alternate relation between the two notions of privacy:

\begin{proposition}
For the definitions given above, the following bounds hold%
\begin{align}
\Pr\left\{  K\neq L\right\}   &  \leq\sqrt{1-F(\rho_{KLE},\overline{\Phi}%
_{KL}\otimes\rho_{E})},\label{eq:err-correctness-to-privacy}\\
\frac{1}{2}\left\Vert \rho_{KE}-\pi_{K}\otimes\rho_{E}\right\Vert _{1}  &
\leq\sqrt{1-F(\rho_{KLE},\overline{\Phi}_{KL}\otimes\rho_{E})}.
\label{eq:security-to-privacy}%
\end{align}
Thus, an $\eta$ Type~II\ secret-key state is an $(\sqrt{\eta},\sqrt{\eta})$
Type~I\ secret-key state.
\end{proposition}

\begin{proof}
The inequality in \eqref{eq:err-correctness-to-privacy}\ follows because%
\begin{align}
\Pr\left\{  K\neq L\right\}   &  =\frac{1}{2}\left\Vert \rho_{KL}%
-\overline{\Phi}_{KL}\right\Vert _{1}\\
&  \leq\sqrt{1-F(\rho_{KL},\overline{\Phi}_{KL})}\\
&  \leq\sqrt{1-F(\rho_{KLE},\overline{\Phi}_{KL}\otimes\rho_{E})}.
\end{align}
In the above, the first equality is well known and straightforward to verify.
The first inequality is a consequence of a well known relation between
fidelity and trace distance \cite{FG98}. The second inequality follows from
the monotonicity of fidelity. The inequality in \eqref{eq:security-to-privacy}
follows for similar reasons, because%
\begin{align}
\frac{1}{2}\left\Vert \rho_{KE}-\pi_{K}\otimes\rho_{E}\right\Vert _{1}  &
\leq\sqrt{1-F(\rho_{KE},\overline{\Phi}_{K}\otimes\rho_{E})}\\
&  \leq\sqrt{1-F(\rho_{KLE},\overline{\Phi}_{KL}\otimes\rho_{E})},
\end{align}
which concludes the proof.
\end{proof}

\section{One-shot distillable entanglement lower bound}

\label{app:one-shot-ent-dist-lower}Here we give a proof of
Proposition~\ref{prop:one-shot-ent-dist-lower}, regarding a lower bound on
one-shot distillable entanglement. Let $\rho_{ABE}\in\mathcal{D}%
(\mathcal{H}_{ABE})$ purify the state $\rho_{AE}$. Let $\mathcal{T}%
_{A\rightarrow A_{1}X_{A}}$ be the following quantum channel corresponding to
a quantum instrument:%
\begin{equation}
\mathcal{T}_{A\rightarrow A_{1}X_{A}}(\cdot)\equiv\sum_{x}\mathcal{V}%
_{A\rightarrow A_{1}}^{x}\left(  P_{A}^{x}(\cdot)P_{A}^{x}\right)
\otimes|x\rangle\langle x|_{X_{A}},
\end{equation}
where $\{P_{A}^{x}\}$ is a set of projectors such that $\sum_{x}P_{A}%
^{x}=I_{A}$ and each $\mathcal{V}_{A\rightarrow A_{1}}^{x}$ is an isometric
channel that isometrically embeds the subspace onto which $P_{A}^{x}$ projects
into $\mathcal{H}_{A_{1}}$. For simplicity, we assume here and in what follows
that $\left\vert A_{1}\right\vert $ divides $\left\vert A\right\vert $.

In this paragraph, we prove a \textquotedblleft non-smooth\textquotedblright%
\ bound and later convert this to a smooth-entropy bound. The non-smooth
decoupling theorem of \cite[Theorem~3.3]{Dupuis2014} then states that%
\begin{equation}
\int_{\mathbb{U}(A)}dU\ \left\Vert \mathcal{T}_{A\rightarrow A_{1}X_{A}}%
(U_{A}\rho_{AE}U_{A}^{\dag})-\tau_{A_{1}X_{A}}\otimes\rho_{E}\right\Vert
_{1}\leq2^{-\frac{1}{2}H_{2}(A|E)_{\rho}}2^{-\frac{1}{2}H_{2}(A^{\prime}%
|A_{1}X_{A})_{\tau}},
\end{equation}
where $H_{2}(C|D)_{\sigma}$ denotes the conditional collision entropy of a
bipartite state $\sigma_{CD}\in\mathcal{D}(\mathcal{H}_{CD})$, defined as%
\begin{equation}
H_{2}(C|D)_{\sigma}\equiv\sup_{\omega_{D}\in\mathcal{D}(\mathcal{H}_{D})}%
-\log\operatorname{Tr}\!\left\{  \left(  \omega_{D}^{-1/4}\sigma_{CD}%
\omega_{D}^{-1/4}\right)  ^{2}\right\}  .
\end{equation}
In the above, $\rho_{AE}$ is the reduction of $\rho_{ABE}$, and $\tau
_{A_{1}X_{A}}$ is the reduction of the following state:%
\begin{equation}
\mathcal{T}_{A\rightarrow A_{1}X_{A}}(\Phi_{AA^{\prime}})\equiv\tau
_{A^{\prime}A_{1}X_{A}}.
\end{equation}
Note under our assumption that $\left\vert A_{1}\right\vert $ divides
$\left\vert A\right\vert $, we have that%
\begin{align}
H_{2}(A^{\prime}|A_{1}X_{A})_{\tau} &  =-\log\left\vert A_{1}\right\vert ,\\
\tau_{A_{1}X_{A}} &  =\frac{I_{A_{1}}}{\left\vert A_{1}\right\vert }%
\otimes\frac{I_{X_{A}}}{\left\vert X_{A}\right\vert }.
\end{align}
If we choose $\left\vert A_{1}\right\vert $ as%
\begin{equation}
\log\left\vert A_{1}\right\vert =H_{2}(A|E)_{\rho}-2\log\!\left(  \frac
{1}{\varepsilon}\right)  ,
\end{equation}
then it follows from the above that%
\begin{equation}
\int_{\mathbb{U}(A)}dU\ \left\Vert \mathcal{T}_{A\rightarrow A_{1}X_{A}}%
(U_{A}\rho_{AE}U_{A}^{\dag})-\tau_{A_{1}X_{A}}\otimes\rho_{E}\right\Vert
_{1}\leq\varepsilon.
\end{equation}

This implies that there exists a unitary $U_{A}$ such that%
\begin{equation}
P\!\left(  \mathcal{T}_{A\rightarrow A_{1}X_{A}}(U_{A}\rho_{AE}U_{A}^{\dag
}),\frac{I_{A_{1}}}{\left\vert A_{1}\right\vert }\otimes\frac{I_{X_{A}}%
}{\left\vert X_{A}\right\vert }\otimes\rho_{E}\right)  \leq\sqrt{\varepsilon
},\label{eq:decoupling-cond-1}%
\end{equation}
for%
\begin{equation}
\log\left\vert A_{1}\right\vert =-H_{\max}(A|B)_{\rho}-2\log\!\left(  \frac
{1}{\varepsilon}\right)  ,
\end{equation}
due to the fact that%
\begin{equation}
H_{2}(A|E)_{\rho}\geq H_{\min}(A|E)_{\rho}=-H_{\max}(A|B)_{\rho},
\end{equation}
where%
\begin{equation}
H_{\min}(A|E)_{\rho}\equiv\sup_{\omega_{E}\in\mathcal{D}(\mathcal{H}_{E})}%
\sup\left\{  \lambda\in\mathbb{R}:\rho_{AE}\leq2^{-\lambda}I_{A}\otimes
\omega_{E}\right\}  ,
\end{equation}
and the equality $H_{\min}(A|E)_{\rho}=-H_{\max}(A|B)_{\rho}$ follows from the
duality result in \cite{KRS09}. Let%
\begin{align}
\sigma_{A_{1}EX_{A}} &  \equiv\mathcal{T}_{A\rightarrow A_{1}X_{A}}(U_{A}%
\rho_{AE}U_{A}^{\dag})\\
&  \equiv\sum_{x}p_{X}(x)|x\rangle\langle x|_{X_{A}}\otimes\sigma_{A_{1}E}%
^{x}.
\end{align}
By rewriting \eqref{eq:decoupling-cond-1} as%
\begin{equation}
\left[  \sum_{x}\sqrt{p_{X}(x)\frac{1}{\left\vert X_{A}\right\vert }}\sqrt
{F}\!\left(  \sigma_{A_{1}E}^{x},\frac{I_{A_{1}}}{\left\vert A_{1}\right\vert
}\otimes\rho_{E}\right)  \right]  ^{2}\geq1-\varepsilon,
\end{equation}
and applying Uhlmann's theorem, we find that there exists a set $\{\mathcal{U}%
_{B\rightarrow B_{1}\overline{B}}^{x}\}$\ of isometric channels such that%
\begin{equation}
\left[  \sum_{x}\sqrt{p_{X}(x)\frac{1}{\left\vert X_{A}\right\vert }}\sqrt
{F}\!\left(  \mathcal{U}_{B\rightarrow B_{1}\overline{B}}^{x}(\sigma_{A_{1}%
BE}^{x}),\Phi_{A_{1}B_{1}}\otimes\psi_{\overline{B}E}\right)  \right]
^{2}\geq1-\varepsilon,\label{eq:uhlmann-app-decou-dist}%
\end{equation}
where $\sigma_{A_{1}BE}^{x}$ is a conditional state arising from
$\mathcal{T}_{A\rightarrow A_{1}X_{A}}(U_{A}\rho_{ABE}U_{A}^{\dag})$ and
$\psi_{\overline{B}E}$ purifies $\rho_{E}$. Using the monotonicity of fidelity
under partial trace and rewriting \eqref{eq:uhlmann-app-decou-dist} in terms
of purified distance, we can conclude that there exists a channel
$\Lambda_{BX_{B}\rightarrow B_{1}}$, where $X_{B}$ is a classical copy of
$X_{A}$ sent over a classical channel to Bob, such that%
\begin{equation}
P(\Lambda_{BX_{B}\rightarrow B_{1}}(\sigma_{A_{1}BX_{B}}),\Phi_{A_{1}B_{1}%
})\leq\sqrt{\varepsilon}.
\end{equation}

Now pick $\overline{\rho}_{AB}\in\mathcal{B}^{\sqrt{\varepsilon}-\eta}%
(\rho_{AB})$ such that $H_{\max}^{\sqrt{\varepsilon}-\eta}(A|B)_{\rho}%
=H_{\max}(A|B)_{\overline{\rho}}$ (the ball $\mathcal{B}^{\sqrt{\varepsilon
}-\eta}(\rho_{AB})$ of states around $\rho_{AB}$ is with respect to purified
distance). Let $\overline{\rho}_{ABE}$ purify $\overline{\rho}_{AB}$. Then by
the non-smooth bound above, we find that%
\begin{equation}
P\!\left(  \overline{\sigma}_{A_{1}EX_{A}},\frac{I_{A_{1}}}{\left\vert
A_{1}\right\vert }\otimes\overline{\omega}_{X_{A}E}\right)  \leq\eta,
\end{equation}
for%
\begin{equation}
\log\left\vert A_{1}\right\vert =-H_{\max}^{\sqrt{\varepsilon}-\eta
}(A|B)_{\rho}-4\log\!\left(  \frac{1}{\eta}\right)  ,
\end{equation}
and where%
\begin{align}
\overline{\sigma}_{A_{1}EX_{A}} &  \equiv\mathcal{T}_{A\rightarrow A_{1}X_{A}%
}(U_{A}\overline{\rho}_{AE}U_{A}^{\dag}),\\
\overline{\omega}_{X_{A}E} &  \equiv\frac{I_{X_{A}}}{\left\vert X_{A}%
\right\vert }\otimes\overline{\rho}_{E}.
\end{align}
Then%
\begin{align}
P\!\left(  \sigma_{A_{1}EX_{A}},\frac{I_{A_{1}}}{\left\vert A_{1}\right\vert
}\otimes\overline{\omega}_{X_{A}E}\right)   &  \leq P\!\left(  \mathcal{T}%
_{A\rightarrow A_{1}X_{A}}(U_{A}\rho_{AE}U_{A}^{\dag}),\mathcal{T}%
_{A\rightarrow A_{1}X_{A}}(U_{A}\overline{\rho}_{AE}U_{A}^{\dag})\right)
\nonumber\\
&  \qquad+P\!\left(  \mathcal{T}_{A\rightarrow A_{1}X_{A}}(U_{A}\overline
{\rho}_{AE}U_{A}^{\dag}),\frac{I_{A_{1}}}{\left\vert A_{1}\right\vert }%
\otimes\overline{\omega}_{X_{A}E}\right)  \\
&  \leq P(\rho_{AE},\overline{\rho}_{AE})+\eta\\
&  \leq\left(  \sqrt{\varepsilon}-\eta\right)  +\eta=\sqrt{\varepsilon}.
\end{align}
Applying Uhlmann's theorem once again as done above and converting purified
distance to fidelity concludes the proof.

\section{Variance of the relative entropy of entanglement for
phase-insensitive Gaussian channels}

\label{app:sep-states-variances}In this appendix, we detail the calculation of
the variance of the relative entropy of entanglement for the phase-insensitive
Gaussian channels given in
\eqref{eq:thermal-channel}--\eqref{eq:additive-noise-channel}. In particular,
we establish the formulas given in
\eqref{eq:unc-rel-ent-var-thermal}--\eqref{eq:unc-rel-ent-var-add}. To begin
with and as reviewed in \cite{PLOB15}, given a two-mode state with covariance
matrix in standard form as in \cite[Eq.~(E1)]{BLTW16}, it is a separable state
if%
\begin{equation}
c\leq c_{\text{sep}}\equiv\sqrt{\left(  a-1/2\right)  \left(  b-1/2\right)  }.
\end{equation}
For any given $a$ and $b$, the two-mode Gaussian state in standard form and
having maximal correlations between the two modes has $c$ set to
$c_{\text{sep}}$ (the amount of correlations is quantified by quantum discord
as done in \cite{PSBCL14}). This choice turns out to be a guiding principle
for selecting a separable state closest in \textquotedblleft relative entropy
distance\textquotedblright\ to the state at the output of a given channel.

We begin by reviewing the various pairs of states from \cite{PLOB15}\ for
comparison. The state that we consider at the input of any of the channels is
the two-mode squeezed vacuum, which has zero mean and standard-form covariance
matrix%
\begin{equation}
V_{\text{in}}^{\mu}\equiv%
\begin{bmatrix}
\mu & c\\
c & \mu
\end{bmatrix}
\oplus%
\begin{bmatrix}
\mu & -c\\
-c & \mu
\end{bmatrix}
,
\end{equation}
where $\mu\geq1/2$ is the energy of the reduced state on a single mode
(directly related to the amount of entanglement in the state) and
$c\equiv\sqrt{\mu^{2}-1/4}$. Sending one mode of this state through the three
channels of interest (thermal, amplifier, additive-noise) leads to two-mode
states $\rho_{\eta,\omega}^{\mu}$, $\rho_{G,\omega}^{\mu}$, and $\rho_{\xi
}^{\mu}$ with the following respective covariance matrices:%
\begin{align}
V_{\eta,\omega}^{\mu}  &  \equiv%
\begin{bmatrix}
\mu & \sqrt{\eta}c\\
\sqrt{\eta}c & \eta\mu+\left(  1-\eta\right)  \omega
\end{bmatrix}
\oplus%
\begin{bmatrix}
\mu & -\sqrt{\eta}c\\
-\sqrt{\eta}c & \eta\mu+\left(  1-\eta\right)  \omega
\end{bmatrix}
,\\
V_{G,\omega}^{\mu}  &  \equiv%
\begin{bmatrix}
\mu & \sqrt{G}c\\
\sqrt{G}c & G\mu+\left(  G-1\right)  \omega
\end{bmatrix}
\oplus%
\begin{bmatrix}
\mu & -\sqrt{G}c\\
-\sqrt{G}c & G\mu+\left(  G-1\right)  \omega
\end{bmatrix}
,\\
V_{\xi}^{\mu}  &  \equiv%
\begin{bmatrix}
\mu & c\\
c & \mu+\xi
\end{bmatrix}
\oplus%
\begin{bmatrix}
\mu & -c\\
-c & \mu+\xi
\end{bmatrix}
,
\end{align}
where $\omega\equiv N_{B}+1/2\geq1/2$. Using the aforementioned guiding
principle, the resulting separable states $\sigma_{\eta,\omega}^{\mu}$,
$\sigma_{G,\omega}^{\mu}$, and $\sigma_{\xi}^{\mu}$ for evaluating the bounds
have zero mean and the following respective covariance matrices:%
\begin{align}
V_{\eta,\omega}^{\mu,\text{sep}}  &  \equiv%
\begin{bmatrix}
\mu & \sqrt{\eta}c_{1}\\
\sqrt{\eta}c_{1} & \eta\mu+\left(  1-\eta\right)  \omega
\end{bmatrix}
\oplus%
\begin{bmatrix}
\mu & -\sqrt{\eta}c_{1}\\
-\sqrt{\eta}c_{1} & \eta\mu+\left(  1-\eta\right)  \omega
\end{bmatrix}
,\label{eq:sep-state-thermal}\\
V_{G,\omega}^{\mu,\text{sep}}  &  \equiv%
\begin{bmatrix}
\mu & \sqrt{G}c_{2}\\
\sqrt{G}c_{2} & G\mu+\left(  G-1\right)  \omega
\end{bmatrix}
\oplus%
\begin{bmatrix}
\mu & -\sqrt{G}c_{2}\\
-\sqrt{G}c_{2} & G\mu+\left(  G-1\right)  \omega
\end{bmatrix}
,\\
V_{\xi}^{\mu,\text{sep}}  &  \equiv%
\begin{bmatrix}
\mu & c_{3}\\
c_{3} & \mu+\xi
\end{bmatrix}
\oplus%
\begin{bmatrix}
\mu & -c_{3}\\
-c_{3} & \mu+\xi
\end{bmatrix}
,
\end{align}
where%
\begin{align}
c_{1}  &  \equiv\sqrt{\left(  \mu-1/2\right)  \left(  \eta\mu+\left(
1-\eta\right)  \omega-1/2\right)  },\\
c_{2}  &  \equiv\sqrt{\left(  \mu-1/2\right)  \left(  G\mu+\left(  G-1\right)
\omega-1/2\right)  },\\
c_{3}  &  \equiv\sqrt{\left(  \mu-1/2\right)  \left(  \mu+\xi-1/2\right)  }.
\end{align}

Using the formula from \cite[Lemma~3]{BLTW16} and relying on a computer
algebra package to handle tedious algebraic
manipulations,\footnote{Mathematica source files are included in the arXiv
posting of this paper.} we find the following expansions of the various
relative entropy variances about $\mu=\infty$:%
\begin{align}
V(\rho_{\eta,\omega}^{\mu}\Vert\sigma_{\eta,\omega}^{\mu})  &  =\left(
\omega^{2}-1/4\right)  \ln^{2}\!\left(  \eta\frac{2\omega+1}{2\omega
-1}\right)  +O\!\left(  \mu^{-1}\right)  ,\\
V(\rho_{G,\omega}^{\mu}\Vert\sigma_{G,\omega}^{\mu})  &  =\left(  \omega
^{2}-1/4\right)  \ln^{2}\!\left(  G^{-1}\frac{2\omega+1}{2\omega-1}\right)
+O\!\left(  \mu^{-1}\right)  ,\\
V(\rho_{\xi}^{\mu}\Vert\sigma_{\xi}^{\mu})  &  =\left(  1-\xi\right)
^{2}+O\!\left(  \mu^{-1}\right)  .
\end{align}
Note that the values $\eta\frac{2\omega+1}{2\omega-1}>1$, $G^{-1}\frac
{2\omega+1}{2\omega-1}>1$, and $\xi<1$ correspond to values for which the
channels are not entanglement-breaking. Also, when $\omega=1/2$, the first two
channels become the pure-loss channel and the quantum-limited amplifier
channel, and we find that $V(\rho_{\eta,\omega=0}^{\mu}\Vert\sigma
_{\eta,\omega=0}^{\mu})=V(\rho_{G,\omega=0}^{\mu}\Vert\sigma_{G,\omega=0}%
^{\mu})=O(\mu^{-1})$, so that in these cases the relative entropy variance
vanishes in the infinite-energy limit $\mu\rightarrow\infty$. Using the
equality $\omega=N_{B}+1/2$, we find that%
\begin{align}
\left(  \omega^{2}-1/4\right)  \ln^{2}\!\left(  \eta\frac{2\omega+1}%
{2\omega-1}\right)   &  =N_{B}(N_{B}+1)\ln^{2}\!\left(  \eta\frac{N_{B}%
+1}{N_{B}}\right)  ,\\
\left(  \omega^{2}-1/4\right)  \ln^{2}\!\left(  G^{-1}\frac{2\omega+1}%
{2\omega-1}\right)   &  =N_{B}(N_{B}+1)\ln^{2}\!\left(  G^{-1}\frac{N_{B}%
+1}{N_{B}}\right)  .
\end{align}

\section{Upper bound for the hypothesis testing relative entropy}

\label{app:CLT-chebyshev}This appendix provides some necessary details for the
bounds given in the proof of Theorem~\ref{thm:bosonic-bounds}. As these bounds
follow directly from the developments in \cite[Eq.~(6.5)]{JOPS12},
\cite[Section~3]{DPR15}, and \cite[Lemma~15]{P10}, we point to these works for
the necessary background. Below we begin with an important proposition from
\cite[Eq.~(6.5)]{JOPS12} and then show how it leads to the desired bound in
\eqref{eq:chebyshev-like-bound}. For convenience, in this appendix we take
$D(\rho\Vert\sigma)$, $V(\rho\Vert\sigma)$, and $D_{H}^{\varepsilon}(\rho
\Vert\sigma)$ to be defined with respect to the natural logarithm.

The following proposition is available as \cite[Eq.~(6.5)]{JOPS12} and
restated as \cite[Corollary~2]{DPR15}:

\begin{proposition}
[{\cite[Eq.~(6.5)]{JOPS12}}]\label{prop:jaksic}Let $\rho$ and $\sigma$ be
faithful states acting on a separable Hilbert space $\mathcal{H}$, let $T$ be
a measurement operator acting on $\mathcal{H}$ and such that $0\leq T\leq I$,
and let $v,\theta\in\mathbb{R}$. Then%
\begin{equation}
e^{-\theta}\operatorname{Tr}\{(I-T)\rho\}+\operatorname{Tr}\{T\sigma
\}\geq\frac{e^{-\theta}}{1+e^{v-\theta}}\Pr\{X\leq v\},
\end{equation}
where $X$ is a random variable with mean $D(\rho\Vert\sigma)$ and variance
$V(\rho\Vert\sigma)$.
\end{proposition}

The following proposition follows from combining Proposition~\ref{prop:jaksic}
with the reasoning used to establish \cite[Lemma~15]{P10}.

\begin{proposition}
Let $\rho$ and $\sigma$ be faithful states acting on a separable Hilbert space
$\mathcal{H}$. The following Chebyshev-like bound holds for all $\varepsilon
\in(0,1)$ and all $n\geq1$:%
\begin{equation}
\frac{1}{n}D_{H}^{\varepsilon}(\rho^{\otimes n}\Vert\sigma^{\otimes n})\leq
D(\rho\Vert\sigma)+\sqrt{\frac{2V(\rho\Vert\sigma)}{n\left(  1-\varepsilon
\right)  }}+\frac{C(\varepsilon)}{n}, \label{eq:chebyshev-like-bound}%
\end{equation}
where $C(\varepsilon)\equiv\ln6+2\ln\left(  \left[  1+\varepsilon\right]
/\left[  1-\varepsilon\right]  \right)  $.
\end{proposition}

\begin{proof}
Let $T^{n}$ be any test satisfying $\operatorname{Tr}\{\left(  I^{\otimes
n}-T^{n}\right)  \rho^{\otimes n}\}\leq\varepsilon$. By applying the above
proposition (making the replacements $\rho\rightarrow\rho^{\otimes n}$ and
$\sigma\rightarrow\sigma^{\otimes n}$, so that $X_{n}$ is a random variable
with mean $nD(\rho\Vert\sigma)$ and variance $nV(\rho\Vert\sigma)$), we find
that%
\begin{align}
\operatorname{Tr}\{T^{n}\sigma^{\otimes n}\}  &  \geq e^{-\theta_{n}}\left(
\frac{\Pr\{X_{n}\leq v_{n}\}}{1+e^{v_{n}-\theta_{n}}}-\operatorname{Tr}%
\{(I-T^{n})\rho^{\otimes n}\}\right) \\
&  \geq e^{-\theta_{n}}\left(  \frac{\Pr\{X_{n}\leq v_{n}\}}{1+e^{v_{n}%
-\theta_{n}}}-\varepsilon\right)  . \label{eq:start-point-CLT-bnd}%
\end{align}
Setting $v_{n}=nD(\rho\Vert\sigma)+\sqrt{\frac{2nV(\rho\Vert\sigma
)}{1-\varepsilon}}$, we find as a consequence of the Chebyshev inequality that%
\begin{align}
\Pr\left\{  X_{n}>v_{n}\right\}   &  =\Pr\left\{  X_{n}-nD(\rho\Vert
\sigma)>\sqrt{\frac{2nV(\rho\Vert\sigma)}{1-\varepsilon}}\right\} \\
&  =\Pr\left\{  \left[  X_{n}-nD(\rho\Vert\sigma)\right]  ^{2}>\frac
{2nV(\rho\Vert\sigma)}{1-\varepsilon}\right\} \\
&  <\frac{1-\varepsilon}{2},
\end{align}
implying that $\Pr\left\{  X_{n}\leq v_{n}\right\}  \geq\left(  1+\varepsilon
\right)  /2$. Substituting above and taking $\theta_{n}=v_{n}+C_{0}%
(\varepsilon)$ for $C_{0}(\varepsilon)$ a constant such that $\frac
{1+\varepsilon}{2\left(  1+e^{-C_{0}(\varepsilon)}\right)  }-\varepsilon>0$,
we find that%
\begin{align}
\operatorname{Tr}\{T^{n}\sigma^{\otimes n}\}  &  \geq e^{-\theta_{n}}\left(
\frac{1+\varepsilon}{2\left(  1+e^{-C_{0}(\varepsilon)}\right)  }%
-\varepsilon\right) \\
&  =\exp\left\{  -\left[  nD(\rho\Vert\sigma)+\sqrt{\frac{2nV(\rho\Vert
\sigma)}{1-\varepsilon}}+C_{0}(\varepsilon)\right]  \right\}  \left(
\frac{1+\varepsilon}{2\left(  1+e^{-C_{0}(\varepsilon)}\right)  }%
-\varepsilon\right)  .
\end{align}
Since this holds for every test $T^{n}$ satisfying $\operatorname{Tr}\{\left(
I^{\otimes n}-T^{n}\right)  \rho^{\otimes n}\}\leq\varepsilon$, we can apply a
negative logarithm and divide by $n$ to conclude that%
\begin{equation}
\frac{1}{n}D_{H}^{\varepsilon}(\rho^{\otimes n}\Vert\sigma^{\otimes n})\leq
D(\rho\Vert\sigma)+\sqrt{\frac{2V(\rho\Vert\sigma)}{n\left(  1-\varepsilon
\right)  }}+\frac{C_{0}(\varepsilon)}{n}-\frac{1}{n}\ln\!\left(
\frac{1+\varepsilon}{2\left(  1+e^{-C_{0}(\varepsilon)}\right)  }%
-\varepsilon\right)  .
\end{equation}
The above bound implies \eqref{eq:chebyshev-like-bound}.

To get the constant $C(\varepsilon)$ in the upper bound, consider that the
condition $\frac{1+\varepsilon}{2\left(  1+e^{-C_{0}(\varepsilon)}\right)
}-\varepsilon>0$ is equivalent to the condition $C_{0}(\varepsilon
)>\ln(2\varepsilon/[1-\varepsilon])$. So we can pick $C_{0}(\varepsilon
)=\ln(3\varepsilon/[1-\varepsilon])$ and this choice implies that%
\begin{equation}
-\ln\!\left(  \frac{1+\varepsilon}{2\left(  1+e^{-C(\varepsilon)}\right)
}-\varepsilon\right)  =\ln\!\left(  \frac{2\left(  1+2\varepsilon\right)
}{\varepsilon(1-\varepsilon)}\right)  .
\end{equation}
We then find that%
\begin{align}
C_{0}(\varepsilon)-\ln\!\left(  \frac{1+\varepsilon}{2\left(
1+e^{-C(\varepsilon)}\right)  }-\varepsilon\right)   &  =\ln6+\ln\!\left(
\frac{1+2\varepsilon}{\left(  1-\varepsilon\right)  ^{2}}\right) \\
&  \leq\ln6+2\ln\!\left(  \frac{1+\varepsilon}{1-\varepsilon}\right)
=C(\varepsilon).
\end{align}
Putting everything together, we get the bound stated in \eqref{eq:chebyshev-like-bound}.
\end{proof}

\bibliographystyle{alpha}
\bibliography{Ref}

\end{document}